\documentclass[a4paper,11pt]{article}
\pdfoutput=1

\usepackage{jheppub} 

\usepackage[T1]{fontenc} 
\usepackage[usenames,dvipsnames]{xcolor}
\definecolor{cadet}{rgb}{0.33, 0.41, 0.47}
\definecolor{darkgreen}{rgb}{0, 0.5, 0}
\usepackage{amsmath}
\usepackage{amsthm}
\usepackage{amssymb}
\usepackage{braket}
\usepackage{graphicx}
\usepackage[utf8]{inputenc}
\usepackage[T1]{fontenc}
\usepackage{lmodern}
\usepackage{tensor}
\usepackage{tikz-cd}
\usepackage{tikz}
\usetikzlibrary{calc}
\usetikzlibrary{patterns}
\usetikzlibrary{positioning}
\tikzset{>=stealth}
\usepackage{color}
\usepackage{colortbl,hhline}
\usepackage{cleveref}
\usepackage{enumitem}
\hypersetup{
    colorlinks,
    citecolor=black,
    filecolor=black,
    linkcolor=black,
    urlcolor=black
}
\usepackage{epic}
\usepackage{pict2e}
\usepackage{bbm}
\usepackage{dsfont}
\usepackage{subcaption}
\usepackage{dashrule}
\usepackage{changepage}
\usepackage{youngtab}
\usepackage{array}
\newcolumntype{L}[1]{>{\raggedright\let\newline\\\arraybackslash\hspace{0pt}}p{#1}}
\newcolumntype{C}[1]{>{\centering\let\newline\\\arraybackslash\hspace{0pt}}p{#1}}
\newcolumntype{R}[1]{>{\raggedleft\let\newline\\\arraybackslash\hspace{0pt}}p{#1}}

\definecolor{lightergray}{rgb}{.9,.9,.9}
\definecolor{lighterblue}{rgb}{.9,.9,1}
\definecolor{lightergreen}{rgb}{.9,1,.9}
\definecolor{lightgray}{rgb}{.75,.75,.75}

\renewcommand{\leq}{\leqslant}
\renewcommand{\geq}{\geqslant}
\renewcommand{\le}{\leqslant}
\renewcommand{\ge}{\geqslant}

\def\clap#1{\hbox to 0pt{\hss#1\hss}}
\def\mathllap{\mathpalette\mathllapinternal}

\def\mathllapinternal#1#2{%
\llap{$\mathsurround=0pt#1{#2}$}}

\makeatletter
\newcommand{\rref}[1]{\protected@edef\@currentlabel{#1}}
\makeatother

\DeclareMathAlphabet{\mathpzc}{OT1}{pzc}{m}{it}

\DeclareMathOperator{\ch}{ch}

\DeclareMathOperator{\End}{End}
\DeclareMathOperator{\Id}{{\mathds{1}}}

\DeclareMathOperator{\Ber}{Ber}
\DeclareMathOperator{\qBer}{qBer}

\DeclareMathOperator{\Nil}{Nil}
\DeclareMathOperator{\diag}{diag}
\DeclareMathOperator{\SW}{SW}
\newcommand{\SWe}{\SW} 

\theoremstyle{definition}

 \numberwithin{equation}{section}
 \numberwithin{mdef}{section}

\let\es\varnothing

\newcommand{\ii}{i}
\newcommand{\CCC}{c}

\newcommand{\CC}{\mathbb{C}}

\newcommand{\lP}{\mathcal{P}}

\newcommand{\lV}{\mathcal{V}}

\newcommand{\lU}{\mathcal{U}}

\newcommand{\sla}{\mathfrak{sl}}

\newcommand{\nl}{\newline\newline}

\newcommand{\wT}{\mathbb{T}}

\newcommand{\wQ}{\mathbb{Q}}

\newcommand{\mtwo}[4]{\left(\begin{matrix} #1&#2\\#3&#4\end{matrix}\right)}
\newcommand{\vtwo}[2]{\begin{pmatrix} #1\\#2\end{pmatrix}}

\newcommand{\CO}{{\mathcal{O}}}
\newcommand{\CI}{{\mathcal{I}}}
\newcommand{\CJ}{{\mathcal{J}}}
\newcommand{\CR}{{\mathcal{R}}}

\newcommand{\CE}{{\mathcal{E}}}

\newcommand{\Bad}{\ensuremath{\mathit{\Theta}_{\,\rm not}}}

\newcommand{\Y}{{\ensuremath{\rm Y}}}
\newcommand{\M}{{\ensuremath{\mathsf{M}}}}
\newcommand{\DD}{\mathcal{D}}
\newcommand{\SG}{\mathsf{S}}

\newcommand{\SU}{{\mathsf{SU}}}
\newcommand{\GL}{{\mathsf{GL}}}
\newcommand{\gl}{\mathfrak{gl}}
\newcommand{\gm}{{\mathsf{m}}}
\newcommand{\gn}{{\mathsf{n}}}
\newcommand{\gln}{{\gl_\gm}}
\newcommand{\glm}{{\gl_\gm}}
\newcommand{\glmn}{{\mathfrak{\gl}_{\gm|\gn}}}
\newcommand{\slmn}{{\mathfrak{sl}_{\gm|\gn}}}

\newcommand{\gls}[1]{{\mathfrak{\gl}_{#1}}}
\newcommand{\Ygln}{\Y(\gln)}
\newcommand{\Yglmn}{\Y(\glmn)}

\newcommand{\CB}{{\mathcal{B}}}
\newcommand{\CF}{{\mathcal{F}}}

\newcommand{\fs}{{\bar\es}}
\newcommand{\be}{\begin{eqnarray}}
\newcommand{\ee}{\end{eqnarray}}

\newcommand{\EE}{\mathsf{E}}
\newcommand{\Tau}{{\mathcal{T}}}

\theoremstyle{plain}

\newtheorem{theorem}{Theorem}[section]
\newtheorem{lemma}[theorem]{Lemma}
\newtheorem{proposition}[theorem]{Proposition}
\newtheorem{corollary}[theorem]{Corollary}
\newtheorem*{theorem*}{Theorem}

\newcommand{\twoone}{\begin{tikzpicture}[scale=.15,baseline=-.1cm]
\draw (1,1) |- (0,-1) |- (2,1) |-(0,0);
\end{tikzpicture}}

\newcommand{\sA}{\ensuremath{\alpha}} 
\newcommand{\sB}{\ensuremath{\beta}} 
\newcommand{\sC}{\ensuremath{\gamma}} 
\newcommand{\sD}{\ensuremath{\delta}} 

\newcommand{\gA}{{\ensuremath{\bar\sA}}} 
\newcommand{\gB}{{\ensuremath{\bar\sB}}} 
\newcommand{\gC}{{\ensuremath{\bar\sC}}}

\newcommand{\ABCD}[2]{{[#1]}_{#2}}
 \newcommand{\pa}{\ensuremath{\ell}}

 \newcommand{\inhom}[1][\pa]{{\ensuremath{\theta_{#1}}}}
 \newcommand{\binhom}[1][\pa]{{\ensuremath{\bar\theta_{#1}}}}

\newcommand{\sse}{\chi}
\newcommand{\bsse}{\bar\sse}
\newcommand{\se}[1]{\ensuremath{\sse_{#1}}}
\newcommand{\sepa}{\ensuremath{\sse_{\pa}}}
\newcommand{\bse}[1]{\ensuremath{\bsse_{#1}}}
\newcommand{\bsepa}{\ensuremath{\bsse_{\pa}}}
\newcommand{\Ccrit}{\ensuremath{\mathcal{C}_{\rm crit}}}
\newcommand{\Cdom}{{\mathcal{C}}}
\newcommand{\Se}{{\mathcal{X}}}
\newcommand{\Secrit}{\ensuremath{\Se_{\rm crit}}}
\newcommand{\Sedom}{\Se}

\newcommand{\WA}{{\mathcal{W}}}
\newcommand{\WAL}{{\WA_{\Lambda}}}
\newcommand{\cWAL}{{\check\WA_{\Lambda}}}
\newcommand{\BA}{{\mathcal{B}}}
\newcommand{\BAL}{{\BA_{\Lambda}}}
\newcommand{\VTw}{{V_{\Lambda}}}
\newcommand{\VTl}{{V_{\Lambda}^+}}
\newcommand{\VTo}{{\mathsf{w}}}
\newcommand{\VV}{{U_{\Lambda}}}

\newcommand{\lVTwS}{{\mathcal{V}_{\Lambda}^{\SG}}}
\newcommand{\lVTlS}{{\mathcal{V}_{\Lambda}^{\SG+}}}

\newcommand{\lVVS}{{\mathcal{U}_{\Lambda}^{\SG}}}

\newcommand{\LD}{{\Lambda^+}}
\newcommand{\LW}{{\Lambda}}
\newcommand{\IPart}{{\underline\lambda}}

\newcommand{\N}[3]{{N_{#2,#3}^{(#1)}}}
\newcommand{\SYT}{\Tau}
\newcommand{\ao}{{\bar a}}
\newcommand{\so}{{\bar s}}
\newcommand{\bu}{{u}}
\newcommand{\vve}{{\bf e}}

\newcommand{\bas}{{\mathfrak{b}}}

\def\clap#1{\hbox to 0pt{\hss#1\hss}}
\def\mathllap{\mathpalette\mathllapinternal}

\def\mathllapinternal#1#2{%
\llap{$\mathsurround=0pt#1{#2}$}}

\newcommand{\eg}{{\it e.g. }}
\newcommand{\ie}{{\it i.e. }}
\newcommand{\cf}{{\it cf. }}
\newcommand{\etc}{{\it etc}}
\newcommand{\via}{{\it via }}
\newcommand{\rhs}{{r.h.s. }}
\newcommand{\lhs}{{l.h.s. }}
\newcommand{\wrt}{{w.r.t. }}

    \usepackage{etoolbox}%
    \makeatletter
    \patchcmd{\maketitle}{\@fpheader}{}{}{}
    \makeatother

\usepackage{mdframed} %
\usepackage{xparse} %
\definecolor{graylight}{cmyk}{.30,0,0,.67} %

\usepackage{hyperref}

\newmdenv[ %
  linecolor=graylight,
  topline=false,
  bottomline=false,
  rightline=false,
  skipabove=\topsep,
  skipbelow=\topsep
]{leftrule}

\NewDocumentEnvironment{example}{O{\textbf{Example:}}} %
{\begin{leftrule}\noindent\textcolor{graylight}{#1}\par}
{\end{leftrule}}    

\title{Completeness of Wronskian Bethe equations for rational $\glmn$ spin chains}

\author{Dmitry Chernyak$^{a,b}$}
\author{Sébastien Leurent$^{c}$}
\author{Dmytro Volin$^{a,d,e}$}

\affiliation[a]{School of Mathematics \& Hamilton Mathematics Institute,\\Trinity College Dublin, College Green, Dublin 2, Ireland}
\affiliation[b]{Département de Physique, École Normale Supérieure, Université PSL, 24 rue Lhomond, 75005 Paris, France}
\affiliation[c]{Université de Bourgogne-Franche-Comté, Institut de
  Mathématiques de Bourgogne, UMR 5584 du CNRS, 9 avenue Alain Savary,
21000 Dijon, France}
\affiliation[d]{Nordita, KTH Royal Institute of Technology and Stockholm University,\\
Roslagstullsbacken 23, SE-106 91 Stockholm, Sweden}
\affiliation[e]{Department of Physics and Astronomy,\\ Uppsala University, Box 516, SE-751 20 Uppsala, Sweden}

\emailAdd{dmitry.chernyak@ens.fr}
\emailAdd{sebastien.leurent@u-bourgogne.fr}
\emailAdd{dmytro.volin@physics.uu.se}

\abstract{We consider rational integrable supersymmetric $\glmn$ spin chains in the defining representation and prove the isomorphism between a commutative algebra of conserved charges (the Bethe algebra) and a polynomial ring (the Wronskian algebra) defined by functional relations between Baxter Q-functions that we call Wronskian Bethe equations. These equations, in contrast to standard nested Bethe equations, admit only physical solutions for any value of inhomogeneities and furthermore we prove that the algebraic number of solutions to these equations is equal to the dimension of the spin chain Hilbert space (modulo relevant symmetries). 

Both twisted and twist-less periodic boundary conditions are considered, the isomorphism statement uses, as a sufficient condition, that the spin chain inhomogeneities $\inhom$, $\pa=1,\ldots,L$ satisfy $\inhom+\hbar\neq\inhom[\pa']$ for $\pa<\pa'$. Counting of solutions is done in two independent ways: by computing a character of the Wronskian algebra and by explicitly solving the Bethe equations in certain scaling regimes supplemented with a proof that the algebraic number of solutions is the same for any value of $\theta_\pa$. In particular, we consider the regime $\theta_{\pa+1}/\theta_{\pa}\gg 1$ for the twist-less chain where we succeed to provide explicit solutions and their systematic labelling with standard Young tableaux.
 }

\notoc

\begin{document} 
\maketitle
\newpage
\section*{Notations}
{
\begin{minipage}{0.5\textwidth}
\begin{tabular}{R{0.15\textwidth}L{0.85\textwidth}}
\multicolumn{2}{l}{Typical values of indices}
\\
$a,b$ &  $1$ to $\gm$
\\
$i,j$ &  $\hat 1$ to $\hat \gn$ (hat is omitted sometimes)
\\
$\sA,\sB$ & from the set $\{1,\ldots,\gm,\hat 1,\ldots,\hat\gn\}$
\\
\pa& $1$ to $L$
\\
\hphantom{$\pa$} & {}
\end{tabular}
\end{minipage}
\begin{minipage}{0.5\textwidth}
\begin{tabular}{R{0.15\textwidth}L{0.85\textwidth}}
\multicolumn{2}{l}{Parameters}
\\
$z_\sA$& twist eigenvalues, $z_a\equiv x_a$, $z_{\hat i}\equiv y_i$
\\
$\inhom$& inhomogeneities (as variables)
\\
$\binhom$& inhomogeneities (fixed number)
\\
$\sepa$ & elementary symmetric polynomials
\\
$\bsepa$& $=\sepa(\binhom[1],\ldots,\binhom[L])$
\end{tabular}

\end{minipage}
\newline
\newline
\newline
\begin{tabular}{R{0.30\textwidth}L{0.70\textwidth}}
\multicolumn{2}{l}{Lie algebra}
\\
$\glmn$ & symmetry of the system (broken to Cartan in the twisted case)
\\
$\EE_{\sA\sB}$ & abstract generators and defining representation
\\
$\CE_{\sA\sB}$ & global spin chain action
\\
$\LD=(\IPart_1,\IPart_2,\ldots)$ & Young diagram $\equiv$ integer partition (typically of $L$)
\\
$(\IPart'_1,\IPart'_2,\ldots)$ & transposed partition, $h_{\LD}:=\IPart'_1$.
\\
$\LW=[\lambda_1,\ldots,\lambda_\gm|\nu_1,\ldots,\nu_\gn]$ & fundamental weight (eigenvalues of $\CE_{\alpha\alpha}$)
\\
$(\hat\lambda_1\ldots,\hat\lambda_{\gm'}|\hat\nu_1,\ldots,\hat\nu_{\gn'})$ & shifted weight (describes $\LD$ with marked point)
\end{tabular}
\nl
\newline
\begin{tabular}{R{0.15\textwidth}L{0.85\textwidth}}
\multicolumn{2}{l}{Spin chain}
\\
$V$ & Hilbert space of the spin chain ($\simeq (\CC^{\gm|\gn})^{\otimes L}$)
\\
$\VTw$ & subspace of $V$ spanned by states of weight $\LW$
\\
$\VTl$ & subspace of $V$ spanned by highest weight states of irreps $\LD$
\\
$\VV$ & either $\VTw$ or $\VTl$
\\
$d_\Lambda$ & dimension of $\VV$
\\
\end{tabular}
\nl
\newline
\begin{tabular}{R{0.15\textwidth}L{0.85\textwidth}}
\multicolumn{2}{l}{Bethe and Wronskian algebras}
\\
$\hat c_k^{(d)}$, $\hat c_\pa$ & operators acting on spin chain, coefficients in Baxter Q-operators, \eg $Q_k=u^{M_k}(1+\frac{\hat c_{k}^{(1)}}{u}+\ldots)$
\\
$c_k^{(d)}$, $c_\pa$ & abstract variables and/or eigenvalues of $\hat c_k^{(d)}$, $\hat c_\pa$
\\
$\BAL$ & Bethe algebra restricted to $\VV$ (generated by $\hat c_\pa$), a $\CC[\sse]$-module
\\
$\BAL(\binhom[])$ &  specialised Bethe algebra (for spin chain representation at point $\binhom[]$)
\\
$\BAL(\bsse)$ &specialised Bethe algebra (for symmetrised representation at point $\bsse$)
\\
$\WAL$ & Wronskian algebra (generated by $c_\pa$ subject to Wronskian Bethe equations)
\\
$\WAL(\bsse)$ & specialised Wronskian algebra
\end{tabular}
\nl
\newline
\begin{tabular}{R{0.15\textwidth}L{0.85\textwidth}}
\multicolumn{2}{l}{Functional relations conventions}
\\
$u$ & spectral parameter
\\
$\hbar$ & Unit of discrete shift in \eg Baxter equation, typically $\hbar=\pm \ii,\pm1,\pm 2$
\\
$f^{[n]}$ & $f^{[n]}\equiv f(u+\frac{\hbar}{2}n)$, $f^\pm\equiv f^{[\pm 1]}$
\\
$f\propto g$& \parbox[t]{30em}{$f$ and $g$, as functions of $u$, are equal up to a normalisation}
\end{tabular}
}
\tableofcontents

\section{Introduction}
Rational integrable spin chains are one of the first quantum integrable systems that were discovered and studied. In fact, their simplest $\SU(2)$ representative was introduced and solved, by means of coordinate Bethe Ansatz, in the seminal paper of Hans Bethe \cite{Bethe:1931hc}. 

In this article we consider periodic integrable spin chains of length $L$  constructed using the $\glmn$-invariant rational R-matrix, and with spin chain nodes being in fundamental (defining) representation of $\glmn$. The parameters defining the model are the twist matrix $G$ and inhomogeneities $\inhom[1],\ldots, \inhom[L]$. We cover the cases when $G$ is either equal to the identity (twist-less case) or is diagonalisable with distinct eigenvalues $x_1,\ldots,x_{\gm},y_1,\ldots,y_{\gn}$ (generic twisted case).

Spectrum of the commuting charges that form the so-called Bethe algebra $\BA$ can be encoded into rational symmetric combinations of the Bethe roots $\bu_{k}^{(\sA)}$. Equations defining the values of $\bu_{k}^{(\sA)}$ shall be called Bethe equations, and their most known presentation is given by nested Bethe Ansatz equations (NBAE)  which is the following relation between fractions \cite{Sutherland:1975vr,Kulish:1979cr,Kulish:1985bj,Ogievetsky:1986hu}
\be
\label{eq:Bethedistinguished}
\prod_{\pa=1}^L\frac{u_k^{(\sA)}-\inhom[\pa]+\frac{c_{1,2}+c_{1,1}}{2}\,\hbar\,\delta_{\sA,1}}{u_k^{(\sA)}-\inhom[\pa]-\frac{c_{1,2}+c_{1,1}}{2}\,\hbar\,\delta_{\sA,1}}
=\frac{{z}_{\sA+1}}{{z}_{\sA}}
\prod_{\substack{
    1\le\sB\le \gm+\gn-1\\
    1\le l\le M_{\sB}\\
    (\sB,l)\neq (\sA,k)}}
\frac{u_{k}^{(\sA)}-u_{l}^{(\sB)}+\frac \hbar 2c_{\sA,\sB}}{u_{k}^{(\sA)}-u_{l}^{(\sB)}-\frac \hbar 2 c_{\sA,\sB}}\,.
\ee
Here $\sA\in\{1,2,\dots,\gm+\gn-1\}$ and all $k\in\{1,\dots,M_{\sA}\}$, we denote $z_\sA=x_\sA$ for $1\leq \sA\leq \gm$, and $z_\sA=y_{\sA-\gm}$ for $\gm+1\leq \sA\leq \gm+\gn$, and $\hbar$ is a non-zero complex number (typical choices are $\ii,1,2$). Finally $c_{\sA,\sB}$ is the Cartan matrix of the $\slmn$ subalgebra of $\glmn$. It is equal \eg to $\left(\begin{smallmatrix} 2& -1 & 0 \\ -1 & 0 & 1 \\ 0 & 1 & -2  \end{smallmatrix}\right)$ for $\sla(2|2)$, the expression for other ranks should be obvious from this example. This expression of the Cartan matrix is written in the so-called distinguished grading of $\glmn$ but other gradings are also possible \cite{Tsuboi:1998ne,Ragoucy:2007kg}, the corresponding equations are obtained \via duality transformations, and we briefly mention them in Section~\ref{sec:NBAE}.
\newline
\newline
Obvious questions arising are whether each solution of the Bethe equations describes some physical state, which we call the {\it faithfulness} property, and whether all physical states can be described in this way, which we  call the {\it completeness} property. In particular, one often asks whether the number of solutions to the Bethe equations is equal to the dimension of the Hilbert space, probably after some obvious symmetries are factored out. In the literature, these properties are typically covered by the name of completeness, however the precise meaning of the word varies.

Quite surprisingly, despite the fundamental nature of these questions, they were properly resolved only in 2009 for the $\gl_2$ case and in 2013 for the $\gl_\gm$ case by Mukhin, Tarasov and Varchenko \cite{Mukhin2009, MTV}. Completeness and faithfulness were also recently proven for $\gl_{1|1}$ by Huang, Lu, and Mukhin \cite{Huang_2019, lu2019supersymmetric}. Proving completeness and faithfulness for an arbitrary rank $\glmn$ case is the subject of the current paper. For formal proofs, we build on ideas of \cite{MTV} and add several new insights, even for the bosonic $\gl_\gm$ subcase, to achieve the result. Besides formal proofs, we also give a recipe to explicitly label solutions.
\newline
\newline
Counting of solutions was first time addressed already in \cite{Bethe:1931hc} using the so-called string hypothesis, and later on this approach was extended to $\gl_\gm$ \cite{Kirillov1987}, $\gl_{2|1}$ \cite{Foerster:1992uk} and $\gl_{2|2}$ \cite{Schoutens:1993bf} cases. Further study of combinatorics implied by string hypothesis for $\glm$ spin chains led to the formulation of the Kerov-Kirillov-Reshetikhin bijection \cite{Kerov1988,Kirillov1988} between rigged configurations of Bethe strings and (in case of spin chains in the defining representation) standard Young tableaux. Although counting assuming string hypothesis leads to correct numbers, the hypothesis is strictly speaking wrong as one can show by a more detailed analysis and explicit counter-examples, see \eg \cite{VLADIMIROV1984418,Isler:1993fc}. Hence this approach, after all, did not accomplish its original thought application -- proving the completeness of rational Bethe equations. Instead, ideas of \cite{Kerov1988,Kirillov1988} became extremely fruitful and were further generalised in various applications of algebraic combinatorics, in particular in the context of ``combinatorial'' integrability that can be viewed as the crystal limit $q\to 0$ of $q$-deformed (XXZ-type) spin chains, see \eg \cite{Kuniba:2001xx,Hatayama:2001gy} and references therein.

Apart from the combinatorial challenge, analysing solutions of \eqref{eq:Bethedistinguished} has clear technical complications. First, solutions with coinciding Bethe roots generically do not correspond to physical states and then they should be discarded. However, there are cases when such solutions should be kept \cite{Avdeev1986,Volin:2010xz,Hao:2013rza}. Second, the so-called exceptional solutions with $u=u'\pm \hbar$ and/or $u=\inhom[]\pm\frac{\hbar}{2}$ (case {\it a} in \cite{Avdeev1986}) render relation \eqref{eq:Bethedistinguished} singular. Some of the exceptional solutions are physical and some of them are not, and, for instance, their behaviour upon change of twist or inhomogeneities  can decide for their physicality. Tracing this behaviour becomes a burden, especially at higher rank. To our knowledge, only the homogeneous $\gl_2$ case was properly understood \cite{Nepomechie:2013mua}. At higher ranks, non-physicality can be also hidden in non-physical exceptional solutions of dual Bethe equations even if the original Bethe equations appear as being free from any singularities \cite{MarboeUnpub}.

It is then not surprising that one should look for a different set of equations instead of \eqref{eq:Bethedistinguished} to prove completeness \cite{Baxter:2001sx}, and it is indeed the case for the proof in \cite{MTV} where a very elegant Wronskian condition was used. Define the finite-difference Wronskian between any number of $k$ functions as 
\be
W(F_1,\ldots,F_k)\equiv \det\limits_{1\leq   i,j\leq k} F_i(u+\hbar(\frac{k+1}{2}-j)),
\ee
 introduce $\gm$ monic polynomials~\footnote{We use notations suited for our generalisation to supersymmetric case. They are different from those in \cite{MTV}.} $q_{a|\es}=u^{M_{a|\es}}\left(1+\sum\limits_{k=1}^{M_{a|\es}}\frac{c_{a|\es}^{(k)}}{u^{k}}\right)$, $a=1,\ldots,\gm$ of degree $M_{a|\es}$. Then the eigenstates of the Bethe algebra of the $\glm$ spin chain are in a one-to-one correspondence with the solutions of
\be
\label{eq:WronskianBethe}
\frac{W(x_1^{u/\hbar}q_{1|\es},\ldots , x_\gm^{u/\hbar} q_{\gm|\es})}{W(x_1^{u/\hbar},\ldots,x_{\gm}^{u/\hbar})}=\prod_{\pa=1}^L(u-\inhom[\pa])\,
\ee
which should be considered as equations on the coefficients $c_{a|\es}^{(k)}$. As formulated, the statement  holds for the case when all $x_{a}$ are pairwise distinct, and it also applies to the twist-less case after arrangements discussed in Section~\ref{sec:parameterisation}.

We shall call \eqref{eq:WronskianBethe} Wronskian Bethe equations (WBE). They are equivalent to NBAE (with coinciding Bethe roots solutions being discarded) for generic values of $\inhom$  but, in contrast to \eqref{eq:Bethedistinguished}, smoothly work at {\it any} values of inhomogeneities, and this includes an important physical case of homogeneous spin chain with all $\inhom[\pa]=0$. Relation to Bethe roots of \eqref{eq:Bethedistinguished} is given by $\frac{W(x_1^{u/\hbar}q_{1|\es},\ldots , x_a^{u/\hbar} q_{a|\es})}{W(x_1^{u/\hbar},\ldots,x_{a}^{u/\hbar})}=\prod\limits_{k=1}^{M_a}(u-\bu_k^{(a)})$, $M_a=\sum\limits_{b=1}^a M_{b|\es}$.
\newline
\newline
WBE are natural in the logic of the analytic Bethe Ansatz \cite{Baxter:1971cs,Reshetikhin1983} (though in early works in this formalism their significance was not recognised), while NBAE  are often associated with the (nested) coordinate or algebraic Bethe Ansatz \cite{GaudinBook, Faddeev:1996iy} (though they are derived \via analytic Bethe Ansatz as well \cite{Tsuboi:1997iq,Arnaudon:2004vd}). The details of the completeness and faithfulness questions depend on the chosen approach. 

The nested coordinate/algebraic Bethe Ansatz is indeed an ansatz to build an eigenfunction and Bethe equations appear as the necessary consistency conditions for the ansatz to succeed.  Faithfulness is then the question of sufficiency of these conditions, while completeness is the question whether there exist eigenfunctions not described by this Ansatz.  In physics literature, faithfullness is often considered as granted reducing completeness to counting the number of solutions to the Bethe equations. The generic position faithfulness indeed follows rather straightforwardly from the ansatz itself, but a full systematic proof covering all exceptional situations would be much harder to achieve. To our knowledge, such a proof using a direct algebraic Bethe Ansatz framework was done only very recently for the $\gl_2$ homogeneous case \cite{Granet:2019knz}. In the context of separation of variables, a new type of ansatz to build eigenstates in higher-rank systems emerged \cite{Gromov:2016itr, Liashyk:2018qfc, Ryan:2018fyo}. Its faithfulness for $\glm$ chains is proven in the case of non-degenerate twist and assuming $\inhom[\pa]-\inhom[\pa']\neq \hbar\mathbb{Z}$.

The analytic Bethe Ansatz is actually not an ansatz to build wave functions. It often starts by considering the Bethe algebra -- a set of commuting operators which satisfy various functional relations as functions of the spectral parameter $u$ that were extensively studied \cite{Krichever:1996qd,Bazhanov:1998dq,Pronko:1998xa,Tsuboi:2009ud,Kazakov:2015efa}. WBE is one of (or a consequence of) these relations with $Q_{a|\es}:=x_a^{u/\hbar} q_{a|\es}$ being the renowned Q-operators, the first example of such an operator is due to Baxter \cite{Baxter:1972hz,Baxter:1982zz}. In the analytic Bethe Ansatz approach, faithfulness is non-trivial to demonstrate even in generic position, and it is an important part of our paper to prove it. On the other hand, one is certain that each of the Q-operators eigenvalues satisfy WBE. Hence the question of completeness becomes equivalent to the question of whether the Bethe algebra generated by Q-operators contains the full set of commuting charges, \ie whether the eigenvalues of Q-operators are sufficient to fully parametrise the Hilbert space. We  shall resolve this question positively by explicitly counting the number of solutions of WBE. Hence we still face the question of counting as in the algebraic/coordinate Bethe Ansatz scenario, however the statement that is proven as a consequence of counting is a bit different.
\newline
\newline
In physical applications, inhomogeneities  are often set to $\inhom=0$. However, keeping inhomogeneities as parameters that we are going to vary is decisive for the approach discussed in this paper~\footnote{The twist values $x_a,y_i$ are always kept fixed however.}. To start with, inhomogeneities are regulators that  put our system to a generic position. One can even explore regimes where Bethe equations can be solved explicitly which provides a constructive way to count solutions. In the twisted case, such a regime is $\left|\frac{\inhom[\pa+1]-\inhom[\pa]}{\hbar}\right|\gg 1$ when we  label solutions using binomial expansions, and in the twist-less case such a regime is $\left|\frac{\inhom[\pa+1]}{\inhom[\pa]}\right|\gg 1$ when we can label solutions using standard Young tableaux (SYT), a result first time stated in \cite{Kirillov1986}. Moreover if $\inhom$ are positive, we will demonstrate that there is an unambiguous analytic continuation from this regime to the point $\inhom=0$.

Furthermore, we show that the Bethe algebra can be generated by only $L$ generators. There are also $L$ inhomogeneities which allows us to prove generic position faithfulness statements using a rigorous version of a ``number of variables equals number of parameters'' argument.

Finally, it can be demonstrated that all properties can be made {\it polynomial} in $\sepa$ -- elementary symmetric polynomials in inhomogeneities. Algebraically, this shall be formalised by proving that certain properly designed objects are free $\CC[\se{1},\ldots,\se{L}]$-modules. This implies that general position completeness and faithfulness statements can be specialised to {\it any} numerical value of $\sepa$. 
\newline
\newline
The paper is organised as follows. 

In Section~\ref{sec:Prelim}, we recall all the necessary known results about the Yangian, the Bethe algebra, and Q-operators. The fact that Q-operators belong to the Bethe algebra on the level of representation is proven in Appendix~\ref{sec:q-belongs-bethe}. We conclude the section with the formulation of the supersymmetric twisted and twist-less versions of WBE \eqref{eq:WronskianBethe}. 

Sections~\ref{sec:C} and \ref{sec:F} provide proofs of completeness and faithfulness.

Section~\ref{sec:C} proves completeness. We introduce the concept of Wronskian algebra (a polynomial ring defined by WBE), prove that it is a free $\CC[\sse]$-module,  and explicitly count its rank using Hilbert series (a.k.a. character, index, partition function) confirming that it coincides with the dimension of the (corresponding) Hilbert space. By a standard argument, the rank of the Wronskian algebra is the number of solutions to WBE counted with multiplicities. The proof of freeness essentially uses the so-called properness of WBE which is proven in Appendix~\ref{sec:Finiteness} and Appendix~\ref{sec:deta-constr-mapp}.

Section~\ref{sec:F} proves faithfulness. More accurately, faithfulness means that the Bethe algebra is a faithful representation of the Wronskian algebra, which will immediately imply that these two algebras are isomorphic. The proof is first done for generic values of inhomogeneities (over the polynomial ring $\CC[\sse]$) and then is specialised to numerical values of $\sepa$. Important results allowing one to specialise  at any numerical value are covered in Appendix~\ref{sec:cyclic} which builds substantially on the approach of \cite{MTV}.

Sections~\ref{sec:variouspar} and \ref{sec:label} aim to make the obtained results more practical. 

Section~\ref{sec:variouspar} discusses various ways to parametrise the Bethe algebra in the twist-less case. In particular, we demonstrate that all the restricted Bethe algebras $\BAL$ for $\glmn$ with fixed Young diagram $\LD$ and different $\gm,\gn$ are isomorphic to one-another and to the Q-system on this diagram. We also explain how this formalism is mapped to NBAE.

Section~\ref{sec:label} considers regimes $\left|\frac{\inhom[\pa+1]-\inhom[\pa]}{\hbar}\right|\gg 1$, $\left|\frac{\inhom[\pa+1]}{\inhom[\pa]}\right|\gg 1$ and shows how to explicitly find solutions of WBE in these regimes. Some technical questions are postponed to Appendix~\ref{sec:enum-solut-with}.

In Section~\ref{sec:summary}, we summarise the results and then discuss their immediate applications. This includes an algorithm to solve Bethe equations (including at $\inhom=0$) with solutions being labelled with SYT, and applications to the Gaudin model and to the separation of variables program. We conclude the section with a review of a relation between the restricted Bethe algebras and certain quantum cohomology rings.

The paper uses substantially results and terminology from algebraic geometry and commutative algebra while the target audience includes researchers with no appropriate background. To alleviate the issue, we illustrate the discussion with numerous examples, including a comprehensive case study in Appendix~\ref{sec:OE}, and supplement the paper with Appendix~\ref{sec:pedestrian} containing mostly textbook material applied to the concrete problem that we consider. Section~\ref{sec:algedes} also summarises textbook knowledge about multiplicity of solutions but we decided to keep it in the main text  given its importance for the  paper.

\section{Definitions and basic properties}
\label{sec:Prelim}
As often happens in mathematical physics, it will be useful to recast a physical question into a problem in representation theory. A spin chain should be considered as a representation of the $\glmn$ Yangian. Commuting Hamiltonians belong to its certain commutative subalgebra known as \cite{Nazarov_1996} the Bethe algebra $\CB$, and the completeness question is closely related to the explicit realisation of this algebra by Baxter polynomials subjected to constraints. 

This section collects definitions of the above-mentioned objects. For the supersymmetric Yangian, our sign conventions are the same as in \cite{Nazarov1991,Gow_thesis}, however we define quantum Berezinian with a different overall shift of the spectral parameter. Many hidden subtleties in a consistent usage of signs are nicely reviewed in \cite{Maillet:2019ayx}.

\subsection{\texorpdfstring{$\glmn$}{gl(m|n)} Lie superalgebra, shifted and fundamental weights}
\label{sec:glmn}
Let us recall some essential facts about the $\glmn$ Lie superalgebra \cite{KAC19778}. Assign parity $\bar a=0$ for any ``bosonic'' index $a\in\{1,\ldots,\gm\}$  and parity $\bar i=1$ for any ``fermionic'' index $i\in\{\hat 1,\ldots,\hat\gn\}$. The $\glmn$ algebra is spanned by the generators $\EE_{\sA\sB}$, $\sA,\sB\in\{1,\ldots,\gm,\hat 1,\ldots,\hat\gn\}$ whose graded Lie bracket is 
\be
\left[\EE_{\sA\sB},\EE_{\sA'\sB'}\right\}=\delta_{\sB\sA'}\EE_{\sA\sB'}-(-1)^{(\gA+\gB)(\gA'+\gB')}\delta_{\sB'\sA}\EE_{\sA'\sB}\,.
\ee

The Hilbert space~\footnote{Usage of terminology ``Hilbert space'' is customary for quantum systems, however we do not use any scalar products in this work, except in Section~\ref{sec:alternative}.} comprising states of the spin chain of length $L$ is $V:= (\CC^{\gm|\gn})^{\otimes L}$. Each spin chain site $\CC^{\gm|\gn}$ transforms under the defining (vector) representation of $\glmn$. The global action of $\glmn$ on $V$ shall be denoted as $\CE_{\sA\sB}$. It is induced from the single site action by using the standard graded product rule, \eg if $\EE_{\sA\sB}\, \vve_{\sC}=\delta_{\sB\sC}\,\vve_{\sA}$ then, for $L=2$ 
\be
\CE_{\sA\sB}\, \vve_{\sC}\otimes\vve_{\sC'}=\delta_{\sB\sC}\,\vve_{\sA}\otimes\vve_{\sC'}+(-1)^{(\bar\sA+\bar\sB)\bar\sC}\delta_{\sB\sC'}\,\vve_{\sC}\otimes\vve_{\sA}.
\ee 
Overall, the sign rule $(A\otimes B)(C\otimes D)=(-1)^{\bar B \bar C}(AC)\otimes(BD)$ for tensor products of graded algebras and modules shall be always assumed. 
\nl
When describing spin chains with periodic boundary conditions ({\it twist-less} case), $\glmn$ is the symmetry of the system in the sense that it commutes with the Hamiltonians we are interested to diagonalise. Correspondingly, we shall encounter covariant representations -- the irreps appearing in the tensor powers of the defining $\glmn$ representation. These representations are of the highest-weight type, with highest-weight vector $v$ being defined by the property $\CE_{\sA\sB}\,v=0$ for $\sA<\sB$. 

The highest-weight property depends on a choice of the total order $<$ in the set $\{1,\ldots,\gm,\hat 1,\ldots,\hat\gn\}$, and different choices related to permutation of bosonic and fermionic indices lead to non-equivalent parameterisations of the system. Most of the properties that we shall discuss do not depend on such an order, and so we will often use an invariant description  based on labelling of the irreps with Young diagrams which is possible due to the supersymmetric version of the Schur-Weyl duality \cite{Sergeev1985,BERELE1987118}.

 \begin{figure}[t]
\begin{center}
\begin{picture}(220,160)(0,-140)
\thinlines
\color{gray}
\drawline(0,20)(0,-138)
\drawline(80,-100)(80,-138)
\drawline(0,20)(210,20)
\drawline(80,-100)(210,-100)
\color{blue!5}
\polygon*(0,20)(0,-100)(20,-100)(20,-80)(40,-80)(40,-40)(80,-40)(100,-40)(100,-20)(120,-20)(120,0)(160,0)(160,20)(0,20)
\color{gray}
\thinlines
\multiput(0,0)(0,-20){6}{\line(1,0){190}}
\multiput(0,20)(20,0){5}{\line(0,-1){150}}
\multiput(40,20)(20,0){8}{\line(0,-1){120}}
\multiput(0,-60)(0,-20){3}{\line(1,0){40}}

\color{black}
\thinlines
\put(0,0){
\dashline{4}(10,-78)(10,-95)
\put(10,-95){\vector(0,-1){5}}
\dashline{4}(10,-62)(10,-40)
\put(10,-45){\vector(0,1){05}}
\put(6,-73){$\hat\nu_1$}
\put(30,-75){\vector(0,-1){5}}
\put(30,-65){\vector(0,1){05}}
\put(26,-73){$\hat\nu_2$}
}
\dashline{4}(44,10)(-20,10)
\put(-15,10){\vector(-1,0){5}}                                                                                                                                                                                                                                                                                                                                                                                                                                                                                                                                                                                                                                                                                                                                                                                                                                                                                                                                                                                                                                                                                                                                                                                                                                                                                                                                                                                                                                                                                                                                                                                                                                                                                                                                                                                                                                                                                                                                                                                                                                                                                                                                                                                                                                                                                                                                                                                                                                                                                                                                                                                                                                                                                                                                                                                                                                                                                                                                                                                                                                                                                                                                                                                                                                                                                                                                                                                                                                                                                                                                                                                                                                                              \dashline{4}(58,10)(160,10)
\put(155,10){\vector(1,0){5}}
\put(45,7){$\hat\lambda_1$}
\dashline{4}(42,-10)(0,-10)
\put(5,-10){\vector(-1,0){5}}
\dashline{4}(57,-10)(120,-10)
\put(115,-10){\vector(1,0){5}}
\put(45,-13){$\hat\lambda_2$}
\dashline{4}(42,-30)(20,-30)
\put(25,-30){\vector(-1,0){5}}
\dashline{4}(57,-30)(100,-30)
\put(95,-30){\vector(1,0){5}}
\put(45,-33){$\hat\lambda_3$}
\put(47.5,-50){\vector(-1,0){7.5}}
\put(47,-53){$\hat\lambda_4\!=0$}

\color{BurntOrange}
\thinlines
\put(0,-20){
\put(35,-115){\vector(-1,0){35}}
\put(45,-115){\vector(1,0){35}}
\put(37.5,-117){$\gn$}
}
\put(-4,0){
\put(200,-35){\vector(0,1){55}}
\put(200,-45){\vector(0,-1){55}}
\put(196,-42){$\gm$}
}

\color{cadet}
\thicklines
\drawline(40,-60)(40,-40)(20,-40)(20,-20)(0,-20)(0,0)(-20,0)(-20,20)(-40,20)
\drawline(40,-60)(20,-60)(20,-40)(0,-40)(0,-20)(-20,-20)(-20,0)(-40,0)(-40,20)
\color{blue}
\thicklines
\drawline(0,-138)(0,-100)(20,-100)(20,-80)(40,-80)(40,-40)(80,-40)(100,-40)(100,-20)(120,-20)(120,0)(160,0)(160,20)(210,20)

\color{Black}
\put(0,-120){\circle*{3}}
\put(0,-100){\circle*{3}}
\put(20,-100){\circle*{3}}
\put(20,-80){\circle*{3}}
\put(40,-80){\circle*{3}}
\color{Red}
\put(40,-60){\circle*{3}}
\color{Black}
\put(40,-40){\circle*{3}}
\put(60,-40){\circle*{3}}
\put(80,-40){\circle*{3}}
\put(100,-40){\circle*{3}}
\put(100,-20){\circle*{3}}
\put(120,-20){\circle*{3}}
\put(120,0){\circle*{3}}
\put(140,0){\circle*{3}}
\put(160,0){\circle*{3}}
\put(160,20){\circle*{3}}
\put(180,20){\circle*{3}}
\put(200,20){\circle*{3}}

\color{cadet}
\put(-20,-20){
\put(100,-80){\vector(-1,1){40}}
\put(82,-58){\rotatebox{-45}{r}}
}
\end{picture}

\caption{\label{fig:ShortLong}\scriptsize One-to-one correspondence between shifted weights $(9,6,4,0|3,2)$ and a Young diagram with a marked point (red dot). Other points on the boundary (black dots) and hence other sets of shifted weights can be chosen. The outlined option is diagonally shifted from the corner of the $\gm|\gn$-hook that is linked to the $\glmn$ representation theory \cite{KAC19778}: the Young diagrams that fit into the hook exactly describe finite-dimensional irreps. The Young diagrams that touch the hook corner correspond to the so-called long (typical) irreps, and the diagrams that do not touch the corner correspond to the so-called short (atypical) irreps.}
\end{center}
\end{figure}
Young diagrams $\LD$ with a marked point $(\gm',\gn')$ on the boundary are in bijection with tuples of shifted~\footnote{by the Weyl vector and with the additional shift by $-1$ of $\hat\nu_i$ to get a symmetric description} weights $(\hat\lambda_1,\ldots,\hat\lambda_{\gm'}|\hat\nu_1,\ldots,\hat\nu_{\gn'})$. We choose the marked point to be $\gm'=\gm-r,\gn'=\gn-r$ for $r\in\mathbb{Z}_{\geq 0}$, the role of $r$ is to reduce diagonally the rank of the $\glmn$ algebra such that the inner corner of the fat hook attains the Young diagram boundary, see Figure~\ref{fig:ShortLong}. The explicit relation between the shifted weights and the shape of $\LD$ is
\be
\begin{array}{rccl}
\hat\lambda_a&=\IPart_a-a-\gn+\gm\,,&\quad&a=1,2,\ldots,\gm-r\,,
\\
\hat\nu_i&=\IPart'_i-i-\gm+\gn\,,&\quad&a=1,2,\ldots,\gn-r\,,
\end{array}
\ee
where $(\IPart_1,\IPart_2,\ldots)$ is the integer partition forming the shape $\LD$, $(\IPart'_1,\IPart'_2,\ldots)$ is the integer partition of the transposed diagram, and  $r=\min\limits_k(k|\IPart_{\gm-k}+k-\gn\geq 0)$.

Most of the results {\it do not depend} on the choice of the marked point, this is demonstrated in Section~\ref{sec:variouspar}. We made the choice of the diagonal reduction only for easier connections with the results already known in the literature.
\nl
For spin chains with twisted boundary conditions, and for generic diagonal twist, only the Cartan subalgebra of $\glmn$ is the symmetry of the system and states are then described using the fundamental weight. We define it as the tuple $\Lambda=[\lambda_1,\ldots,\lambda_\gm|\nu_1,\ldots,\nu_\gn]$, where $\CE_{aa}\,v= \lambda_{aa} v$, $\CE_{ii}\,v=\nu_i v\,.$
\nl
Dictated by the symmetry of the problem, we  introduce restrictions of the Hilbert space to the weight subspaces
\be
\label{eq:23}
V\supset\VTw \supset \VTl\,.
\ee
$\VTw$ is defined as the space of all vectors with fundamental weight $\Lambda$. Its dimension is given by the multinomial coefficient  ${L!}/({\prod_{1\le a\le \gm} \lambda_a!\prod_{1\le i\le \gn} \nu_i!})$. $\VTl$ is defined as the space of the highest-weight vectors for all irreps with Young diagram $\LD$ inside $V$. Its dimension is equal to the number of standard Young tableaux of shape $\LD$. For concreteness we choose the standard order $1<2<\ldots<\gm<\hat 1<\ldots<\hat \gn$ in which case the fundamental weight $\Lambda$ of the highest-weight vectors of the irrep $\LD$ is given by the rule
\be
\lambda_a=\IPart_a\,,\quad \nu_i={\rm max}(0,{\IPart'_i}-\gm)\,.
\ee
As is explained on page~\pageref{p65}, any other choice of the total order would lead to an isomorphic description and to the same conclusions albeit explicit realisation of $\VTl$ and certain related objects will be modified.

We shall use the notation $\VV$ to denote either $\VTl$ or $\VTw$ when discussion equally applies to both subspaces $\VTl$ and $\VTw$.
\subsection{Yangian}
\label{sec:Yangian}
The Yangian $\Yglmn$ is a quasi-triangular Hopf algebra with rational R-matrix. We summarise below its properties which will be relevant for us, see \eg \cite{molev2007yangians,Gow_thesis} for a more detailed discussion. 

Let $\sA,\sB\in\{1,\ldots,\gm\}\cup\{\hat 1,\ldots, \hat \gn\}$ and $k\in\{1,2,\dots\}$. The Yangian's generators $t_{\sA\sB}^{(k)}$ are collected, \via formal series in  $\hbar/u$
\be
t_{\sA\sB}(u)=\delta_{\sA\sB}\Id+\frac{\hbar}{u}t_{\sA\sB}^{(1)}+\ldots
\ee
into ``monodromies'' $t_{\sA\sB}(u)$ whose parity is equal to $\gA+\gB$.

Quasi-triangularity is an RTT-type relation that reads in component form as
\be
\label{eq:defYangian}
[t_{\sA\sB}(u),t_{\sC\sD}(v)\}=\frac{\hbar(-1)^{\gA\gB+\gA\gC+\gB\gC}}{u-v}\left(t_{\sC\sB}(u)t_{\sA\sD}(v)-t_{\sC\sB}(v)t_{\sA\sD}(u)\right)\,.
\ee
From Hopf algebra structures, we will only need the co-product
$\Delta(t_{\sA\sB}(u))=\sum_{\sC} t_{\sA\sC}\otimes t_{\sC\sB}$.
\newline
\newline
To realise the Yangian representation on the spin chain, consider first the evaluation homomorphism
\be\label{eq:ev1}
ev_{\inhom}:t_{\sA\sB}(u)\mapsto \delta_{\sA\sB}\Id+\hbar\,\frac{(-1)^{\gA}\EE_{\sA\sB}}{u-\inhom}\,,
\ee
where $\EE_{\sA\sB}$ are the $\glmn$ generators in the defining representation. Then, for $\inhom[]:=(\inhom[1],\ldots,\inhom[L])$, combine $L$ such maps by repetitively using the co-product
\be\label{eq:ev2}
ev_{\inhom[]}:t_{\sA\sB}\mapsto ev_{\inhom[1]}\otimes\ldots
\otimes ev_{\inhom[L]}\left(\sum_{\sC_1,\ldots,\sC_{L-1}}t_{\sA\sC_1}\otimes t_{\sC_1\sC_2}\otimes\ldots\otimes t_{\sC_{L-1}\sB}\right)\,.
\ee
We shall call \eqref{eq:ev2} the {\it spin chain representation} of the Yangian. 

This representation contains, in the first non-trivial coefficient of the $\hbar/u$ expansion, the global action of $\glmn$ on the spin chain defined in Section~\ref{sec:glmn}:
\be
\label{eq:globalaction}
ev_{\inhom[]}(t_{\sA\sB}(u))=\delta_{\sA\sB}\Id+(-1)^{\gA}\frac{\hbar}{u}{\CE_{\sA\sB}}+\ldots\,.
\ee

For $T_{\sA\sB}\equiv Q_{\theta}(u)t_{\sA\sB}(u)$, where  $Q_{\theta}(u)=\prod\limits_{\pa=1}^L(u-\inhom[\pa])$, $ev_{{\theta}}(T_{\sA\sB})\equiv Q_{\theta}\,ev_{{\theta}}(t_{\sA\sB})$ are {\it polynomials} in $u$ of degree at most $L$. Note that $ev_{{\theta}}(T_{\sA\sB})$ are also polynomials in $\inhom$. Construction \eqref{eq:ev2} corresponds to the graphics commonly used to define the monodromy matrix of a spin chain from Lax operators: $ev_{\inhom[]}(T_{\sA\sB})=\!\!
\begin{tikzpicture}[baseline=-.1cm,scale=.5] 
\foreach\x/\i/\pos/\dist in {1/1/right/-.1cm,2/2/right/-.1cm,4/L/left/-.16cm} {\fill (\x,0) circle (4pt);
\draw (\x,.5) -- (\x,-.5) node [\pos=\dist] {\tiny $\inhom[\i]$};}
\draw (0.5,0)  node [below=-1pt] {\tiny \ensuremath{\alpha}} -- (2.5,0) (3.5,0) -- (4.5,0)node [below=-1.5pt] {\tiny \ensuremath{\beta}};
\draw [ dotted, thick] (2.6,0)--(3.4,0);
\end{tikzpicture}
$, and $\inhom$, $\pa=1,\ldots, L$ are commonly known as the spin chain {\it inhomogeneities}.

We will use $T_{\sA\sB}=Q_{\theta}\,t_{\sA\sB}$ to denote both the Yangian generators and their images $ev_{{\theta}}(T_{\sA\sB})$. The context shall make it clear which meaning is being used. Note that $ev_{{\theta}}$ is not a faithful map and hence not all algebra-level results subdue the representation-level properties.

For the discussion of this paper, it will be often important to consider $\inhom$ as unevaluated commuting variables. When it is the case, the image of the map $ev_{\inhom[]}$ are  endomorphisms with polynomial coefficients~\footnote{Here we allow freedom of speech and consider $ev_{\theta}(\Yglmn)$ in the sense of $ev_{\inhom[]}(T_{\sA\sB})$.}:
\be\label{eq:steps}
ev_{\theta}:\Yglmn\longrightarrow (\End(\mathbb{C}^{\gm|\gn}))^{\otimes L}\otimes\CC[\theta]\,,
\ee
where $\CC[\theta]\equiv \CC[\inhom[1],\inhom[2],\ldots,\inhom[L]]$ is the polynomial ring in variables $\inhom$. Such operators naturally act on $\lV:=V\otimes\CC[\theta]$. Such a description also appears in the context of a Hecke algebra, see page~\pageref{pdAHA}.

If we are interested in inhomogeneities having particular numerical values -- in which case we typically denote them by $\binhom$ -- then we get a representation of the Yangian in a more standard sense
\be\label{eq:steps2}
ev_{\binhom[]}:\Yglmn\longrightarrow (\End(\mathbb{C}^{\gm|\gn}))^{\otimes L}\,.
\ee
If we need to emphasise that \eqref{eq:steps2} but not \eqref{eq:steps} is being used, we shall refer to \eqref{eq:steps2} as the {\it spin chain representation at point} $\binhom[]\equiv (\binhom[1],\ldots,\binhom[L])$.

\subsection{Bethe algebra}
The below-defined Bethe algebra $\CB$ is a commutative subalgebra of $\Yglmn$ which depends on a constant $\GL(\gm|\gn)$ group matrix $G$ dubbed twist. We restrict ourselves to the case  when $G$ is diagonalisable and furthermore choose a reference frame that diagonalises $G$, so $\CB$ actually depends only on the eigenvalues~\footnote{Matrices of the $\GL(\gm|\gn)$ group have entries belonging to a Grassmann algebra, hence their eigenvalues are in principle not complex numbers. Our discussion will assume that twists are complex numbers nevertheless. One can then check that the results still hold for any twists of type $x_a=A_a+n_a$, $y_i=B_i+n_i$ where $A,B\in\mathbb{C}$ and $n$ - even nilpotent elements of the Grassmann algebra, assuming that $A_a,B_i$ are pairwise distinct.} $x_1,\ldots, x_\gm$, $y_1,\ldots, y_\gn$ of $G$. We will consider only two opposing cases: of generic twist, when all $x_a,y_i$ are distinct, and of no-twist when $G=\Id$. Considering intermediate cases is possible but combinatorially bulky~\footnote{Analytic structure of Q-functions for partially degenerate twists, which is an essential ingredient for the completeness statements, was explored in detail in \cite{Kazakov:2015efa}.}. 

The Bethe algebra $\CB$ is defined as the algebra that is polynomially generated by the transfer matrices $\wT_{\mu}$ in covariant representations of $\glmn$ labelled with integer partitions  or equivalently Young diagrams $\mu$. By ``polynomially generated'' we mean that elements of $\CB$ are finite degree polynomials in $\hat d_k$ --  coefficients of the ({\it a priori} formal) $\hbar/u$ expansion  $\wT_{\mu}=\chi_{\mu}(G) u^{L|\mu|}(\Id+\hat d_1\frac{\hbar}{u}+\ldots)$, where $\chi_{\mu}(G)$ is the character of $G$ in representation $\mu$. 

Transfer matrices $\wT_{\mu}$ can be constructed using fusion from $T_{\sA\sB}$ \cite{Kulish1986,Zabrodin:1996vm} and hence are defined on the level of the Yangian as well as its representation $ev_{\theta}$. When we descend to the representation level, $\wT_{\mu}(u)$ is a degree-$L|\mu|$ polynomial in $u$, so the  $\hbar/u$ expansion truncates.

Let $(1^a)$ denote the Young diagram consisting of one column of height $a$, and $(s)$ -- the Young diagram consisting of one row of width $s$. To avoid discussing fusion in detail, we note that the $\wT_{\mu}$, and hence the Bethe algebra, can be polynomially generated from $\wT_{(1^a)}$ or from $\wT_{(s)}$, $a,s=1,2,\ldots$ using the determinant Cherednik-Bazhanov-Reshetikhin (CBR) formula \cite{Cherednik1987,Bazhanov:1989yk,Kazakov:2007na}  -- the Yangian version of Jacobi-Trudi  identities for characters $\chi_{\mu}(G)$, while $\wT_{(1^a)}$ and  $\wT_{(s)}$ are compactly defined through monodromies $T_{\sA\sB}$ as \cite{Molev2009}
\begin{subequations}
\label{eq:genseries}
\begin{align}
\Ber\left[\Id-\,\DD\,T(u)G\,\DD\right] &= \sum_{a=0}^{\infty} (-1)^a\DD^a\,\wT_{(1^a)}(u)\,\DD^a\,,
\\
\frac 1{\Ber\left[\Id-\,\DD\,T(u)G\,\DD\right]} &= \sum_{s=0}^{\infty} \DD^s\,\wT_{(s)}(u)\,\DD^s\,,
\end{align}
\end{subequations}
where $\DD\equiv e^{-\frac 12\hbar\partial_u}$.

The \lhs of \eqref{eq:genseries} is defined as follows. For $\M=\Id-\DD\,T G\,\DD$, introduce the notation $\M=\mtwo{\ABCD{\M}A}{\ABCD{\M}B}{\ABCD{\M}C}{\ABCD{\M}D}$ and $\M^{-1}= \mtwo{\ABCD{\M^{-1}}A}{\ABCD{\M^{-1}}B}{\ABCD{\M^{-1}}C}{\ABCD{\M^{-1}}D}$, where $\ABCD{\M}A$ is the $\gm\times \gm$ block of $\M$, $\ABCD{\M^{-1}}D$ is the $\gn\times\gn$ block of $\M^{-1}$ {\it etc.} Then the Berezinian is defined as $\Ber(\M)=\det \ABCD{\M}A\det \ABCD{\M^{-1}}D$, where the determinant of the block ``$A$'' (resp. ``$D$'') is defined through a column-ordered (resp. line-ordered) expansion
{\it e.g.}  $\det \ABCD{\M}A=\epsilon_{a_1\ldots a_\gm}(\ABCD{\M}A)_{a_11}\ldots (\ABCD{\M}A)_{a_\gm \gm}$.  Although $\M$ is a matrix with non-commutative entries $\M_{\sA\sB}$, the entries satisfy the (supersymmetric version of) Manin relations $[\M_{\sA\sB},\M_{\sC\sD}\}=(-1)^{\gA\gB+\gA\gC+\gB\gC}[\M_{\sC\sB},\M_{\sA\sD}\}$ which ensure that the above-defined Berezinian can only change sign if columns/rows of $\M$ are permuted. From earlier works, we mention that Berezinians in the context of $\Yglmn$ were introduced in \cite{Nazarov1991}, and generalise similar constructions in the $\Ygln$ case \cite{Nazarov_1996, Talalaev:2004qi, Chervov2008}, see also \cite{molev2007yangians}.

The physical Hamiltonian of the system is an element of the Bethe algebra and it is usually chosen to be $\mathbb{H}=\partial_u\log\wT_{(1)}(u)|_{u=0}$.  The algebraic equivalent of the statement that the Bethe algebra contains all commuting charges is the statement that it is a maximal commutative subalgebra that contains $\mathbb{H}$. The question about maximality can be asked on the level of the Yangian algebra or of its spin chain representation. On the level of the algebra, in the bosonic $\Ygln$ case, polynomial combinations of $\wT_{(1^a)}$ indeed generate, for non-degenerate twist, a maximal commutative subalgebra of $\Ygln$ \cite{Nazarov_1996, molev2007yangians} but it seems an equivalent statement was not proved for the supersymmetric case. To our knowledge, a comprehensive study of the Bethe algebra on the Yangian algebra level is still lacking.

However, our goal is to describe the Bethe algebra represented on the spin chain in which case it can be understood much better. In particular, the quantum Berezinian defined
\be
\qBer\equiv \DD^{(\gn-\gm)}\Ber\left[\DD\,T(u)G\,\DD\right]\DD^{(\gn-\gm)}
\ee
and which is known to generate the center of the Yangian \cite{Nazarov1991,Gow_2007} can be expressed, at least on the level of representation, as a ratio of transfer matrices $\qBer\propto
\frac{\wT_{((\gn+1)^{\gm})}}{\wT_{(\gn^{\gm+1})}}$ \cite{Maillet:2019ayx} and hence, by Hamilton-Cayley, belongs to the Bethe algebra.

It will be one of our results that, under a mild sufficient assumption on $\inhom[]$ ($\inhom+\hbar\neq\inhom[\pa']$ for $\pa<\pa'$), the Bethe algebra is a maximal commutative subalgebra of $\End(V)$ -- the algebra of all linear transformations of the spin chain Hilbert space.  %

\subsection{Q-operators}
Although there are infinitely many $\wT_{(1^a)}$ and $\wT_{(s)}$  in the expansions \eqref{eq:genseries}, we need finitely many functions of $u$ to generate the Bethe algebra. This can be seen for instance by analysing the CBR formula. On the level of representation, probably the most economic way is to express transfer matrices through Baxter Q-operators that were explicitly constructed as operators acting on the supersymmetric spin chain in \cite{Belitsky:2006cp,Bazhanov:2008yc,Kazakov:2010iu,Frassek:2010ga,Tsuboi:2019vvv}\footnote{For constructions of Q-operators for bosonic spin chains, see eg \cite{Bazhanov:1998dq,Derkachov:2008aq,Bazhanov:2010ts,Bazhanov:2010jq,Frenkel:2013uda}. }. The Q-operators are not elements of the Yangian, but they do belong to the Bethe algebra in the representation $ev_{\theta}$, in particular they are matrices whose coefficients are polynomials in $\inhom$, see Appendix~\ref{sec:q-belongs-bethe}.

The Q-operators generate the Bethe algebra as follows~\footnote{Equation \eqref{eq:WronskianwT} expresses $\wT$  for the so-called rectangular representations, where the Young diagram $(s^a)\equiv (\smash{\underbrace{s,s,s,\dots,s}_{a\textrm{ times}}})$ is of rectangular shape; representations $(1^a)$, $(s)$ are special subcases. \linebreak[4]\rule{0pt}{0pt} \hspace{8em}\hfill  Generalisation to arbitrary representations is known \cite{Tsuboi:2011iz}.} \cite{Krichever:1996qd,Gromov:2010km,Tsuboi:2011iz}
\begin{align}\label{eq:WronskianwT}
  \wT_{(s^a)}&=u^{a\,s\,L} \chi_{a,s}(G)(\Id +\hat d_1\frac{\hbar}{u}+\ldots)\notag\\&\propto
 \frac 1{{Q_{\bar\emptyset|\bar\emptyset}^{[a-s]}}
}  \prod\limits_{k=1}^a\prod\limits_{l=1}^s \frac {{Q_{\bar\emptyset|\bar\emptyset}}^{[a+s+2-2k-2l]}}{(\Ber G)^{u/\hbar}}\times
  \begin{cases}
  \displaystyle
      {{\epsilon^{b_1\dots b_\gm}}{} Q_{b_1\dots b_a|\emptyset}^{[\gm-\gn+s]}Q_{b_{a+1}\dots b_\gm|\smash{\bar\emptyset}}^{[-s]}}     \,, & s\ge a-\gm+\gn\\
      \displaystyle
      {{\epsilon^{i_1\dots i_\gn}}{} Q_{\emptyset\vphantom{\bar\emptyset}|i_1\dots i_s}^{[\gm-\gn-a]}Q_{\bar\emptyset|i_{s+1}\dots i_{\gn}}^{[+a]}}
     \,,  & a\ge s+\gm-\gn
  \end{cases}\,,
\end{align}
where $\epsilon$ denotes the Levi-Cevita antisymmetric tensor, summation over repeated indices is performed, and the $\propto$ symbol involves a proportionality factor which is identified by imposing that the coefficient of  the highest degree of $\wT_{(s^a)}$ (as a polynomial in $u$) is the character $\chi_{(s^a)}(G)$.

\label{page:Qsys}
In total, there are $2^{\gm+\gn}$ Q-operators. They are  labelled as $Q_{A|I}$, where $A$ is a multi-index from $\{1,\ldots,\gm\}$ and $I$ is a multi-index from $\{1,\ldots,\gn\}$. $Q_{A|I}$ are anti-symmetric \wrt permutations in $A$ and $I$, and polynomial up to an exponential prefactor (as in \ref{eq:WronskianBethe}):
\begin{equation}
  \label{eq:Qq}
  Q_{A|I}\propto \frac{\prod_{a\in A} x_a^{u/\hbar}}{\prod_{i\in I} y_i^{u/\hbar}}\, q_{A|I}\,,
\end{equation}
where the proportionality factor~\footnote{One can set the proportionality factor to be equal to one when $|A|+|I|\le 1$, which fixes this factor for other values of $A$ and $I$ due to the relations \eqref{QQrelations}. It is explicitly spelled out in \eg \cite{Kazakov:2015efa}.} in ``$\propto$'' is fixed by the condition that each $q_{A|I}$ is a monic polynomial in the variable $u$.

The Q-operators satisfy the following QQ-relations
\begin{subequations}
\label{QQrelations}
\be
Q_{Aab|I}Q_{A|I}=W(Q_{Aa|I},Q_{Ab|I})\,,
\\
\label{QQferm} Q_{Aa|I}Q_{A|Ii}=W(Q_{Aa|Ii},Q_{A|I})\,,
\\
Q_{A|Iij}Q_{A|I}=W(Q_{A|Ii},Q_{A|Ij})\,.
\ee
\end{subequations}
Furthermore, the two operators $Q_{\es|\es}$ and $Q_{\fs|\fs}$ are central: on the one hand $Q_{\fs|\fs}$ is by construction a multiple of identity operator (hence a central element), and one can note that (\ref{eq:WronskianwT}) gives\footnote{There exists a different choice of normalizations that simplifies this expression to $\qBer=Q_{\fs|\fs}^-/Q_{\fs|\fs}^+$.}
\begin{align}
\qBer\propto
\frac{\wT_{((\gn+1)^{\gm})}}{\wT_{(\gn^{\gm+1})}}\propto
  \begin{cases}
\frac{Q_{\bar\emptyset|\bar\emptyset}^{[-\gm+\gn+1]}Q_{\bar\emptyset|\bar\emptyset}^{[-\gm+\gn+3]}\dots Q_{\bar\emptyset|\bar\emptyset}^{[+\gm-\gn+1]}}{(\Ber G)^{(\gm-\gn)u/\hbar}Q_{\bar\emptyset|\bar\emptyset}^{[\gm-\gn-1]}}&\textrm{ if }\gm\ge\gn\,,\\
\frac{(\Ber G)^{(\gn-\gm)u/\hbar} Q_{\bar\emptyset|\bar\emptyset}^{[\gm-\gn+1]}}{Q_{\bar\emptyset|\bar\emptyset}^{[-\gm+\gn-1]}Q_{\bar\emptyset|\bar\emptyset}^{[-\gm+\gn-3]}\dots Q_{\bar\emptyset|\bar\emptyset}^{[+\gm-\gn-1]}}&\textrm{ if }\gn\ge\gm\,.
  \end{cases}
\end{align}
The statement that $\qBer$ is in the center of the Yangian hence reflects the property that $Q_{\fs|\fs}$ itself is central.

On the other hand $Q_{\es|\es}=1$, which allows to write all $Q$-operators explicitly in terms of $Q_{a|\es}$, $Q_{\es|i}$, and $Q_{a|i}$ \cite{Tsuboi:2009ud,Kazakov:2015efa}:
\be\label{eq:qqforms}
Q_{(\mathsf{a}|\mathsf{a}+\mathsf{c})}\propto (Q_{(1|1)}^{[\mathsf{c}]})^{\mathsf{a}}Q_{(0|1)}^{[\mathsf{c}-1]}\ldots Q_{(0|1)}^{[1-\mathsf{c}]}\,,\quad Q_{(\mathsf{a}+\mathsf{c}|\mathsf{a})}\propto (Q_{(1|1)}^{[\mathsf{c}]})^{\mathsf{a}}Q_{(1|0)}^{[\mathsf{c}-1]}\ldots Q_{(1|0)}^{[1-\mathsf{c}]}\,,
\ee
For compactness, we used exterior forms $Q_{(\mathsf{a}|\mathsf{b})}=\frac 1{{\mathsf{a}}!\mathsf{b}!}Q_{a_1\ldots a_{\mathsf{a}}|i_1\ldots i_{\mathsf{b}}}\psi_{0}^{a_1}\ldots \psi_0^{a_{\mathsf{a}}}\psi_{1}^{i_1}\ldots\psi_{1}^{i_{\mathsf{b}}}$ with $\psi_0^a,\psi_1^{i}$ being auxiliary Grassmann variables.

We shall need the expression for $Q_{\fs|\fs}\equiv Q_{1\ldots \gm|1\ldots \gn}$ which explicitly is the following determinant
\begin{equation}
\label{susyw}
Q_{\fs|\fs}(u)=
(-1)^{\gn(\gm-\gn)}\begin{vmatrix}
Q_{1|1}^{[\gm-\gn]} & \cdots & Q_{1|\gn}^{[\gm-\gn]} & Q_{1|\es}^{[\gm-\gn-1]} & Q_{1|\es}^{[\gm-\gn-3]} & \cdots & Q_{1|\es}^{[-(\gm-\gn)+1]}\\
\vdots & ~ & \vdots & \vdots & \vdots & ~ & \vdots\\
Q_{\gm|1}^{[\gm-\gn]} & \cdots & Q_{\gm|\gn}^{[\gm-\gn]} & Q_{\gm|\es}^{[\gm-\gn-1]} & Q_{\gm|\es}^{[\gm-\gn-3]} & \cdots & Q_{\gm|\es}^{[-(\gm-\gn)+1]}
\end{vmatrix}
\,,
\end{equation}
for the case $\gm\geq \gn$; a similar expression can be written also for $\gm>\gn$.

One should also note that the coefficients $Q_{a|\es}$, $Q_{\es|i}$,
and $Q_{a|i}$ are related by \eqref{QQferm}:
\begin{equation}
Q_{a|i}^+-Q_{a|i}^-= Q_{a|\es}Q_{\es|i}\,.
\label{sport}
\end{equation}

\subsection{Quantisation condition (Wronskian Bethe equations)}
\label{sec:parameterisation}
The essential property for the description of the Bethe algebra is the explicit analytic structure of Q-operators (a.k.a. rational analytic Bethe Ansatz) which is known from results of \cite{Bazhanov:2008yc,Kazakov:2010iu,Bazhanov:2010jq}, and can be also derived using the logic of Appendix~\ref{sec:q-belongs-bethe}.

The definition of Q-operators depends on a gauge choice, see \eg \cite{Gromov:2014caa} for details. Below we write expressions in one particular gauge which is suitable for our goals~\footnote{Apart from the gauge choices, there also exists several other discrepancies in labelling conventions across the literature. First, an arrangement in the spectral parameter can be present. Second, the role of  the Q-operators and their Hodge duals $Q^{A|J}\propto \varepsilon^{AA'}\varepsilon^{JJ'}Q_{JJ'}$ can be swapped. Third, a permutation of indices $1\ldots \gm|1\ldots \gn$ can be used.}.

For the central element $Q_{\fs|\fs}$ it is possible to directly compute its explicit value  which in the gauge of our choice becomes
\be
\label{eq:Qfs}
Q_{\fs|\fs}(u)\propto  (\Ber G)^{u/\hbar} Q_{\theta}(u)\,.
\ee
This property is an important aspect of the Bethe algebra and it is essentially equivalent to the set of  Bethe equations as it will become clear below.

We set $\Ber G=1$ for convenience. It only affects the overall normalisation of transfer matrices and hence is inessential.
\newline
\newline
To write expressions for the other Q-operators, we need to restrict  the representation space to a certain subspace. The generic twist and twist-less cases should be treated separately.

\subsubsection*{Twisted case} 
The Cartan generators $\CE_{\sA\sA}$ of the global $\glmn$ action \eqref{eq:globalaction} commute with $\CB$ (and belong to it). To describe analytic properties of the Q-operators in a useful manner, we need to restrict to an eigenspace of $\CE_{\sA\sA}$ which is the weight space $\VTw$ defined after \eqref{eq:23}. The Bethe algebra restricted to this subspace shall be denoted as
\be
\BAL:=\BA|_{\VTw}\,.
\ee
Upon restriction to $\VTw$, the polynomial operators $q_{A|I}$ of (\ref{eq:Qq}) read \be
\label{qAJ}
q_{A|J}=u^{M_{A|J}}+\sum_{k=1}^{M_{A|J}}\hat c_{A|J}^{(k)}u^{M_{A|J}-k}\,.
\ee
These are monic polynomials of degree $M_{A|J}$ with operator-valued coefficients $\hat c_{A|J}^{(k)}$. The diagonalisation of $\hat c_{A|J}^{(k)}$ is the subject of a Bethe Ansatz.
The degree $M_{A|J}$ has fixed value on each $\VTw$, which can for instance be identified by explicit computations following Appendix~\ref{sec:q-belongs-bethe} (see \cite{Kazakov:2010iu}), they have the following expression in terms of the fundamental weight $\Lambda=[\lambda_1,\ldots,\lambda_{\gm}|\nu_1,\ldots,\nu_{\gn}]$ \begin{align}
\label{degree}
M_{A|J}&=\sum_{a\in A}\lambda_a+\sum_{j\in J}\nu_j\,,
\end{align}
and can be interpreted as the ``magnon'' numbers of the nested Bethe Ansätze.

For generic twist, \eqref{sport} is a non-degenerate system of linear equations in the coefficients $\hat c_{a|i}^{(k)}$ that fixes $\hat c_{a|i}^{(k)}$ and thus $Q_{a|i}$ uniquely. Hence all the Q-operators are generated by the single-index Q-functions $Q_{a|\es}$, $Q_{\es|i}$. There are precisely $L$ coefficients $\hat c_{a|\es}^{(k)}$ and $\hat c_{\es|i}^{(k)}$ as one can quickly conclude form \eqref{degree} and the invariant value of the total charge $\sum\limits_{a=1}^{\gm}\lambda_a+\sum\limits_{j=1}^{\gn}\nu_j=L\,.$

We can use $C_{\Lambda}$ -- the set of all coefficients $\hat c_{a|\es}^{(k)}$ and $\hat c_{\es|i}^{(k)}$ -- to polynomially generate $Q_{\fs|\fs}$ using \eqref{susyw}. This operation is a supersymmetric generalisation of the Wronskian determinant in \eqref{eq:WronskianBethe} and shall be denoted as $\SWe(C_{\Lambda})(u)$,
\begin{equation}
\label{mastereq}
\fbox{
$
\displaystyle
\SWe(C_{\Lambda})(u)=\prod_{\pa=1}^L(u-\inhom[\pa])
$   
}\,.
\end{equation}
Note that we chose the normalisation of $\SWe$ such that the leading $u^L$ term is monic.

We call \eqref{mastereq} the quantisation condition or the Wronskian Bethe equations (WBE). Its important feature is to provide exactly $L$ equations on $L$ variables (elements of the set $C_{\Lambda}$) and to contain $L$ free parameters $\inhom$. We shall denote this system of $L$ equations as
\be
\label{mastereq2}
\SWe_\pa(c) =\sepa\,,\quad \pa=1,2,\ldots,L\,,
\ee
where $\sepa$ are elementary symmetric polynomials of $\inhom[1],\inhom[2],\ldots,\inhom[L]$. Dependence on inhomogeneities only through their symmetric combinations $\sepa$ will be very important in our studies. Quite often, we will consider $\sepa$ as independent variables instead of inhomogeneities.

We shall consider the quantisation condition as an equation both on the level of operators denoted uniformly as $\hat c_{\pa}$, $\pa=1,\ldots, L$ and on the level of abstract variables denoted as $c_{\pa}$. We shall show eventually that {\it any} $c_{\pa}$ solving \eqref{mastereq} provides eigenvalues for $\hat c_{\pa}$'s.

\begin{example}
Consider a $GL(3)$ spin chain of length $L=3$, and the weight subspace $\VTw$  with $\Lambda=[2,1,0]$.  The Q-system is parameterised by
\be
Q_1 &\propto& x_1^{u/\hbar}\times (u^2+\hat c_{1}^{(1)}u+\hat c_{1}^{(2)})\,,\nonumber\\
Q_2 &\propto& x_2^{u/\hbar}\times (u+\hat c_2^{(1)})\,,\nonumber\\
Q_3 &\propto& x_3^{u/\hbar}\,.
\ee
$\SWe\propto W(Q_1,Q_2,Q_3)=\det\limits_{1\leq a,b\leq 3}Q_a(u+\hbar (2-b))$.
Set for simplicity $x_3=3,x_2=2,x_1=1$, then \eqref{mastereq2} becomes explicitly
\begin{align}
  &\mathllap{(\hat c_1^{(2)}\hat c_2^{(1)}}-\hbar(\hat c_1^{(2)}+\frac 5 2\hat c_1^{(1)}\hat c_2^{(1)})+\frac {\hbar^2}2(9
\hat c_1^{(1)}+7\hat c_2^{(1)})-\frac {15} 2 \hbar^3)
  &&&  &-\inhom[1]\,\inhom[2]\,\inhom[3] \nonumber\\
+ u\times& (\hat c_1^{(2)}+\hat c_{1}^{(1)}c_2^{(1)}+\hbar (-\frac 7 2 \hat c_1^{(1)}-5\hat c_2^{(1)})+\frac{25}2\hbar^2 )&=&& +u\times&(\inhom[1]\inhom[2]+\inhom[1]\inhom[3]+\inhom[2]\inhom[3])\nonumber\\
+u^2\times& (\hat c_1^{(1)}+\hat c_2^{(1)}-6\hbar) &&& +u^2\times& (-\inhom[1]-\inhom[2]-\inhom[3])\nonumber\\
+u^3& &&& +u^3&\,
\end{align}
which yields us three equations satisfied by $c_{1}^{(1)},c_{1}^{(2)},c_{2}^{(1)}\,,$ both on the level of operators and their eigenvalues.  

Counting the number of solutions is easy in this example: one can derive a cubic equation on $c_{2}^{(1)}$ with a cubic term that never vanishes and furthermore we observe that $c_{1}^{(1)},c_{1}^{(2)}$ follow uniquely if we fix the value of $c_{2}^{(1)}$. So there are always three solutions which is the dimension of $\VTw$. For generic values of $\inhom[\pa]$, all solutions are distinct.

\end{example}
\begin{example}
Consider a $\gls{1|1}$ spin chain of length $L=3$, and $\VTw$ with $\Lambda=[\lambda|\nu]=[2|1]$. We parameterise the Q-system by
\be
Q_{1|\es} &\propto& x^{+u/\hbar}\times (u^2+\hat c_{1|\es}^{(1)}u+\hat c_{1|\es}^{(2)})\,,\nonumber\\
Q_{\es|1} &\propto& y^{-u/\hbar}\times (u+\hat c_{\es|1}^{(1)})\,.
\ee
In this case $Q_{\theta}=\SWe\propto Q_{1|1}$, and one needs to compute $Q_{1|1}$ from the finite-difference equation $Q_{1|1}(u+\hbar/2)-Q_{1|1}(u-\hbar/2)=Q_{1|\es}Q_{\es|1}$ which supplies the equations on $c_{1|\es}^{(1)},c_{1|\es}^{(2)}$, and $c_{\es|1}^{(1)}$. For $x=3,y=1$ they are
\be
\label{explsu11twist}
&&u^3-u^2(\se1-3\hbar)+u(\se2-2\se1\hbar+\frac 34\hbar^2)-(\se3-\se2\hbar+\frac 14\se1\hbar^2-\frac 14\hbar^3) \nonumber\\
&&=
(u^2+ c_{1|\es}^{(1)}u+ c_{1|\es}^{(2)})(u+ c_{\es|1}^{(1)})\,,
\ee
where $\se1=\inhom[1]+\inhom[2]+\inhom[3]$, $\se2=\inhom[1]\inhom[2]+\inhom[1]\inhom[3]+\inhom[2]\inhom[3]$, $\se3=\inhom[1]\inhom[2]\inhom[3]$.

Counting solutions for this example is even simpler, and it is a good demonstration of when a supersymmetric system can be advantageous for finding the spectrum of the Bethe algebra. The \lhs of \eqref{explsu11twist} is a degree-three polynomial with all coefficients known through the parameters of the theory. It has three roots, and one choses which of these roots is  %
$-c_{\es|1}^{(1)}$ which fixes $c_{\es|1}^{(1)}$. Values for other $c$'s follow. Hence there are three solutions  which is indeed the dimension of the weight subspace.

\end{example}

\subsubsection*{Twist-less case}
\label{sec:twistlesscase}

For the twist-less case, the symmetry of the system is enhanced as all the generators $\CE_{\sA\sB}$ commute with the Bethe algebra. Now, the Cartan subalgebra of $\glmn$ does not belong to the Bethe algebra and so the latter acting on the spin chain is definitely not maximal commutative. However, if we restrict ourselves to the weight subspace $\VTl$, maximal commutativity on this subspace will follow from completeness.
 \be
\BAL^+:=\BA|_{\VTl}\,.
\ee
We will typically drop the superscript $^+$ and denote the restricted Bethe algebra as $\BAL$.
 
The Q-operators were constructed in \cite{Belitsky:2006cp,Bazhanov:2008yc,Kazakov:2010iu,Frassek:2010ga,Tsuboi:2019vvv} for the case of generic twist. Taking the twist-less limit is quite a tricky procedure \cite{Bazhanov:2010ts} which was analysed substantially in sections 3.3 and 3.4 of \cite{Kazakov:2015efa}. The result of this analysis is that the below-presented properties that define the twist-less Q-system remain true at the level of operators.
 
\paragraph{Long representations}
Consider first a situation when $\Lambda$ is a long representation of $\glmn$. The Young diagram of such a representation touches the corner of the fat hook, consider for instance the situation with $\gm\to \gm-r, \gn\to \gn-r$ in Figure~\ref{fig:ShortLong}.
 
In this case, the Q-operators $Q_{A|J}$ are polynomials in $u$ of degree $M_{A|J}$~\footnote{For comparison with other literature, it might be needed to relabel Q-functions using the maps $a\to\gm+1-a, i\to\gn+1-i$. One checks the notation by asking for which $a,i$ $Q_{a|i}$ is a polynomial of the smallest degree. In our conventions, it is $Q_{\gm|\gn}$.},
\begin{align}
\label{eq:poldegree}
M_{A|J}&=\sum_{a\in A}\hat\lambda_a+\sum_{j\in J}\hat\nu_j-\frac{|A|(|A|-1)}{2}-\frac{|J|(|J|-1)}{2}+|A||J|\,,
\end{align}
where $\hat\lambda_a,\hat\nu_i$ are the shifted weights defined as in Figure~\ref{fig:ShortLong}.

We will need mostly
\be
\label{eq:221}
M_{a|\es}=\hat\lambda_a\,,\quad M_{\es|i}=\hat\nu_i\,,\quad M_{a|i}=\hat\lambda_a+\hat\lambda_i+1\,.
\ee
 
A few modifications have to be made to obtain an equivalent of \eqref{mastereq}. First, we notice that \eqref{sport} fixes $Q_{a|i}$ only up to an additive constant and hence  $\hat c_{a|i}^{(M_{a|i})}$ are new independent parameters used in the computation of $Q_{\fs|\fs}$. Second, the computation of $Q_{\fs|\fs}(u)$ and of other physically relevant quantities such as transfer matrices is invariant under the transformations 
\be
\label{eq:225}
Q_{\es|i}\rightarrow Q_{\es|i}+\alpha\, Q_{\es|j}\,,\quad Q_{a|\es}\rightarrow Q_{a|\es}+\beta\, Q_{b|\es}\,.
\ee 
We impose inequalities $i\geq j$ and $a\geq b$ to preserve the polynomial degrees \eqref{eq:221}. We fix these symmetry transformations by putting to zero all the coefficients $c_{\es|i}^{(\hat\nu_j)}$ for $j\leq i$ and $c_{a|\es}^{(\hat\lambda_b)}$ for $b\leq a$.  

A straightforward counting shows that the set consisting of all non-zero $c_{a|\es}^{(k)},c_{\es|i}^{(k)}$ combined together with  $c_{a|i}^{(M_{a|i})}$ gives us exactly $L$ variables. Denote this set by $C_{\Lambda}$. $Q_{\fs|\fs}\propto Q_{\theta}$ is unambiguously and polynomially reconstructed from $C_{\Lambda}$ according to \eqref{susyw} supplemented with \eqref{sport}, we denote the corresponding operation again as  $\SWe(C_{\Lambda})(u)$ though explicit polynomial realisation of $\SWe$ is different now.

In this modified setting, \eqref{mastereq} holds.
  
 \begin{example}
Consider a $\gls{3}$ spin chain of length $L=3$, and consider states in the representation $\LD=\twoone$. By the recipe of Figure~\ref{fig:ShortLong}, $\hat\lambda_1=4,\hat\lambda_2=2,\hat\lambda_3=0$. Then one generates the Bethe algebra by three Q-functions
\begin{subequations}
\label{Qsgl3}
\be
Q_1 &=& u^4+\hat c_{1}^{(1)}u^3+\hat c_{1}^{(3)}u\,,\nonumber\\
Q_2 &=& u^2+\hat c_{2}^{(1)}u\,,\nonumber\\
Q_3 &=& 1\,.
\ee
\end{subequations}
We fixed $c_{1}^{(2)}=c_{1}^{(4)}=c_{2}^{(1)}=0$ using symmetries of the system. The Wronskian condition \eqref{mastereq} which is $\det\limits_{1\leq a,b\leq 3}Q_a(u+(2-b)\hbar)\propto Q_{\theta}$ provides three equations to be satisfied by $c_1^{(1)},c_{1}^{(3)},c_2^{(1)}$:
\begin{align}
c_{1}^{(3)}-\hbar^2(c_{2}^{(1)}-c_{1}^{(1)})&=8\se3\,,&
3c_{1}^{(1)}c_2^{(1)}-2\hbar^2 &= 8\se2\,,&
3c_{1}^{(1)}+6 c_{2}^{(1)}&= -8\se1\,,
\end{align}
it has two solutions.
\end{example}
 \begin{example}
Consider a $\gls{1|1}$ spin chain of length $L=3$, again in the representation $\LD=\twoone$. By the recipe of Figure~\ref{fig:ShortLong}, $\hat\lambda=1,\hat\nu=2$. Then we use \eqref{eq:poldegree} to deduce the degree of Q-functions and get
\begin{subequations}
\label{Qsgl11}
\be
Q_{1|0} &=& u+\hat c_{1|0}\,,\nonumber\\
Q_{0|1} &=& u+\hat c_{0|1}\,.
\ee
\end{subequations}
Equation $Q_{1|1}^+-Q_{1|1}^-= Q_{1|0}Q_{0|1}$ provides $Q_{1|1}$ up to an additive constant $c_{1|1}$,
\be
Q_{1|1}\propto u^3+\frac 32(\hat c_{1|0}+\hat c_{0|1})u^2+(3\hat c_{1|0}\hat c_{0|1}-\frac 14\hbar^2)u+\hat c_{1|1}\,,
\ee
which together with $c_{1|0}$ and $c_{0|1}$ yields three variables that generate the Bethe algebra. The Wronskian condition \eqref{mastereq} is $Q_{1|1}\propto Q_{\theta}$, it implies the equations on $c$'s:
\be
\label{eq:229}
c_{1|0}+c_{0|1}=-\frac 23 \se1\,,\quad c_{1|0}\,c_{0|1}= \frac 1{3}\se2+\frac 1{12}\hbar^2\,,\quad c_{1|1}=-\se3\,
\ee
which have two solutions.
\end{example}

\paragraph{Short representations}
\label{sec:shortrep}
The Young diagram of a short representation does not touch the internal corner of the $\gm|\gn$ fat hook. Define $r$ according to Figure~\ref{fig:ShortLong}. Introduce sets ${\bf A}=\{\gm-r+1,\gm-r+2,,\ldots,\gm\}$ and ${\bf J}=\{\gn-r+1,\gn-r+2,\ldots,\gn\}$, and label all Q-functions as $Q_{AA_0|JJ_0}$, where $A_0$ is a multi-index from ${\bf A}$ and $J_0$ is a multi-index from ${\bf J}$. Then the properties of the Q-functions can be described as follows \cite{Kazakov:2015efa}: If $|A_0|=|J_0|$ then $Q_{AA_0|JJ_0}=Q_{A|J}$ and, if $|A_0|\neq |J_0|$ then $Q_{AA_0|JJ_0}$ are not uniquely defined in the twist-less limit but also such Q-operators appear in the physically-relevant quantities, such as transfer matrices and $Q_{\fs|\fs}$, in combinations that {\it vanish} in the twist-less limit.

The described property allows us to restrict the $\glmn$ Q-system to the $\gls{\gm-r|\gn-r}$ Q-system defined as $Q_{A|J}^{\rm rest}=Q_{A{\bf A}|J{\bf J}}$ which has the property $Q_{\es|\es}^{\rm rest}=Q_{{\bf A}|{\bf J}}=Q_{\es|\es}=1$ and $Q_{\fs|\fs}^{\rm rest}=Q_{1\ldots \gm-r|1\ldots \gn-r}^{\rm rest}=Q_{\fs|\fs}$. This subsystem is sufficient to generate the Bethe algebra. Since an originally short representation becomes long from the point of view of $\gls{\gm-r|\gn-r}$ subalgebra and since the polynomial degrees are correctly captured by \eqref{eq:poldegree},  we can use the same logic as for the long representations and formulate the supersymmetric Wronskian condition \eqref{mastereq} using $C_{\Lambda}$ of the $\gls{\gm-r|\gn-r}$ Q-system.

\begin{example}
The representation $\twoone$ can be considered as a short one of the $\gls{2|2}$ algebra. Then ${\bf A}=\{2\},{\bf J}=\{2\}$, and so all the physical information is contained in the functions $Q_{12|12}=Q_{1|1}$, $Q_{12|2}=Q_{1|\es}$, $Q_{2|12}=Q_{\es|1}$, $Q_{2|2}=Q_{\es|\es}$. The Wronskian is given by
\be
\SWe=Q_{1|1}Q_{2|2}-Q_{1|2}Q_{2|1}\,.
\ee
While $Q_{1|2}$ and $Q_{2|1}$ are not uniquely defined in the twist-less limit, any prescription would imply that at least either $Q_{1|2}$ or $Q_{2|1}$ vanish and so their product vanishes as well. Given that $Q_{2|2}=Q_{\es|\es}=1$, $\SWe=Q_{\theta}$ implies  $Q_{1|1}=Q_{\theta}$ which fully parallels the above-described $\gls{1|1}$ example.
\end{example}

\section{Completeness}
\label{sec:C}
So far we introduced $\BAL$ -- the restriction of the Bethe algebra to the weight subspace $\VV$ (which is $\VTw$ or $\VTl$). It is generated by the restriction of the Q-operators who in turn are (twisted) polynomials of the spectral parameter.  We also selected precisely $L$ coefficients of these polynomials assembled into the set $C_{\Lambda}$ and explained how they are used to generate the whole $\BAL$. 

From now on, the elements of $C_{\Lambda}$ are labelled in a uniform manner as $c_{\pa}$, $\pa=1,\ldots,L$. It is important to articulate what notation $c_{\pa}$ means exactly. If it is an explicit matrix acting on $\VV$ then we denote it as $\hat c_{\pa}$. In contrast, we agree to denote by $c_{\pa}$ without hat abstract commuting variables that have, by definition, only one property: they satisfy Wronskian Bethe equations $\eqref{mastereq2}$. 

A one-to-one correspondence between $c_\pa$ solving WBE and eigenvalues of $\hat c_{\pa}$ will be established later, in Section~\ref{sec:F}. In the current section, we show that  WBE have precisely the right number of solutions $d_{\Lambda}=\dim_{\mathbb{C}}\VV$. This property is usually referred to as {\it completeness}, why this naming is justified was discussed in the introduction.

\subsection{Analytic description}
\label{sec:basicproperties}
First, we develop some intuition about analytic description of the Wronskian Bethe equations $\SWe_{\pa}(c)=\se{\pa}$.
Think about them as a polynomial map
\begin{equation}
\label{Wmap}
\begin{array}[t]{lrcl}
\SW : &\mathbb{C}^L & \longrightarrow & \mathbb{C}^L\,, \\
& (c_1,\dots ,c_L) & \longmapsto & (\se1=\SWe_1(c),\dots,\se L=\SWe_L(c))\,.
\end{array}
\end{equation}
Denote by $\Cdom\simeq \CC^L$ the domain of definition of the map and by $\Sedom:=\SW(\Cdom)$ its image.

\paragraph{\it Surjectivity.} The map $\SW$ is in fact surjective, that is $\Sedom\simeq\mathbb{C}^L$ which means that the Wronskian relations \eqref{mastereq2} have at least one solution for any complex value $\bsepa$ of $\sepa$.
Indeed, matrix coefficients of $\hat c_{\pa}$ are polynomials in $\inhom[\pa]$, {\it e.g.} by construction of Baxter Q-operators, and so they are defined for any numerical value $\binhom[\pa]$. Furthermore $\hat c_{\pa}$ commute and so they have at least one common eigenvector ${\bf u}({\inhom[]})$. Eigenvalues of $\hat c_{\pa}$ on this vector satisfy \eqref{mastereq2} and so they provide a solution to $\SWe_\pa(c)=\bsepa$ for $\bsepa=\sepa(\binhom[])$.

\paragraph{\it Critical and regular points.}  Denote by $\Ccrit$ the set of all the critical (degeneration) points $c$ where the differential of $\SW$ is not invertible. Its image $\Secrit\equiv \SW(\Ccrit)$ shall be called the set of critical points $\sse$.  Using {\it e.g.} Sard's theorem one states that $\Secrit$ is of measure zero in $\Sedom$. The complement to $\Ccrit$, resp. $\Secrit$, shall be called domain of regular points $c$, resp $\sse$. Restricted to the regular points, the map $\SW$ is locally a diffeomorphism, \ie  for each point $c\notin\Ccrit$ there is a neighbourhood of $\SW(c)$ where $\SW$ can be smoothly inverted. This implies that all solutions to the Bethe equations are {\it distinct} in a neighbourhood around a regular point $\sse$. This also shows that in such a neighbourhood the fibers of $\SW$ are all finite and of the same cardinality ($\SW$ is polynomial and so it cannot have infinite discrete fibers).

\paragraph{\it Properness.} \label{par:prop} All solutions $c_{\pa}$ are bounded at any finite value of $\inhom$'s or, in more abstract terms, the inverse image of a compact set is compact. $\SW$ is then said to be proper. This very important technical point is proved in two independent ways: using the fact that $Q$-operators have bounded spectrum, as is explained in the remark on page~\pageref{propernessfromQ}; and by a direct analysis of the equations themselves, in Appendix~\ref{sec:Finiteness} for the twisted case and Appendix~\ref{sec:deta-constr-mapp} for the twist-less case.

\paragraph{\it Path-connectivity}  $\Ccrit$ can be easily described as $\mathrm{det} \frac{\partial \SWe_{\pa}}{\partial c_{\pa'}}=0$ which is just a polynomial equation on $c_{\pa}$ that, obviously, defines a domain of (complex) co-dimension $1$.

This implies that any two solutions $c_{\pa}$ and $c'_{\pa}$ can be connected by a smooth path $\gamma$ that avoids the singular domain $\Ccrit$. We can always choose $\gamma$ such that its image $\SW(\gamma)$ also passes only through regular points of $\Se$. Note that one or both points $c_{\pa}$ and $c'_{\pa}$ can actually belong to $\Ccrit$. So any singular solution can be obtained as a limit of regular solutions. Sporadic solutions, defined as solutions that exist only for some subspace of points in $\Se$ cannot exist either by the same argument~\footnote{For an example of equations with sporadic solutions, consider $x(x-1)=0, \theta x=0$. For all $\theta\neq 0$ there is only one solution $x=0$. However, for $\theta=0$ there is one extra sporadic solution $x=1$.}. 

We in particular conclude that for {\it any} choice of $\se{\pa}$, the number of solutions of Bethe equations is less or equal to $d_\Lambda$, where $d_{\Lambda}$ is defined as the number of solutions at regular points of $\Se$ (this number does not depend on $\sepa\notin\Secrit$ since the regular domain of $\Sedom$ is path-connected). 

\paragraph{\it Finiteness} By definition, a map is called finite if it is proper and its fibers at all points are finite~\footnote{The concept of finite morphism is usually defined in a more general set-up using a rather abstract algebraic formalism. Here we are working with analytic varieties when the general ``algebraic'' definition is equivalent to the ``topological'' definition that we are using, see \cite{SGA1, EGA1}.}, so SW is an example of such a map. This property will be used later. 
\newline
\newline
As we have established, all solutions to the Bethe equations \eqref{mastereq2} are distinct for $\sepa$ being in the regular domain of $\Sedom$. Some solutions coincide if $\sepa\in\Secrit$, and so the number of distinct solutions is smaller. It is typical to count solutions with multiplicities in such a case. When we deal with equations in several variables, the notion of multiplicity requires an appropriate formalism to be introduced which is our next goal.
\subsection{How to count solutions with multiplicity}
\label{sec:algedes}
Starting from now, we will gradually introduce an algebraic formalism to analyse the Wronskian Bethe equations. We will be using standard terminology from commutative algebra which is briefly summarised in Appendix~\ref{sec:alggeom}.
\newline
\newline
Let us introduce a polynomial ring  $\WAL$ that shall be called the {\it Wronskian algebra} and which is defined as follows. Consider $\CC[\sse][c]$ -- the algebra of polynomials in variables $\se1,\se2,\ldots,\se{L},$ $c_1,\ldots,c_L$; and $\CI_\Lambda=\langle \se1-\SWe_1(c),\se2-\SWe_2(c),\ldots,\se{L}-\SWe_L(c)\rangle$ -- the ideal generated by the equations \eqref{mastereq2}. Then
\be\label{WALdef}
\WAL:={{\CC[\sse][c]}}/{\CI_{\Lambda}}\,.
\ee
Over $\CC$, $\WAL$ is obviously isomorphic to $\CC[c]$. However, additionally, it is also naturally a $\CC[\sse]$-module. Namely, one defines the action of $\sepa$ on $\WAL\simeq\CC[c]$ as follows: we multiply elements of $\WAL$ by $\sepa$ and then replace $\sepa$ with $\SWe_\pa(c)$.

To link the $\CC[\sse]$-module structure with the Wronski map from the previous section we note that $\CC[c]$ is the coordinate ring of $\Cdom$ and $\CC[\sse]$ is the coordinate ring of $\Sedom$. The map 
\be\label{SWstar}
\SW^*:\CC[\sse]\to\CC[c]\,,\quad \chi_\pa\mapsto \SWe_\pa(c)
\ee 
used in the definition of the $\CC[\sse]$-action on $\WAL$ is a pullback of \eqref{Wmap}.

The number of solutions to the Wronskian Bethe equations appears as follows in the algebraic context. We {\it specialise} the Wronskian algebra to the complex point $\bsse$ where we would like to count the solutions. Specialisation is defined as
\be\label{specialisationW}
\WAL(\bar\sse):=\WAL/\langle \sepa-\bsepa \rangle \simeq \CC[c]/\langle \SWe_\pa(c)-\bsepa \rangle\,.
\ee
Then, it is a standard result that the number of solutions of a polynomial set of equations $P_i(x_1,\ldots,x_n)=0$, $i=1,\ldots m$ is equal to the dimension of the quotient ring $\mathcal{R}=\CC[x]/\langle P_1,\ldots, P_m\rangle$ (as a vector space over $\CC$). Moreover, in the case when solutions degenerate, the dimension of the quotient ring is used as {\it a definition}~\footnote{Again, the general definition of multiplicity in the full formalism of algebraic geometry is much more intricate but in our case it is equivalent to the one we use.} of the algebraic number of solutions (\ie  counted with multiplicity). In our case, the quotient ring in question is $\CR=\WAL(\bar\sse)$, and so the algebraic number of solutions of the Wronskian Bethe equation at point $\bsse$ is equal to $\dim_\CC \WAL(\bar\sse)$. Since at all points $\bar\sse$ the number of solutions is finite $\dim_\CC \WAL(\bar\sse)<+\infty$.

To see how this definition comes about in practice, consider the regular representation of the (finite-dimensional) quotient ring $\CR$ which is a map from elements of the ring to the ring endomorphisms defined by the ring multiplication. We can describe this map in terms of explicit matrices. Let $\bas_1,\ldots,\bas_r$ be some basis elements of $\mathcal{R}$. Then, for any $X\in\mathcal{R}$, one has $X\,\bas_i=\sum\limits_{j=1}^r \check X_{ij} \bas_j$, where $\check X_{ij}\in\CC$. The regular representation maps $X$ to the matrix $\check X$ whose components are $\check X_{ij}$.

\begin{example}
\label{ex:36}
Consider $\CR=\CC[x]/\langle x^2-a x+b\rangle $. Elements $x$ and $1$ span $\CR$, choose them as basis elements. Then one has  $x\cdot x=x^2=a\,x -b$, $x\cdot 1=x$ and so
\be
\label{checkx}
\check x=\left(\begin{matrix} 0 & 1 \\  -b & a \end{matrix}\right)\,.
\ee
\end{example}

It is easy to prove that the image of the regular representation is isomorphic to the algebra $\CR$. This allows one to understand properties of a polynomial ring in a more familiar setting of a matrix algebra that we denote as $\check \CR$.  

By the isomorphism, $P_{i}(\check x_1,\ldots,\check x_n)=0$ for $i=1,2,\ldots,m$. So all joint eigenvalues of $\check x_1,\ldots,\check x_n$ are solutions of the set of equations. And each solution should be one of the joint eigenvalues (to see this, take $\sum\limits_j (\check x_k)_{ij} \bas_j=x_k \bas_i$ and evaluate $x_k$ and $\bas_i$, who are polynomials in $x_k$, to numerical values corresponding to the solution of interest).

Hence, when $\check x_{\pa}$ are diagonalisable then it is clear that the number of solutions is equal to the size of the matrix which is the same as the dimension of the quotient ring. Moreover, all solutions should be distinct (otherwise, isomorphism between $\check\CR$ and $\CR$ won't hold).

When $\check x_{\pa}$ are not diagonalisable, one could expect that the solutions degenerate and that the multiplicity of degeneration is the size of the corresponding Jordan block, as in \eqref{checkx} when $a=b=0$. 
But since commuting matrices are not always simultaneously jordanisable, this intuitive picture should be slightly updated:  Any commuting set of matrices, $\check\CR$ in our case, admits a Dunford-Jordan-Chevalley decomposition, namely there is a basis where all the matrices take the form $D+N$, where $D$ is diagonal and $N$ is upper-triangular with zeros on the diagonal. Moreover, all elements of $N$ form a subalgebra in $\check\CR$ known as the nil-radical $\Nil(\check\CR)$ which is the ideal of all nilpotent elements of the ring. The quotient algebra $\diag(\check\CR)\equiv \check\CR/\Nil(\check\CR)$ is isomorphic to the algebra of matrices from $D$. 

Resorting to the regular representation was of course optional, the concepts of the nil-radical $\Nil(\CR)$ and the quotient $\diag(\CR)$ exist for any (commutative) ring. In summary, one has a short exact sequence
\be
\label{ses}
0\longrightarrow \Nil(\CR)\longrightarrow \CR \longrightarrow\diag(\CR) \longrightarrow 0\,.
\ee
$\dim_\CC\diag(\CR)$ is precisely the number of {\it distinct} solutions to the polynomial equations. $\dim_\CC \Nil(\CR)$ counts the amount of degeneration in solutions, and 
\be
\dim_\CC \CR=\dim_\CC\diag(\CR)+\dim_\CC \Nil(\CR)
\ee
is the total number of solutions counted with multiplicity.

We would like to emphasise that $\dim_\CC \CR$ is both the dimension of the quotient ring and the dimension of its regular representation (size of matrices). Eigenspaces of $\check\CR$ are all of dimension one~\footnote{Not to confuse with degeneration of solutions of polynomial equations. By definition, eigenspaces of a matrix $\check X$ are those that are annihilated by $\lambda-\check X$. In contrast, degenerate solutions correspond to existence of vectors that are annihilated by $(\lambda-\check X)^n$ for $n>1$.} and $\check\CR$ is a maximal commutative subalgebra of $\End(\CR)$. This remark will become important in our study of the Bethe algebra.

\subsection{Wronskian algebra is a free \texorpdfstring{$\mathbb{C}[\sse]$}{C[chi]}-module}
\label{sec:wronsk-algebra-free}
The following very powerful result can be proven about the Wronskian algebra:
\begin{proposition}
\label{freeness}
$\WAL$ is a free $\CC[\sse]$-module.
\end{proposition}
\begin{proof}
$\SW^*$ \eqref{SWstar} is a ring morphism from $\CC[\sse]$ to $\CC[c]$ making $\WAL$ a $\CC[\sse]$-algebra and therefore a $\CC[\sse]$-module. On the other hand, we can view $\SW$ \eqref{Wmap} as an algebraic morphism from the variety $\Cdom\simeq\CC^L$ (for the current discussion, the affine space $\mathbb{A}^L$) to $\Sedom\simeq\mathbb{A}^L$. We know that all the fibres of $\SW$ are finite sets and are therefore of dimension $0$. Moreover, $\mathbb{A}^L$, as an algebraic variety, is regular and (therefore) Cohen-Macaulay. Then, by a general result (sometimes called ``miracle flatness theorem'') $\SW^*$ is a flat ring morphism and so $\WAL$ is a flat $\CC[\sse]$-module, see for example \cite{matsumura1989commutative, hartshorne2013algebraic}. Since $\CC[\sse]$ is Noetherian and $\WAL$ is finitely-generated as a $\CC[\sse]$-module (because $\SW$ is finite) it is actually projective \cite{cohn2002further}. Finally, by the Quillen–Suslin theorem \cite{Quillen1976} it is free.
\end{proof}
The above proof looks very short, however it is based on several abstract results from algebraic geometry. In appendix~\ref{sec:pedestrian} we provide an elementary study of the Wronskian algebra which helps in understanding the logic behind the above proof.

By definition of a free module, $\WAL$ has a basis -- a collection of elements $\bas_1,\ldots,\bas_r$ such that any element $\mathfrak{a}\in\WAL$ is represented in a unique way as 
\be
\label{eq:39}
\mathfrak{a}=k_1\,\bas_1+\ldots+k_r\,\bas_r\,,
\ee
where $k_i\in\CC[\sse]$. It is easy to prove that $\bas_1,\ldots,\bas_r$ remains a basis after specialisation \eqref{specialisationW} which leads to an immediate corollary of Proposition~\ref{freeness}:
\begin{corollary}
The algebraic number of solutions of the Wronskian Bethe equations is the same for any value of $\bsepa$.
\end{corollary}
This number is equal to the number $d_{\Lambda}$ of solutions to Bethe equations at regular points.
Next section shows that this number is dim $U_{\Lambda}$.

\subsection{Number of solutions {\it via} Hilbert series}
\label{sec:Hilbert}
Let us find the value of $d_{\Lambda}$ explicitly. With the help of Proposition~\ref{freeness}, counting  is reduced to a simple dimensional analysis as we shall describe now.
\newline
\newline
If one chooses a rule by which we assign degree $0$ to the identity element and some positive integer degrees to other elements of a ring $\CR$ then we can define the ring filtration 
$$
\CF_0\subset\ldots \CF_{k}\subset \CF_{k+1}\subset\ldots\,,
$$ 
where $\CF_{k}$ is the vector subspace of $\CR$ (over $\CC$) spanned by all elements of degree not exceeding $k$. Grading assignment should be compatible with the ring structure meaning that for any $k,k'$ and any ring element $\mathfrak{r}_k$ of degree $k$  one has $\mathfrak{r}_k\CF_{k'}\subset \CF_{k+k'}$.

A useful characterisation of a graded ring is given by its Hilbert series defined as
\be
\ch_\CR(t)=\sum_{k=0}^\infty \dim_{\CC}\left(\CF_{k}/\CF_{k-1}\right) t^k\,.
\ee

Since the Wronskian ring $\WAL$ is $\CC$-isomorphic to $\CC[c_1,\ldots,c_L]$, computing its Hilbert series is particularly simple. It is just given by
\begin{equation}
\ch_{\WAL}(t)=\prod_{\pa=1}^L \frac{1}{1-t^{\deg c_\pa}}\,.
\end{equation}
Recall that $c_{\pa}$ is a selected subset of $c_{A|J}^{(k)}$. Define $\deg c_{A|J}^{(k)}=k$. With the labeling $c_{A|J}^{(k)}$ of \eqref{qAJ},
this grading is consistent with the ring structure because the latter follows from the relations between Q-functions which are polynomials in $u$.  Then we have
\begin{subequations}
\begin{lemma}
\label{thm:Hilbertseries}
The Hilbert series of the Wronskian algebra $\WAL$ of the twisted system for $\Lambda=[\lambda_1,\ldots,\lambda_m;\nu_1,\ldots,\nu_\gn]$ is
\be
\label{twistedhs}
\ch_\Lambda(t):=\ch_{\WAL}(t)=\prod_{a=1}^{\gm}\prod_{k=1}^{\lambda_a}\frac 1{1-t^k}\prod_{i=1}^{\gn}\prod_{k=1}^{\nu_i}\frac 1{1-t^k}\,.
\ee
The Hilbert series of the Wronskian algebra $\WAL$ of the $\glmn$ twist-less system with $\Lambda=\LD$ being a Young diagram depends on the Young diagram $\LD$ alone and is given by
\be
\label{twistlesshs}
\ch_{\LD}(t):=\ch_{\WAL}(t)=\prod_{(a,s)\in\LD}\frac 1{1-t^{h_{a,s}}}\,,
\ee
where the product runs over all boxes of $\LD$ and $h_{a,s}$ is the hook length at box $(a,s)$.
\end{lemma}
\end{subequations}
\begin{proof}
The result in the twisted case is immediately obvious from \eqref{degree} and the fact that $c_1,\ldots,c_{L}$ are precisely all $c_{a|\es}^{(k)}$ and $c_{\es|i}^{(k)}$. For the twist-less case, one needs to perform a little analysis on precisely what $c_{A|J}^{(k)}$ generate the Wronskian algebra. One can do it by filling the boxes of $\LD$ with degrees of the variables $c_\pa$ in a special way. This procedure outlined in Figure~\ref{fig:degrees} clearly establishes a bijection with the lengths of the corresponding hooks.
\end{proof}
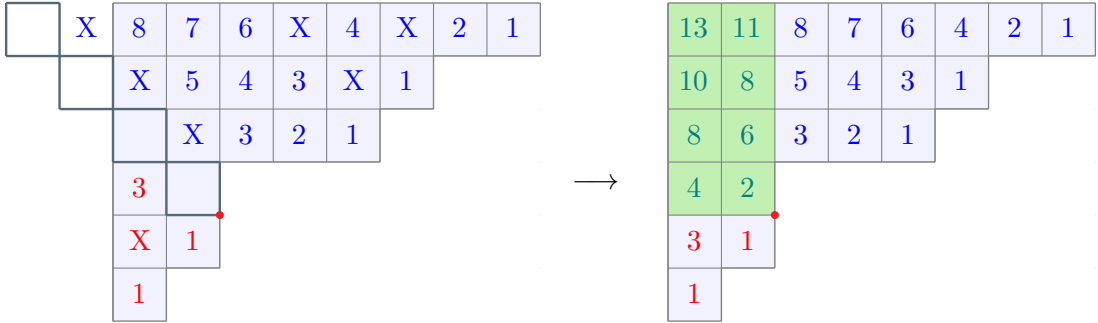
\begin{figure}[t]
\begin{center}
{
\begin{picture}(180,120)(-30,-100)
\thinlines

\color{blue!5}
\polygon*(0,20)(0,-100)(20,-100)(20,-80)(40,-80)(40,-40)(80,-40)(100,-40)(100,-20)(120,-20)(120,0)(160,0)(160,20)(0,20)
\color{gray}
\thinlines
\multiput(0,0)(0,-20){5}{\line(1,0){160}}
\multiput(0,20)(20,0){8}{\line(0,-1){120}}
\color{black}
\color{blue}
\put(147,7){$1$}
\put(127,7){$2$}
\put(106,7){X}
\put(87,7){$4$}
\put(66,7){X}
\put(47,7){$6$}
\put(27,7){$7$}
\put(7,7){$8$}
\put(-14,7){X}
\put(107,-13){$1$}
\put(86,-13){X}
\put(67,-13){$3$}
\put(47,-13){$4$}
\put(27,-13){$5$}
\put(6,-13){X}
\put(87,-33){$1$}
\put(67,-33){$2$}
\put(47,-33){$3$}
\put(26,-33){X}
\color{red}
\put(7,-53){$3$}
\put(6,-73){X}
\put(7,-93){$1$}
\put(27,-73){$1$}
\color{white}
\polygon*(0,-100)(20,-100)(20,-80)(40,-80)(40,-40)(80,-40)(100,-40)(100,-20)(120,-20)(120,0)(160,0)(160,20)(160,-100)(0,-100)
\thinlines
\color{gray}
\drawline(0,-100)(20,-100)(20,-80)(40,-80)(40,-40)(80,-40)(100,-40)(100,-20)(120,-20)(120,0)(160,0)(160,20)(0,20)(0,-100)

\color{cadet}
\thicklines
\drawline(40,-60)(40,-40)(20,-40)(20,-20)(0,-20)(0,0)(-20,0)(-20,20)(-40,20)
\drawline(40,-60)(20,-60)(20,-40)(0,-40)(0,-20)(-20,-20)(-20,0)(-40,0)(-40,20)

\color{Red}
\put(40,-60){\circle*{3}}

\end{picture}
}
\raisebox{50pt}{
$\quad\longrightarrow\quad$}
{
\begin{picture}(160,120)(0,-100)
\thinlines

\color{blue!5}
\polygon*(0,20)(0,-100)(20,-100)(20,-80)(40,-80)(40,-40)(80,-40)(100,-40)(100,-20)(120,-20)(120,0)(160,0)(160,20)(0,20)
\color{green!80!red!30}
\polygon*(0,20)(40,20)(40,-60)(0,-60)(0,20)
\color{gray}
\thinlines
\multiput(0,0)(0,-20){5}{\line(1,0){160}}
\multiput(0,20)(20,0){8}{\line(0,-1){120}}
\color{black}
\color{blue}
\put(147,7){$1$}
\put(127,7){$2$}
\put(107,7){$4$}
\put(87,7){$6$}
\put(67,7){$7$}
\put(47,7){$8$}
\put(107,-13){$1$}
\put(87,-13){$3$}
\put(67,-13){$4$}
\put(47,-13){$5$}
\put(87,-33){$1$}
\put(67,-33){$2$}
\put(47,-33){$3$}
\color{red}
\put(7,-73){$3$}
\put(7,-93){$1$}
\put(27,-73){$1$}
\color{teal}
\put(24,7){$11$}
\put(4,7){$13$}
\put(4,-13){$10$}
\put(7,-33){$8$}
\put(7,-53){$4$}
\put(27,-13){$8$}
\put(27,-33){$6$}
\put(27,-53){$2$}
\color{white}
\polygon*(0,-100)(20,-100)(20,-80)(40,-80)(40,-40)(80,-40)(100,-40)(100,-20)(120,-20)(120,0)(160,0)(160,20)(160,-100)(0,-100)
\thinlines
\color{gray}
\drawline(0,-100)(20,-100)(20,-80)(40,-80)(40,-40)(80,-40)(100,-40)(100,-20)(120,-20)(120,0)(160,0)(160,20)(0,20)(0,-100)

\color{Red}
\put(40,-60){\circle*{3}}

\end{picture}
}

\caption{\label{fig:degrees}\scriptsize Degrees of $c_{\alpha}$ that generate the twist-less Wronskian algebra. On the left: The blue numbers show degrees of $c_{a|\es}^{(k)}$, where $a=1,2,\ldots$ is the corresponding row of the Young diagram. The red numbers show degrees of $c_{\es|i}^{(k)}$ where $i=1,\ldots$ is the corresponding column. Crosses mean the terms $c_{a|\es}^{(k)}$ and $c_{\es|i}^{(k)}$ that are excluded by symmetry reasons, see the discussion after \eqref{eq:221}. For instance, the term $c_{1|\es}^{(9)}$ (the constant term of $Q_{1|\es}$) is cancelled in the linear combination $Q_{1|\es}+\alpha Q_{4|\es}$. On the right: the same blue/red numbers compressed to the right/bottom after the crosses are removed. Note that this arrangement corresponds to the hook lengths. The green rectangular area has exactly as many boxes as number of physically-relevant functions $Q_{a|i}$. The constant terms of these functions, $c_{a|i}^{(k)}$ with $k=\deg Q_{a|i}=\hat\lambda_a+\hat\lambda_i+1$, are independent variables used in generation of the Wronskian algebra. Their degrees, shown in green, match the hook lengths as well.
}
\end{center}
\end{figure}

To make the above-introduced filtration compatible with the $\CC[\sse]$-action on $\WAL$, one needs to define $\deg\sepa=\pa$. This is a simple reflection of the fact that the Wronskian relations $\sepa=\SWe_\pa(c)$ come from equating coefficients between polynomials in $u$. 

Because $\WAL$ is a free $\CC[\sse]$-module, $\dim\WAL(\bsse)=\frac{\ch_\CR(t)}{\ch_{\CC[\sse]}(t)}\raisebox{-0.5em}{$|_{t=1}$}$ for any $\bsepa$. We then compute $\ch_{\CC[\sse]}(t)=\prod\limits_{\pa=1}^L\frac 1{1-t^\ell}$ and easily conclude
\begin{theorem}
The number of solutions of the Wronskian Bethe equations counted with algebraic multiplicity is equal to
\begin{subequations}
\begin{alignat}{3}
\dim\WAL(\bsse) &=  \dim \VTw &=&\, \frac{L!}{\prod\limits_{a=1}^{\gm}\lambda_a!\prod\limits_{i=1}^{\gn}\nu_i!}&&\quad\text{for the twisted case}\,,
\\
\dim\WAL(\bsse) &= \dim \VTl &=&\, \frac{L!}{\prod\limits_{(a,s)\in \Lambda} h_{a,s}}&&\quad\text{for the twist-less case}\,.
\end{alignat}
\end{subequations}
\qed
\end{theorem}
In summary, the number of solutions to the Wronskian Bethe equations is the correct one and hence the equations are complete for any numerical choice of the inhomogeneities.

\section{Faithfulness}
\label{sec:F}
Until now, we have concentrated on the Wronskian algebra $\WAL$ to study the properties of the Bethe equations \eqref{mastereq2}. In particular we have shown that they have the expected number of solutions counted with multiplicity. Now we will establish a bijective correspondence between these solutions and eigenvalues (in degenerate cases, trigonal blocks) of the Bethe algebra. Mathematically, this correspondence is formulated as an isomorphism between the Wronskian algebra and the Bethe algebra. The map $\varphi:c_{\pa}\mapsto \hat c_{\pa}$  from the Wronskian algebra to the Bethe algebra is tautologically surjective, and it is its injectivity (faithfulness) that we need to prove.

In this section, we shall first establish the isomorphism over the polynomial ring $\CC[\sse]$ (which in practice means ``in general position'') and then prove that the isomorphism also holds for any numerical value $\bsepa$ of $\sepa$. A sufficient condition for the isomorphism to hold is that inhomogeneities that solve $\bsepa=\sepa(\binhom[1],\ldots,\binhom[L])$ satisfy $\binhom+\hbar\neq\binhom[\pa']$ for $\pa<\pa'$.

\subsection{Isomorphism between Wronskian and Bethe algebra}
Recall that $\BAL$ -- the Bethe algebra restricted to $U_\Lambda$ -- can be viewed as  an algebra of operators generated by $\hat c_\pa$. As we learned from the previous section, it is beneficial to first keep inhomogeneities as indeterminates, and this is what we are going to do with the Bethe algebra as well. Then $\BAL$ is naturally a subalgebra of $\End(\VV)\otimes\CC[\inhom[]]$, however one should be careful because we define $\BAL$ as the algebra generated from $\hat c_\pa$ by considering any polynomials in $\hat c_\pa$ with coefficients from $\CC$, and not from $\CC[\inhom[]]$. This is a non-trivial remark because of the following
\begin{lemma}\label{thm:1} If $p\in \CC[\inhom[]]$ and $p\times\Id \in \BAL$ then $p$ is a symmetric polynomial in inhomogeneities. 
\end{lemma} 
\begin{proof} Recall that the spin chain vector space is a tensor product of $L$ copies of $\CC^{\gm|\gn}$. Introduce $r_{\pa}(\inhom[])=(\inhom-\inhom[\pa+1])\mathcal{P}_{\pa,\pa+1}+\hbar\Id$, where $\mathcal{P}_{\pa,\pa+1}$ is the graded permutation of two copies of $\CC^{\gm|\gn}$ at the $\pa$'th and the $(\pa+1)$'th position of $(\CC^{\gm|\gn})^{\otimes L}$. This is an intertwining operator that satisfies
\be
\label{eq:braiding}
r_{\pa}\, ev_{(\inhom[1],\ldots,\inhom,\inhom[\pa+1],\ldots,\inhom[L])}\left(T_{ij}\right)\, = ev_{(\inhom[1],\ldots,\inhom[\pa+1],\inhom,\ldots,\inhom[L])}\left(T_{ij}\right)\,r_{\pa}\,
\ee
which can be graphically represented as
\begin{tikzpicture}[baseline=.05cm,scale=.5] 
\foreach\x/\i/\y in {1/1/1.2,2/2/1.2,4/{\pa}/0,5/{\pa+1}/0,7/L/1.2} {\fill (\x,0) circle (4pt);
\draw (\x,\y) -- (\x,-.5) node [right=-.1cm] {\tiny $\inhom[\i]$};}
\draw[rounded corners=.1cm] (4,0) -- (4,.25) -- (4.9,.7) -- (4,1)-- (4,1.2);
\draw[rounded corners=.1cm] (5,0) -- (5,.25) -- (4.1,.7) -- (5,1)-- (5,1.2);
\fill (4.5,0.5) circle (4pt);
\draw (0.5,0) -- (2.5,0) (3.5,0) -- (5.5,0)  (6.5,0) --(7.5,0);
\draw [ dotted, thick] (2.6,0)--(3.4,0) (5.6,0)--(6.4,0);
\end{tikzpicture}=
\begin{tikzpicture}[baseline=-.3cm,xscale=.5,yscale=-.5] 
\foreach\x/\i/\y in {1/1/1.2,2/2/1.2,4/{\pa+1}/0,5/{\pa}/0,7/L/1.2} {\fill (\x,0) circle (4pt);
\draw (\x,\y) -- (\x,-.5) (\x,1.2) node [below=-.1cm] {\tiny $\inhom[\i]$};}
\draw[rounded corners=.1cm] (4,0) -- (4,.25) -- (4.9,.7) -- (4,1) -- (4,1.2);
\draw[rounded corners=.1cm] (5,0) -- (5,.25) -- (4.1,.7) -- (5,1) -- (5,1.2);
\fill (4.5,0.5) circle (4pt);
\draw (0.5,0) -- (2.5,0) (3.5,0) -- (5.5,0)  (6.5,0) --(7.5,0);
\draw [ dotted, thick] (2.6,0)--(3.4,0) (5.6,0)--(6.4,0);
\end{tikzpicture}
and is essentially the Yang-Baxter equation in the physical channel.

If we also introduce $\Pi_\pa$ -- permutation of inhomogeneities $\inhom[\pa]$ and $\inhom[\pa+1]$ in $\CC[\inhom[]]$ then $\Pi_\pa r_\pa$ commutes with the Yangian action and hence with the Bethe algebra. Therefore, if there is any equation $P(\hat c_1,\ldots,\hat c_L)=p(\inhom[1],\ldots,\inhom[L])\times\Id$ that holds so will hold $P(\hat c_1,\ldots,\hat c_L)=p(\inhom[\sigma(1)],\ldots,\inhom[\sigma(L)])\times \Id$ for any $\sigma\in\SG_L$. Since inhomogeneities are independent, this is only consistent if $p$ is a symmetric polynomial.

To be prudent, we notice that the derivation of $P(\hat c_1,\ldots,\hat c_L)=p(\inhom[\sigma(1)],\ldots,\inhom[\sigma(L)])\times\Id$ emerges from the following argument:\\ $\begin{array}{l}0=\Pi_\pa r_\pa(P(\hat c_1,\ldots,\hat c_L)-p(\inhom[1],\ldots,\inhom[\pa],\inhom[\pa+1],\ldots,\inhom[L])\times\Id)\Pi_\pa r_\pa\\\phantom{0}=(P(\hat c_1,\ldots,\hat c_L)-p(\inhom[1],\ldots,\inhom[\pa+1],\inhom[\pa],\ldots,\inhom[L])\times\Id)(\Pi_\pa r_\pa)^2=:A\times (\Pi_\pa r_\pa)^2.\end{array}$\\ $(\Pi_\pa r_\pa)^2=(\hbar^2-(\inhom[\pa]-\inhom[\pa+1])^2)\times\Id$. Then, because components of the matrix $A$ are polynomials in $\inhom$, $0=A\times(\hbar^2-(\inhom[\pa]-\inhom[\pa+1])^2)$ is only possible if $A=0$.
\end{proof}

We see that for instance $\inhom\times\Id$ does not belong to the Bethe algebra, except for $L=1$. But any symmetric polynomial in inhomogeneities (times the identity operator) is an element of $\BAL$ because of \eqref{mastereq2}. So $\BAL$ is naturally a $\CC[\sse]$-algebra. 

Recall now also the definition of the Wronskian algebra $\WAL$ \eqref{WALdef} which is a polynomial algebra generated by $c_\pa$ and which is also a $\CC[\sse]$-algebra.

There is a potential difference between $\WAL$ and $\BAL$. The generators of the Wronskian algebra, by definition, satisfy only \eqref{mastereq2}. The generators of the Bethe algebra are certain explicit operators and they could in principle satisfy some other additional constraints. However, we can show that they do not.

\begin{theorem}
\label{thm:isomorphism}
The map $\varphi$ defined as
\begin{equation}
\label{isomorphism}
\varphi:\CR\longrightarrow\BAL\,, \qquad \varphi:\,c_{\pa} \longmapsto \hat{c}_{\pa}
\end{equation}
is an isomorphism of $\CC[\sse]$-algebras.
\end{theorem}
In other words, $\hat{c}_{\pa}$ not only satisfy \eqref{mastereq2} but  any polynomial relation between $\hat{c}_{\pa}$ with coefficients in $\CC[\sse]$ should follow from \eqref{mastereq2}.

The proof below makes precise the following argument: as there are as many variables $c_\pa$ as the parameters $\sepa$ in \eqref{mastereq2}, there cannot exist an extra relation between the variables because it would imply a relation between the parameters which are known to be independent.

\begin{proof} We need to show that the exhibited map $\varphi$ is a well-defined (consistent) morphism and that it is surjective and injective. It is well-defined because Q-operators form a commutative algebra and they satisfy  \eqref{mastereq2} from the very derivation of this relation. It is surjective because  $\hat{c}_{\pa}$ generate $\BAL$.

The non-trivial part is the injectivity (faithfulness). To prove this we take an element $P\in\WAL$ (so just a polynomial in the variables $c_{\pa}$ and $\se{\pa}$ modulo relations in the ideal) and show that $\varphi(P)=0$ implies $P=0$.

Note that $P$ can be viewed as a polynomial in $c_{\pa}$ with constant coefficients as all occurrences of $\se{\pa}$ can be replaced by $\SWe_{\pa}(c)$. Since $\varphi(P(c))=P(\hat c)$, $P$ has to vanish every time when $c_{\pa}$ are eigenvalues of $\hat c_{\pa}$ on a joint eigenvector. Then it suffices to construct enough of such eigenvalues to conclude that $P=0$.

To this end consider ${\sse}\notin \Secrit$. There exists a neighbourhood $\mathcal{O}_{\sse}$ where all the solutions of \eqref{mastereq2} are distinct and can be parameterised by $d_{\Lambda}$ diffeomorphisms $\SW^{-1}_i$ from $\mathcal{O}_{\sse}$ to $d_{\Lambda}$ non-intersecting open sets $U_i$ in $\Cdom$. 

We know that for all points of $\mathcal{O}_{\sse}$ the Bethe algebra has at least one common eigenvector and that the corresponding eigenvalues of $\hat c_{\pa}$ provide a solution of \eqref{mastereq2}. By choosing in some way exactly one eigenvector at each point of $\mathcal{O}_{\sse}$ we create a disjoint partition of $\mathcal{O}_{\sse}$ into $d_\Lambda$ sets $\mathcal{O}_i$ corresponding to points of $\mathcal{O}_{\sse}$ where the common eigenvector gives the $i$-th solution. The closure (in $\mathcal{O}_{\sse}$) of one of the $\mathcal{O}_i$'s, say $\bar{\mathcal{O}}_1$, contains an $L$-dimensional ball $\mathcal{O}$. This is proved as follows~\footnote{This can be also proven by arguing that the common eigenvector can be chosen continuously in which case taking closure is unnecessary as well.}. Consider $\lambda$ the Lebesgue measure on $\mathcal{O}_{\sse}$ normalised to $1$. Then either $\lambda(\bar{\mathcal{O}}_1)=1$ and therefore $\bar{\mathcal{O}}_1=\mathcal{O}_{\sse}$ since its complementary in $\mathcal{O}_{\sse}$ is an open set of measure zero or $\lambda(\bar{\mathcal{O}}_1)<1$ in which case by restricting to its complementary in $\mathcal{O}_{\sse}$ we are brought back to the same problem but with $d_\Lambda-1$ sets. We conclude by induction. 

Then $P$ vanishes on $\SW^{-1}_1(\mathcal{O}_1)$ and since the zeros of a polynomial form a closed set and $\SW^{-1}_1$ is a diffeomorphism on $\mathcal{O}_{\sse}$ it will also vanish on $\SW^{-1}(\bar{\mathcal{O}}_1)$ which contains an $L$-dimensional ball. $P$ being a polynomial in $L$ variables thus implies $P=0$.
\end{proof}

One may ask whether there are some additional polynomial relations between $\hat c_\pa$ with coefficients being non-symmetric polynomials of inhomogeneities. This is not possible either as can be shown using a slightly updated version of Lemma~\ref{thm:1}, see Appendix~\ref{sec:cth}.

\subsection{What can happen upon specialisation}
\label{sec:specialisation}
Now we shall consider what happens with this isomorphism when inhomogeneities $\inhom[\pa]$ get concrete numerical values. We call this procedure specialisation at point $\binhom[]$.

Specialisation of the Bethe algebra $\BAL(\binhom[])$ is replacing $\inhom$ with $\binhom$ in {\it all matrix entries} of the operators $\hat c_{\pa}$. On the other hand, specialisation of the Wronskian algebra is 
\be
\WAL(\binhom[])\equiv\WAL(\bsepa=\sepa(\binhom[]))\simeq \WAL/{\langle \sse-\bsse\rangle}\,.
\ee 
Its image under the map $\varphi$ is $\BAL/{\langle (\sse(\theta)-\bsse)\times\Id\rangle}$ which explicitly means the following: replace $\se{\pa}$ with its numerical value each time it {\it multiplies} some matrix belonging to $\BAL$. This operation is less restrictive than specialisation of the Bethe algebra and hence one can state that the morphism
\begin{equation}
\label{specmorphism}
\varphi_{\binhom[]}~:~\WAL({\binhom[]})\longrightarrow \BAL(\binhom[])
\end{equation}
is surjective but may have a non-zero kernel. We denote by $\Bad$ the set of $\binhom[]$ when $\varphi_{\binhom[]}$ is not an isomorphism.

\begin{example}
\label{ex:gb}
Consider a  Wronskian algebra $\WA$ realised by relations $c_1+c_2=\se1$ and $c_1c_2=\se2$~\footnote{Up to an isomorphism, it is the Wronskian algebra $\WA_{\twoone}$ with $\hbar =0$ which is partially specialised to {\it e.g.} $\theta_3=0$, {\it cf.} Appendix~\ref{sec:OE} }. It is a free $\CC[\se1,\se2]$-module, for the basis one can choose $1,c_1$, and the ring multiplication rule follows from $c_1^2-\se 1 c_1+\se 2=0$. Then $\check c_1=\mtwo{0}{1}{-\se2}{\se1}$.

Consider two ``Bethe algebras'' $\BA^{\rm good}$ and $\BA^{\rm bad}$, with, respectively,
\be
\label{eq:45}
\hat c_1^{\rm good}=\mtwo {\inhom[1]}{1}{0}{\theta_2}\,,\quad {\rm and}\quad \hat c_1^{\rm bad}=\mtwo {\inhom[1]}{0}{0}{\theta_2}\,.
\ee
They are both $\CC[\se1,\se2]$-isomorphic, as algebras, to $\WA$. Note however that they realise non-isomorphic representations over  $\CC[\se1,\se2]$, \ie  there is no intertwiner  matrix mapping $c_1^{\rm good}$ to $c_1^{\rm bad}$ whose coefficients are polynomial in $\se 1,\se 2$.

If we specialise at any point where $\binhom[1]\neq \binhom[2]$, the corresponding $\varphi_{\binhom[]}$ would be an algebra isomorphism both for $\CB^{\rm good}(\binhom[])$ and $\CB^{\rm bad}(\binhom[])$, also there would be obviously an intertwiner over $\mathbb{C}$ between ``good'' and ``bad'' representations making them isomorphic.

Now, let us specialise to a point $\binhom[1]=\binhom[2]$. The specialised Wronskian algebra becomes a two-dimensional algebra over $\CC$ generated by $1,c_1$ and relation $(c_1-\binhom[1])^2=0$. It is isomorphic to the algebra generated by $\check c_1=\mtwo{0}{1}{-\binhom[1]^2}{2\binhom[1]}$ which cannot be diagonalised, \cf \eqref{checkx}.

The Bethe algebra $\CB^{\rm good}(\binhom[])$ is also two-dimensional and isomorphic to $\WA(\binhom[])$ whereas $\CB^{\rm bad}(\binhom[])$ is one-dimensional since $c_1-\binhom[1]$ is in the kernel of $\varphi_{\binhom[]}$.
\end{example}

The difference between ``good'' and ``bad'' cases is in the presence of the nilpotent piece $\mtwo{0}{1}{0}{0}$ in $\hat c_1^{\rm good}$ which becomes an element of $\BA(\binhom[])$ each time $\binhom[1]=\binhom[2]$. This well illustrates what happens in the general situation. As $\BAL(\binhom[])$ is a commutative algebra, we can define the short exact sequence \eqref{ses} for it. Then we can state the following
\begin{theorem}
\label{eq:alwaysiso}
In the map between the two sequences
\be
\begin{tikzcd}
0 \arrow{r}{} &
 \Nil(\WAL(\binhom[])) \arrow{r}{} \arrow[swap]{d}{\varphi_{\binhom[]}^{\rm nil}} &
 \WAL(\binhom[]) \arrow{r}{} \arrow[swap]{d}{\varphi_{\binhom[]}} &
 \diag(\WAL(\binhom[])) \arrow{r}{} \arrow[swap]{d}{\varphi_{\binhom[]}^{\rm diag}} &
 0 
 \\%
 0 \arrow{r}{} &
 \Nil(\BAL(\binhom[])) \arrow{r}{}  &
 \BAL(\binhom[]) \arrow{r}{}  &
 \diag(\BAL(\binhom[])) \arrow{r}{} &
 0
\end{tikzcd}
\,,
\ee
$\varphi_{\binhom[]}^{\rm diag}$ is an isomorphism for any $\binhom[]$.
\end{theorem}
The isomorphism  $\varphi_{\binhom[]}^{\rm diag}$ literally means that the distinct solutions of the Wronskian Bethe equations are in one-to-one correspondence with the eigenspaces of the Bethe algebra (there are no non-physical solutions). 
 
In the regular case when $\sse(\binhom[])\notin\Secrit$,  $\Nil(\WAL(\binhom[]))={0}$ and so the theorem implies that the Wronskian and the Bethe algebras are isomorphic. Hence $\sse(\Bad)\subset\Secrit$.
 
In the degenerate case $\sse(\binhom[])\in\Secrit$, the non-isomorphism between the Wronskian and Bethe algebras, if present, can be only due to $\varphi_{\binhom[]}^{\rm nil}$ having non-zero kernel. Roughly speaking, one can only lose information about Jordan block structure. The Bethe algebra could in principle have eigenspaces of dimension higher than $1$ while this never happens with $\cWAL$, see Section~\ref{sec:algedes}. 

\begin{proof}
First note that $\varphi_{\binhom[]}^{\rm nil}$ and hence $\varphi_{\binhom[]}^{\rm diag}$ are well-defined because the nil-radical is an ideal and the image of a nilpotent element is nilpotent.

Take a regular point $\binhom[]$ ($\sse(\binhom[])\notin\Secrit$) and consider a linear combination $X$ of $c_\pa$'s taking pairwise distinct values at the $d$ solutions~\footnote{We do not assume $d=d_{\Lambda}$ to make this proof independent of the counting result of Section~\ref{sec:Hilbert}. We also note that this proof does not rely on Proposition \ref{freeness}.} of \eqref{mastereq2} at $\binhom[]$ (this is possible since all the solutions are distinct). Then $(X^i)_{0\leq i\leq d-1}$ is a basis of $\WAL(\binhom[])$ and by continuity it will remain a basis upon specialisation to any $\theta$ in some open neighbourhood $\mathcal{O}_{\binhom[]}$. Other choices of local bases are possible but for convenience we will work with this one.

Suppose $\varphi_{\theta}$ \emph{is not} an isomorphism for all $\theta\in\mathcal{O}_{\binhom[]}$. Then $(\hat{X}^i)_{0\leq i\leq d-1}$ are not linearly dependent for all $\theta\in\mathcal{O}_{\binhom[]}$. Construct columns from the $d^2$ components of the matrices $\hat{X}^i$ and combine the columns into a $d^2\times d$ matrix.  Linear dependence implies that all of the $d\times d$ minors of this matrix vanish on $\mathcal{O}_{\binhom[]}$. Since these minors are polynomials in $\inhom$ this means that they are zero as polynomials. This in turn provides a non-trivial relation $\sum_{i=1}^{d-1}p_i(\inhom[\pa]) \hat{X}^i=0$ with $p_i\in\CC[\theta]$.

Now we would like to be able to take $p_i\in\CC[\sse]$. To this end we use the braiding property \eqref{eq:braiding} which implies $\sum_{i=1}^{d-1}p_i(\theta_{\sigma(\pa)}) \hat{X}^i=0$ for any $\sigma\in\SG_L$. Thus we can replace the $p_i$ by their symmetric part. To ensure that it is non-zero for at least one of them we can multiply the relation we started with by $\prod\limits_{\sigma\in\SG_L\backslash\mathrm{Id}}p_k(\inhom[\sigma(\pa)])$ for some non-zero $p_k$ and then take the symmetric part. 

In the end we obtain a non-zero polynomial $P$ with coefficients in $\CC[\sse]$ of degree smaller or equal to $d-1$ such that $P(\hat{X})=0$. But by the isomorphism \eqref{isomorphism} this implies that $P(X)\in\mathcal{I}_\Lambda$. Specialising $\WAL$ at a point of $\mathcal{O}_{\binhom[]}$ where one of the coefficients of $P$ does not vanish we obtain a contradiction with the fact that $(X^i)_{0\leq i\leq d-1}$ must be a basis at that point.

Therefore $\varphi_{\binhom[]'}$ is an isomorphism for at least one regular point $\binhom[]'\in\mathcal{O}_{\binhom[]}$. By path-connectivity this immediately propagates to all regular points. Indeed, $\varphi_{\binhom[]}$ can cease to be an isomorphism only if the dimension of the Bethe algebra drops but since the spectrum of Q-operators (the set of roots of their characteristic polynomials) is continuous in $\theta$ this can only happen when two solutions cross, that is, at singular points.

At singular points, by continuity of the spectrum, solutions of the Wronskian Bethe equations are still in bijection with the spectrum of Q-operators. The only information that can be lost is the multiplicity of the solutions. Then considering $\varphi_{\binhom[]}$ up to nilpotent parts restores the isomorphism.
\end{proof}

Note that the set of $\binhom[]$ where $\varphi_{\binhom[]}(X^i)$ for $0\leq i\leq d-1$ do not form a basis is {\it a priori} not related to the set of $\binhom[]$ for which $\sse(\binhom[])\in\Secrit$. Hence we can typically expect that $\sse(\Bad)$ is of measure zero inside $\Secrit$. 

\paragraph{Remark} \label{propernessfromQ} The continuity of the spectrum of Q-operators combined with the above theorem provides an immediate proof that $\SW$ is proper as was previously announced. There is no circular argument as we did not use properness in the proof above.

\subsection{Specialisation of the isomorphism}
\label{sec:specialisationofiso}
Although the set $\Bad$ where the Wronskian and the Bethe algebras are not isomorphic is constrained to be, most likely, in a measure zero subset of critical points $\binhom[]$, we still do not have means to locate $\Bad$. This is unsatisfactory because we cannot guarantee to be outside $\Bad$ for physically interesting cases, for instance when all inhomogeneities coincide. In this section we will provide an explicit constraint on $\Bad$. Since the required formalism is quite heavy we will only present the logic behind it and the final results. The technical details are postponed to Appendix~\ref{sec:cyclic}.

The main conceptual step is the following. Although the Wronskian and the Bethe algebras were shown to be isomorphic in Theorem~\ref{thm:isomorphism}, there is an important qualitative difference between them. Namely, the Bethe algebra is represented by matrices and so it naturally acts on a vector space (the spin chain Hilbert space), whereas the Wronskian algebra is abstractly defined by generators and relations and does not admit such a representation. The only natural space on which $\WAL$ could possibly act is itself (\ie by the regular representation). In addition to an isomorphism between algebras, we would like to build an isomorphism between this representation of $\WAL$ and the physical representation of $\BAL$.

To build such an isomorphism, a standard procedure is to try to find a cyclic vector. By definition, for a given algebra $\mathcal{A}$ that acts on some vector space $V$, a vector $\omega$ is said to be cyclic if the action of $\mathcal{A}$ on $\omega$ spans $V$. Then $V$ is said to be a cyclic $\mathcal{A}$-module. Equivalently, $\omega$ is cyclic if and only if the map $\psi_\omega : \mathcal{A}\rightarrow V$, $A\mapsto A\cdot\omega$ is surjective. If moreover it is injective, it is an isomorphism identifying $V$ with the regular representation of $\mathcal{A}$.

In the case of $\BAL$ acting on $U_\Lambda\otimes \CC[\inhom[]]$, it turns out that $\psi_\omega$ is injective for \emph{any} nonzero vector $\omega$ as shown in Lemma~\ref{thm:nonzero}. Unfortunately, the image of $\psi_\omega$ is not $U_\Lambda\otimes \CC[\inhom[]]$ except probably for $L=1$. This can already be seen by the fact that  $\BAL$ as an algebra involves only symmetric polynomials $\sepa$ whereas the matrix coefficients of $\hat c_\pa$ are from $\CC[\inhom[]]$. Nevertheless, for a specific and unique (up to normalisation) choice of $\omega$ one can explicitly identify the image of $\psi_\omega$ as a subspace $\lVVS\subset U_\Lambda\otimes \CC[\inhom[]]$ invariant under an action of the symmetric group $\SG_L$ commuting with the Yangian (Lemma~\ref{thm:chiiso}). Thus the regular representation of $\WAL$ can be identified with the representation $\lVVS$ of $\BAL$.

The above remarks are the tools to prove a powerful and explicit constraint on $\Bad$. Since it is a central result to us, we first recall the definitions of all the objects.
\newline
\newline
\indent $\WAL(\bar\sse)$ is the specialised Wronskian algebra at point $\bar\sse\equiv(\bse{1},\ldots,\bse{L})\in\Sedom\simeq\CC^L$. It is defined as $\WAL(\bar\sse):=\CC[c_1,\ldots,c_L]/\CI$, where $\CI :=\langle\SW_1-\bse1,\ldots\SW_L-\bse{L}\rangle$ is an ideal in $\CC[c_1,\ldots,c_L]$. For the definition of $\SW_\pa$ see \eqref{mastereq2}. Also denote $\hat\CI:=\varphi(\CI)$, where $\varphi$ is the map \eqref{isomorphism}, and $\CJ:=\langle\se1-\bse1,\ldots,\se{L}-\bse{L}\rangle$ considered as a (maximal) ideal of $\CC[\sse]$.

$\BAL(\binhom[])$ is the Bethe subalgebra of the Yangian in the spin chain representation at point $\binhom[]=(\binhom[1],\ldots,\binhom[L])$ restricted to the weight subspace $U_{\Lambda}$. It is generated by operators $\hat c_1,\ldots,\hat c_L$. Matrix entries of these operators are polynomials in inhomogeneities $\inhom$ that are being set to values $\binhom$.

$\chi_\pa(\inhom[]):=\sum\limits_{1\leq i_1<\ldots<i_\pa\leq L}\inhom[i_1]\ldots\inhom[i_\pa]$ are the elementary symmetric polynomials of degree $\pa$, $\pa=1,\ldots,L$.

Let us first state an immediate consequence of Theorem~\ref{thm:isomorphism} and the discussion above.
\begin{lemma}
\label{thm:trivialsp}
\begin{itemize}
\item[i)] $\varphi$ induces an isomorphism of $\CC$-algebras $\WAL(\bar\sse)\simeq\BAL/\hat{\CI}$,
\item[ii)] $\psi_\omega$ induces an isomorphism of representations $\WAL(\bar\sse)\simeq\lVVS/\CJ\cdot\lVVS$.
\end{itemize}
\end{lemma}
\begin{proof}
One easily checks that indeed $\psi_\omega(\hat\CI)=\hat\CI\cdot\omega=\CJ\cdot\BAL\cdot\omega=\CJ\cdot\lVVS$.
\end{proof}
This does not seem to be a very helpful statement since it is not clear how one should interpret the abstract quotients $\BAL/\hat{\CI}$ and $\lVVS/\CJ\cdot\lVVS$. However, Theorem~\ref{thm:iso} implies the following
\begin{theorem}
\label{repisotheorem} 
If $\bsepa=\sepa(\binhom[1],\ldots,\binhom[L])$ and $\binhom+\hbar\neq\binhom[\pa']$ for $\pa<\pa'$ then
\be
\label{eq:47}
\mathrm{ev}_{\binhom[]}: \lVVS/\CJ\cdot\lVVS   \longrightarrow  U_\Lambda
\,,\quad
[v] & \longmapsto  v(\binhom[1],\ldots,\binhom[L])
\ee
is an isomorphism of the Bethe algebra representations.
\qed
\end{theorem}
Here $\rm{ev}_{\binhom[]}$ denotes the map~\footnote{Not to confuse with $ev_{\binhom[]}$ in \eqref{eq:ev2}} induced by $\rm{Ev}_{\binhom[]}$, the evaluation of vectors of $\lVVS\subset U_\Lambda\otimes \CC[\inhom[]]$ at $\binhom[]$. As $\rm{Ev}_{\binhom[]}(\CJ\cdot\lVVS)=0$, \eqref{eq:47} is well-defined. Since the Bethe algebra is represented by $\BAL/\hat{\CI}\simeq \WAL(\bar\sse)$ in $\End (\lVVS/\CJ\cdot\lVVS)$ and by $\BAL(\binhom[])$ in  $\End(U_\Lambda)$ we thus obtain
\begin{theorem}
\label{isotheorem} 
If $\bsepa=\sepa(\binhom[1],\ldots,\binhom[L])$ and $\binhom+\hbar\neq\binhom[\pa']$ for $\pa<\pa'$ then
\be
\varphi_{\binhom[]}:  \WAL(\bar\sse)  \longrightarrow  \BAL(\binhom[])
\,,\quad
c_{\pa} & \longmapsto  \hat c_{\pa}
\ee
is an algebra isomorphism over $\CC$. \qed
\end{theorem}

This result improves Theorem~\ref{eq:alwaysiso} by giving an explicit condition under which not only the diagonal but also the nilpotent parts of the Wronskian and Bethe algebras are isomorphic. Note that values of $\bsepa$ are not restricted in any way, there is only a restriction on which solution of $\bsepa=\sepa(\binhom[1],\ldots,\binhom[L])$ can be taken. 

Assuming the condition $\binhom+\hbar\neq\binhom[\pa']$ for $\pa<\pa'$ is satisfied, the construction above has several immediate consequences. Composing $\mathrm{ev}_{\binhom[]}\circ\psi_\omega$ or, equivalently, acting with $\WAL(\bar\sse)$ on $\omega(\binhom[])$ \via $\varphi_{\binhom[]}$ we obtain
\begin{corollary}
\label{repisomorphismsp}
The spin chain representation of the Bethe algebra restricted to $U_\Lambda$ is isomorphic to the regular representation of $\WAL(\bsse)$.
\qed
\end{corollary}

This result is important for separation of variables, as is discussed in Section~\ref{sec:SoV}.

\begin{corollary}
 $\BAL(\binhom[])$ is a maximal commutative subalgebra of $\End(U_{\Lambda})$.  
\end{corollary}
\begin{proof}
$\cWAL(\bar\sse)$ -- the image of generators $\WAL(\bar\sse)$ under the regular representation -- is a maximal commutative subalgebra of $\End_{\CC}(\WAL(\bar\sse))$ (this is true for any regular representation, see Section~\ref{sec:algedes}).
\end{proof}
In physical terms, this means that the Bethe algebra contains all commuting charges of the system.
\begin{corollary}
\label{thm:49}
If $\BAL(\binhom[])$ is diagonalisable then its spectrum is simple.
\end{corollary}
\begin{proof}
A diagonalisable commutative algebra of matrices whose spectrum is not simple is not maximal commutative.
\end{proof}
In particular, the Bethe algebra has simple spectrum every time it is invariant under a Hermitian conjugation which is often the case in physical applications. Simplicity of the spectrum allows us to introduce a new way of classifying solutions by continuously deforming inhomogeneities, see Section~\ref{sec:alternative}.

\begin{example}
Let us anticipate on the example detailed in Appendix~\ref{sec:OE} corresponding to a $\gl_2$ non-twisted spin chain of length $L=3$. The Bethe algebra on the two-dimensional highest-weight subspace $V_{\begin{tikzpicture}[scale=.15,baseline=-.1cm]
\draw (1,1) |- (0,-1) |- (2,1) |-(0,0);
\end{tikzpicture}}$ is generated (as a $\CC[\sse]$-module) by the identity and the non-trivial operator
\begin{equation}
  \CCC=
\begin{pmatrix}
    2\sse_1-\sqrt 3 (\inhom[1]-\theta_3)&\sse_1-3\theta_2+\sqrt 3\hbar\\
    \sse_1-3\theta_2-\sqrt 3 \hbar&2\sse_1+\sqrt 3 (\inhom[1]-\theta_3)
  \end{pmatrix}\,.
\end{equation}
As long as $\CCC$ is not proportional to the identity, the evaluated Bethe algebra will be of dimension $2$ and therefore isomorphic to the Wronskian algebra. We see that $c\propto \mathbbm 1$ if and only if $\hbar=0$, that is in the Gaudin limit, and  $\binhom[1]=\binhom[2]=\binhom[3]$.

This example demonstrates that the hypothesis $\inhom[\pa']\neq \inhom+\hbar$ in Theorem~\ref{isotheorem} is not a necessary condition.

\end{example}

\subsection{Construction of Bethe vectors}
Let us finally comment on how to use the established isomorphism to construct in a bijective way eigenstates of the Bethe algebra from solutions of the Bethe equations. This gives us the practical meaning of the words ``completeness'' and ``faithfulness''.

We can always take an element of the Wronskian algebra $X$ such that $(X^k)_{0\leq k\leq d_{\Lambda}-1}$ forms a basis. In other words the polynomial of smallest degree such that $P(X)=0$ in $\WAL(\bsse)$ is of degree $d_{\Lambda}$~\footnote{The existence of such an $X$ is obvious at non-degenerate points $\binhom[]$ and otherwise follows from the analysis in Appendix~\ref{sec:multiplicity}. We introduced $X$ for clarity, but the discussed construction of the Bethe algebra eigenstates can be also formulated in a way that does not rely on the existence of $X$.}. Its roots $x_i$, $i=1,\ldots, d_{\Lambda}$ is a way to encode the solutions of WBE \eqref{mastereq}. 

By the established isomorphism, $P$ is both the characteristic and minimal polynomial of the matrix $\hat X\in \BAL(\binhom[])$ which means that any smaller-degree polynomial of $\hat X$ yields a non-zero matrix. Let $x_i$ be an eigenvalue of $\hat X$. Then $\frac{\det(\lambda-\hat X)}{\lambda-x_i}$ is a polynomial in $\lambda$. Take a cyclic vector $\omega$~\footnote{Corollary~\ref{repisomorphismsp} implies that under the usual assumption on $\binhom[]$ such a vector always exists. By continuity, this implies that any generic vector will also be cyclic.} and define
\be\label{ketw}
v_{x_i}=\frac{\det(\lambda-\hat X)}{\lambda-x_i}_{|\lambda =\hat{X}}\omega\,.
\ee
Since $\det(\lambda-\hat X)_{|\lambda = \hat{X}}=0$, one has $(\hat X-x_i)v_{x_i}=0$ and so  $v_{x_i}$ is an eigenvector of $\hat{X}$ with eigenvalue $x_i$. 

We emphasise that as $B_{\Lambda}(\binhom[])$ is isomorphic to the regular representation of the Wronskian algebra all eigenspaces are one-dimensional and so all eigenvectors with eigenvalue $x_i$ are collinear with $v_{x_i}$ which guarantees the bijection.

In case of degeneration of solutions, different Jordan blocks of $\hat X$ must have distinct eigenvalues.  If $x_i$ is a root of multiplicity $n$ then $v_{x_i}^{(m)}\equiv \frac{\det(\lambda-\hat X)}{(\lambda-x_i)^m}_{|\lambda =\hat{X}}\omega$, $m\leq n$, provide a Jordan basis for $\hat X$. Since $\hat{X}$ generates $B_{\Lambda}(\binhom[])$ this basis will also trigonalise the Bethe algebra.

\section{Various parameterisations of the Bethe algebra}
\label{sec:variouspar}
Although we proved that the number of solutions of the Wronskian Bethe equations is the correct one, we did not develop any intuition about how these solutions are organised. We address this issue in the next two sections by proposing techniques to systematically label solutions.
\subsection{Restriction and extension of Q-systems}
\label{sec:restr-extens-q}
We introduced a restricted Q-system on page~\pageref{sec:shortrep} to cover the case of short representations: for special sets ${\bf A},{\bf J}$, $Q_{A|J}^{\rm rest}:=Q_{A{\bf A}|J{\bf J}}$. Now, we remark that the Q-functions $Q_{A|J}^{\rm rest}=Q_{A{\bf A}|J{\bf J}}$  satisfy the QQ-relations \eqref{QQrelations} and $Q_{\fs|\fs}^{\rm rest}=Q_{\fs|\fs}$ for {\it any} sets ${\bf A},{\bf J}$.

If $Q_{{\bf A}|{\bf J}}=1$ then one can interpret the restricted Q-system as a Q-system  of a smaller $\gls{\gm'|\gn'}$ algebra, where $\gm'=\gm-|{\bf A}|,\gn'=\gn-|{\bf J}|$. The condition $Q_{{\bf A}|{\bf J}}=1$ is of course non-trivial to demand. By  counting degrees of polynomials according to \eqref{eq:poldegree} we see that this is possible if  ($\gm,\gn$) is outside of the Young diagram or if ($\gm,\gn$) is situated on the boundary of the Young diagram such that ($\gm',\gn'$) is also on the boundary. We note that the degrees of the restricted Q-functions come out to be given by \eqref{eq:poldegree} for the $\gls{\gm'|\gn'}$ Q-system.

Let us now construct an opposite to the restriction procedure. For simplicity consider an ``elementary'' move. Define an extension from the $\gls{\gm|\gn-1}$ Q-system to the $\glmn$ Q-system as follows:
\be
\label{eq:ext1}
Q_{A|J\gn}^{\rm ext}=Q_{A|J}\,,\quad Q_{\es|\gn}^{\rm ext}=Q_{\es|\es}^{\rm ext}=1\,,
\ee
supplemented with the requirement that the Q-functions that do not contain $\gn$ are fixed by consistency of the QQ-relations. An example of this extension is depicted in Figure~\ref{hasseproof} using a Hasse diagram \cite{Tsuboi:2009ud}.

\definecolor{amazon}{rgb}{0.23,0.78,0.34}
\begin{figure}[t]
\centering 
\begin{tikzpicture}
  \node (max) at (0,4) {$Q_{\fs}=Q_{\fs|\fs}^{\rm ext}$};
  \node (max2) at (0.1,3.76) {};
  \node (a) at (-2,2) {\hspace{-2em}$Q_1=Q_{1|1}^{\rm ext}$};
  \node (b) at (0,2) {$Q_2=Q_{2|1}^{\rm ext}$};
  \node (b2) at (0.1,2.24) {};
  \node (c) at (2,2) {\hspace{2em}$Q_{12|\es}^{\rm ext}$};
  \node (d) at (-2,0) {$Q_{\es}=Q_{\es|1}^{\rm ext}=1$};
  \node (e) at (0,0) {$Q_{1|\es}^{\rm ext}$};
  \node (f) at (2,0) {$Q_{2|\es}^{\rm ext}$};
  \node (min) at (0,-2) {$Q_{\es|\es}^{\rm ext}=1$};
  \thicklines
  \draw[line width=0.35mm,blue] (max) -- (a) -- (d) -- (b) -- (max);
  \draw[line width=0.35mm,amazon] (max2) -- (b2);
  \draw[line width=0.35mm,red] (c) -- (f) -- (min) -- (e);
  \draw[line width=0.35mm,dotted] (max) -- (c);
  \draw[line width=0.35mm,dotted] (b) -- (f);
  \draw[line width=0.35mm,dotted] (d) -- (min);
  \draw[line width=0.35mm,preaction={draw=white, -,line width=6pt},dotted] (a) -- (e);
  \draw[line width=0.35mm,preaction={draw=white, -,line width=6pt},red] (e) -- (c);
\end{tikzpicture}
\caption{Hasse diagram for extension of the $\gls{2}$ Q-system (blue square) to the $\gls{2|1}$ Q-system (blue and red squares). One considers nodes of the blue square as known, supplements this data with $Q_{\es|\es}^{\rm ext}= 1$ condition, and finds the rest by QQ-relations. Note that if $Q_2=1$ then $Q_{2|\es}^{\rm ext}=0$. In this case, $Q_{\fs}$ contains all physical information and we can restrict both the $\gls{2}$ and the $\gls{2|1}$ systems to the $\gls{1}$ system (green line) that consists of $Q_{\fs}^{\rm rest}=Q_{\fs}=Q_{\fs|\fs}^{\rm ext}$ and $Q_{\es}^{\rm rest}=Q_2=Q_{2|1}^{\rm ext}=1$.}
\label{hasseproof}
\end{figure}

\begin{lemma}
The extension  \eqref{eq:ext1} is always possible and moreover it defines all the Q-functions $Q_{A|J}^{\rm ext}$ uniquely up to symmetries.
\end{lemma}
\begin{proof} Let us find  $Q_{a|\es}^{\rm ext}, Q_{\es|j}^{\rm ext}, Q_{a|j}^{\rm ext}$ by the prescribed construction. Since $Q_{\es|\gn}=1$, one computes $Q_{a|\es}^{\rm ext}=(Q_{a|\gn}^{\rm ext})^+-(Q_{a|\gn}^{\rm ext})^-$ using the known value of $Q_{a|\gn}^{\rm ext}=Q_{a|\es}$. To find $Q_{\es|j}^{\rm ext}$  one solves $(Q_{\es|j}^{\rm ext})^+-(Q_{\es|j}^{\rm ext})^-=Q_{\es|j\gn}^{\rm ext}$. The polynomial solution is fixed up to an additive constant, but we remind that $Q_{\es|j}^{\rm ext}\to Q_{\es|j}^{\rm ext}+\alpha\, Q_{\es|\gn}^{\rm ext}$ is a symmetry of the twist-less Q-systems. Finally, we use $(Q_{a|j}^{\rm ext})^-(Q_{\es|\gn}^{\rm ext})-(Q_{a|\gn}^{\rm ext})^-Q_{\es|j}^{\rm ext}=Q_{a|j\gn}^{\rm ext}$ which is a consequence of \eqref{QQrelations} to uniquely fix $Q_{a|j}^{\rm ext}$ from the already identified quantities. One can now check that the above-constructed Q's satisfy $(Q_{a|j}^{\rm ext})^+-(Q_{a|j}^{\rm ext})^-=Q_{a|\es}^{\rm ext}Q_{\es|j}^{\rm ext}$ and so they properly generate the whole Q-system. 
\end{proof}

A  caveat of the extension procedure is that  $Q_{a|\es}^{\rm ext}=0$ if $Q_{a|\gn}^{\rm ext}$ is a constant. This does not happen however if both points $(\gm,\gn-1)$ and $(\gm,\gn)$ belong to the boundary of the Young diagram.

The analogous definitions and statements can be made also for the extension from the $\gls{\gm-1|\gn}$ to the $\glmn$ Q-system.

\subsection{Isomorphism of twist-less \texorpdfstring{$\BAL$}{B Lambda} across \texorpdfstring{$\glmn$}{gl(m|n)} algebras of various ranks}

An important conclusion from the made observations is that if $\gls{\gm'|\gn'}\subset \gls{\gm|\gn}$ and both points $(\gm',\gn')$ and $(\gm,\gn)$ belong to the boundary of the Young diagram then the restriction and extension procedures of the Q-system are inverse of one another and hence both Q-systems contain precisely the same physical information.

By performing a sequence of restrictions and extensions, we conclude that all the $\gls{\gm'|\gn'}$ Q-systems with $(\gm',\gn')$ being on the boundary  of $\LD$ (black and red dots in Figure~\ref{fig:ShortLong}) are bijectively related. Moreover, this is done by using only polynomial operations. Hence we conclude the following.
\begin{lemma}
\label{thm:Baiso}
In the twist-less case, all the Bethe algebras $\BAL$ generated by the $c_\pa$ of $\gls{\gm|\gn}$ Q-systems, where $(\gm,\gn)$ is any point on the boundary of the Young diagram $\LD$, are isomorphic as $\CC[\sse]$-algebras.\qed
\end{lemma}
We note that arguments leading to this conclusion do not require isomorphism of the Wronskian and Bethe algebra.

\begin{example}
The $\gls{3}$ and $\gls{1|1}$ Q-systems for representation $\twoone$ from the previous examples are related by the prescribed procedure:
\be
\label{reproc}
\gls{3|0} \xrightarrow{\rm rest} \gls{2|0} \xrightarrow{\rm ext} \gls{2|1} \xrightarrow{\rm rest} \gls{1|1}\,,
\ee
the extension bit of which is outlined in Figure~\ref{hasseproof}.

Hence $Q_1^{\gls{3}}$, $Q_2^{\gls{3}}$ (and $Q_3^{\gls{3}}=1$) should contain the same physical information as $Q_{1|\es}^{\gls{1|1}},Q_{\es|1}^{\gls{1|1}}$ (and $Q_{1|1}^{\gls{1|1}}=Q_{\theta}$) and expressible through one another.

By performing the restriction and extension transformations as is described above one finds
\begin{subequations}
\label{311}
\be
Q_2^{\gls{3}}&\propto&\Psi(Q_{\es|1}^{\gls{1|1}})\,,\\
Q_1^{\gls{3}}&\propto&\Psi\left(\frac{W(Q_{1|1}^{\gls{1|1}},\Psi(Q_{1|\es}^{\gls{1|1}}))}{Q_{1|\es}^{\gls{1|1}}}\right)\label{eq:Q1sugl}
\,,
\ee
\end{subequations}
where $\Psi(A)=B$ if $B^+-B^-=A$ and the ratio in \eqref{eq:Q1sugl} is a polynomial (Bethe equations ensure that the Euclidean remainder of the numerator and the denominator is the zero polynomial). $\Psi$ is defined up to addition of a constant, there are three such constants in \eqref{311} which corresponds to three symmetries in the $\gls{3}$ system restricted to $V_{\twoone}$. We fixed the symmetry in \eqref{Qsgl3} by setting certain $c_i^{(k)}$ to zero, and we should set the constants of integration to the values that reproduce this choice.

Explicitly in terms of  $c$'s, the transformation \eqref{311} becomes
\begin{subequations}
\label{c311}
\begin{align}
c_{1}^{(3)} &=4c_{1|\es}\,,\\
c_{1}^{(1)} &=(2 c_{\es|1}-4c_{1|\es})\hbar^2-8c_{1|1}\,,\\
c_{2}^{(1)} &=2 c_{\es|1}\,.
\end{align}
\end{subequations}
By reversing \eqref{reproc}, we find that
\begin{subequations}
\label{113}
\begin{align}
Q_{\es|1}^{\gls{1|1}}&\propto Q_{23}^{\gls{3}}\,,\\
Q_{1|\es}^{\gls{1|1}}&\propto W(Q_{1}^{\gls{3}},Q_{2}^{\gls{3}},Q_{3}^{\gls{3}},u)
                      \,,\\
Q_{1|1}^{\gls{1|1}}&\propto Q_{123}^{\gls{3}}\,.
\end{align}
\end{subequations}
In terms of $c$'s, this becomes
\begin{subequations}
\label{c113}
\begin{align}
c_{\es|1} &=\frac{c_2^{(1)}}2\,,\\
c_{1|\es} &=\frac{c_1^{(3)}}4\,,\\
c_{1|1}    &=\frac{(c_2^{(1)}-c_1^{(3)})\hbar^2-c_1^{(1)}}8\,.
\end{align}
\end{subequations}
Transformation \eqref{c113} is the inverse of \eqref{c311}. %
\end{example}

\paragraph{Remark} In general, the isomorphism between Bethe algebras for different $(\gm,\gn)$ is a non-linear polynomial map in $c_\pa$, and it is remarkable that its inverse is also polynomial.

\vspace{.5cm}

By Lemma~\ref{thm:Baiso}, we can use the Q-system which corresponds to $\gls{\gm=h_{\LD}|\gn=0}$, where $h_{\LD}$ is the height of the Young diagram, and is therefore purely bosonic. Hence we can use in principle results from \cite{MTV} for bosonic Bethe algebras to formulate and prove the completeness and faithfulness statements for supersymmetric spin chains~\footnote{It is also possible to establish a bijection between bosonic and supersymmetric Q-systems in the presence of twist. To this end one first extends the original Q-system to a larger one with a partially degenerate twist where an analog of the Young diagram boundary and hence a possibility to move along it emerges.}, but we choose to not rely on this relation to \cite{MTV}, highlight instead novel important features of the system, and to produce results in a way that does not use the nested Bethe Ansatz.

\subsection{Q-system on Young diagrams}
\label{sec:reconstr-full-q}
One of the interesting features is that the Bethe algebra can be also generated from the so-called Q-system on a Young diagram which was introduced in \cite{Marboe:2016yyn}, extensively used in application to the AdS/CFT spectral problem \cite{Marboe:2017dmb,Marboe:2018ugv} and more recently in several other studies, see \eg \cite{Jacobsen:2018pjt,Bajnok:2019zub}.

The Q-system on a Young diagram is a collection of monic polynomials $\wQ_{a,s}$ defined as 
\be\label{wQasdef}
\wQ_{a,s} \propto Q_{a+1,a+2,\ldots,\gm|s+1,s+2,\ldots,\gn}\,,
\ee
where the Q-function on the \rhs is a member of the $\glmn$ Q-system and $(\gm,\gn)$ belongs to the boundary of the Young diagram. It is clear from the extension-restriction procedure that $\wQ_{a,s}$ does not depend on $\gm,\gn$. In particular, $\wQ_{0,0}=Q_{\theta}$. 

 $\wQ_{a,s}$ are naturally assigned to the nodes of the $\mathbb{Z}^2$-lattice which belong to the Young diagram shape, $\wQ_{a,s}=1$ on the boundary of the diagram (black/red dots in Figure~\ref{fig:ShortLong}). The QQ-relations between $\wQ_{a,s}$ are
\be
\label{eq:58}
\wQ_{a+1,s+1}\wQ_{a,s} \propto W(\wQ_{a+1,s},\wQ_{a,s+1})\,,
\ee
this follows from \eqref{QQferm}.
\begin{example}
For the Young diagram $\twoone$, $\wQ_{a,s}$ that are not equal to $1$ are $\wQ_{1,0}=Q_{\es|1}^{\gls{1|1}}=Q_2^{\gls{2}}=Q_{23}^{\gls{3}}$, $\wQ_{0,1}=Q_{1|\es}^{\gls{1|1}}$, and $\wQ_{0,0}=Q_{\theta}$.  There is only one non-trivial relation $\wQ_{0,1}\wQ_{1,0}=\wQ_{0,0}^+-\wQ_{0,0}^-$, so there is no significant difference between this simple Q-system and the above-discussed example of a $\gls{1|1}$ Q-system. Systems for larger Young diagrams are more interesting of course, there are a plenty of examples in \cite{Marboe:2017dmb,Marboe:2018ugv} and we give one explicit example on page \pageref{exmp:p44}.
\end{example}

The Q-system on a Young diagram $\LD$ is polynomially generated from the Bethe algebra $\BAL$ as it follows directly from \eqref{wQasdef}. The converse is also true  (the polynomiality lemma of \cite{Marboe:2016yyn}):
\begin{lemma}
\label{thm:pollemma}
Consider the Q-system on a Young Diagram $\LD$, and choose any $\gm,\gn$ that lie on the boundary of $\LD$. Then, up to symmetry, one can uniquely construct a solution of the $\glmn$ QQ-relations such that \eqref{wQasdef} holds and that the coefficients of the $Q_{A|J}$ are polynomial functions of the coefficients of the $\wQ_{a,s}$.
\end{lemma}
So the Q-system on a Young diagram is yet another description of the Bethe algebra restricted to $\VTl$. We shall benefit from it for counting purposes.

We will now prove Lemma~\ref{thm:pollemma} in a bit different way compared to \cite{Marboe:2016yyn}. The techniques introduced in this  proof are used in Appendix~\ref{app:BetheDetails}.
\begin{proof}
It suffices to find all $Q_{a|\es}^{\gl_\gm}$ for the $\gl_{\gm|0}$ Q-system, where $\gm=h_{\LD}$ is the height of the Young diagram. Then we can use extensions and restrictions to get all other $\gl_{\gm'|\gn'}$ Q-systems. To this end, let us extend $\gl_{\gm|0}$ to the $\gl_{\gm|\gn}$ Q-system, where $\gn=\lambda_1$ is the width of the Young diagram. Upon extension $Q_{a|\es}^{\glm}=Q_{a|\fs}^{\glmn}$. If $(\gm,\gn)$ lies outside of the Young diagram boundary, the extended Q-system contains non-physical Q-functions making the extension procedure non-unique but this will not induce an ambiguity in fixing $Q_{a|\es}^{\glm}$.

In the following, the superscript $^{\glmn}$ will be omitted.

One has $Q_{\es|\fs}=\wQ_{\gm,0}=1$ which is very suitable for applying
the so-called bosonisation trick \cite{Gromov:2010km}. The bosonisation trick is the observation that the Q-functions $B_{AJ}:= Q_{A|\bar J}$, where $\bar J$ means the complementary set to $J$, satisfy the QQ-relations of the bosonic $\gls{\gn+\gm}$ system:
\be\label{Brelations}
B_{\Sigma\sA\sB}B_{\Sigma}\propto W(B_{\Sigma\sA},B_{\Sigma\sB})\,,
\ee
where $\Sigma$ is a multi-index and $\sA,\sB$ are indices from the set $\{1,\ldots,\gm,\hat 1,\ldots,\hat \gn\}$. This immediately follows from \eqref{QQrelations} and the definition of $B$.

Since $B_{\es}=Q_{\es|\fs}=1$, the relations \eqref{Brelations} are solved by $B_{\Sigma}=W(B_{\sA_1},B_{\sA_2},\ldots,B_{\sA_k})$, where $\Sigma=\{\sA_1,\sA_2,\ldots,\sA_k\}$ and so \eqref{wQasdef} becomes 
\begin{equation}
  \label{eq:reconsQ}
  \wQ_{a,s}\propto W(B_{a+1},\ldots,B_{\gm},B_{\hat 1},\ldots,B_{\hat s})\,.
 \end{equation}
We are going to solve these equations for $B_\sA$ in a unique way modulo admissible symmetry transformations. Note that our main interest is $B_{a}=Q_{a|\es}^{\glm}$.

From \eqref{eq:221}, $\deg B_{a}=\deg Q_{a|\es}^{\glm}=\IPart_{a}+\gm-a$. Furthermore, remark the following property of the Wronskian determinant: if  $W(P_1,\ldots,P_k)\propto 1$ for some polynomials $P_1,\ldots,P_k$  then, modulo permutations, $\deg P_r = r-1$ for $r=1,\ldots,k$. And so, by examining \eqref{eq:reconsQ} for $(a,s)$ being on the boundary of Young diagram, where $\wQ_{a,s}=1$, we conclude that all $B_{\sA}$ for $\sA=1,\ldots,\gn+\gm$ should have distinct degrees from $0$ till $\gm+\gn-1$. The degrees $d_a\equiv \deg B_a$'s  satisfy $d_1\le d_2\le \dots\le d_\gm$. Their assignment rule is explained in Figure~\ref{fig:boso}.
\begin{figure}
\centering
\includegraphics[width=8cm]{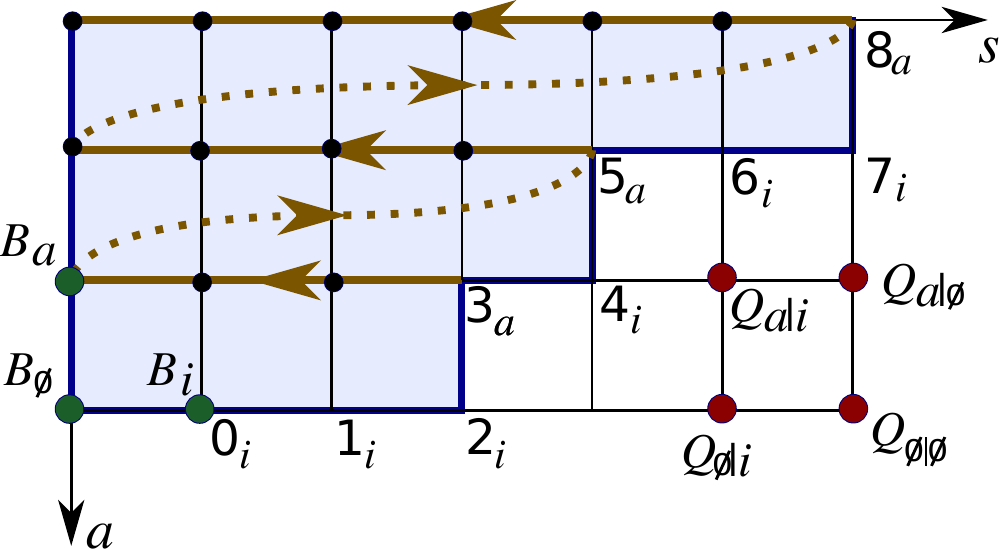}
\caption{\label{fig:boso}\scriptsize Rectangular lattice of size $\gm\times\gn$ represents the $\glmn$ Q-system used in the proof of Lemma~\ref{thm:pollemma}. We associate the collection of all $Q_{A|J}$ with $|A|=\gm-a$, $|J|=\gn-s$ to the node (a,s) of the lattice. When $(a,s)$ is on the Young diagram $\LD$, $\wQ_{a,s}$ is the smallest degree polynomial among $Q_{A|J}$. The Q-system is generated by $Q_{a|\es},Q_{\es|i},Q_{a|i}$ situated at the down-right corner or, using bosonisation, by $B_i=Q_{\es|\bar i}$ and $B_a=Q_{a|\fs}$ situated in the down-left corner. Degrees of polynomials $B_a,B_i$ are given by the Manhattan distance from $(\gm,1)$ to the appropriate points on the boundary of the Young diagram. For instance, $6_i$, the fifth number with subscript $i$, means that $\deg B_{\hat 5}=6$. Correspondingly $8_a$ means $\deg B_3=8$.}
\end{figure}

We can actually set $B_i=u^{d_i}$ because any subleading orders in polynomials $B_i$ do not affect physically relevant Q-functions (\ie those that survive restrictions that make short representations long, see page~\pageref{sec:shortrep}), in particular $\wQ_{a,s}$. Also, using \eqref{eq:225}, we can restrict $B_a$ to the form~\footnote{If we assign grading to $c_{a,s}$ as is done in Section~\ref{sec:Hilbert} then $\deg c_{a,s}$ is equal to $h_{a+1,s+1}$ -- the hook length, \emph{cf.}~Figure~\ref{fig:degrees}.}
\be
B_a=u^{\IPart_{a}+\gm-a}+\sum_{s=1}^{\IPart_{a}}c_{a-1,s-1}\,u^{d_s}\,.
\ee
Finally, we recursively fix coefficients $c_{a,s}$ following the serpentine path in Figure~\ref{fig:boso}. The point $(a,s)$ on the path is used to fix $c_{a,s}$. Any other $c_{a',s'}$ present in the equation \eqref{eq:reconsQ} are fixed from the previous recursion steps.  $c_{a,s}$ appears only linearly in the \rhs of \eqref{eq:reconsQ} and is multiplied by the {\it non-zero} prefactor $W(u^{d_{s+1}},B_{a+2},\ldots,B_\gm,u^{d_1},\ldots,u^{d_s})=g\,\wQ_{a+1,s+1}$, where $g$ is a non-zero constant, and so solution is unique. To confirm that it exists, recall that equation \eqref{eq:58} has a solution by assumptions of the lemma. But an arbitrary solution to this equation forms a one-parametric space $\wQ_{a,s}=\wQ_{a,s}^{(0)}+\tilde c_{a,s} \wQ_{a+1,s+1}$ parameterised by $\tilde c_{a,s}$. We can set $\wQ_{a,s}^{(0)}=W(B_{a+1}-c_{a,s}u^{d_{s+1}},B_{a+2},\ldots,B_{\hat s})$ and then $\tilde c_{a,s}=g\,c_{a,s}$.
\end{proof}
We remark that $B_a$ can be fixed uniquely up to symmetries from \eqref{eq:reconsQ} by considering only $a=0,\ldots,\gn-1$ and $s=0$, but the above proof asserts that the solution is polynomial if only if the Q-system on Young diagram has a polynomial solution. What is typically redundant is the number of equations \eqref{eq:58} needed to ensure polynomiality of the Q-system.  A conjecture about what is the minimal set of equations needed was given in \cite{Marboe:2016yyn}. For $\gls{2}$ spin chains, it was proven \cite{Granet:2019knz} that taking four equations with $a=0,1$, $s=0,1$ exactly suffices if $\inhom=0$.

\subsection{Relation to nested Bethe equations and quantum eigenvalues}
\label{sec:NBAE}
This topic was discussed numerously in the literature, it started to develop with the formulation of the analytic Bethe Ansatz \cite{Baxter:1971cs,Reshetikhin1983} and received a more covariant point of view after \cite{Krichever:1996qd}. It is worth emphasising \cite{Kazakov:2007fy} where the connection between Q-functions and Bethe equations \via the choice of different Kac-Dynkin paths was elucidated. In this section we summarise the known results and complement them with a discussion that focuses on completeness questions, see also Appendix~C of~\cite{Kazakov:2015efa}.

For a $\glmn$ spin chain, choose any permutation of the sequence $12\ldots\gm\hat 1\ldots\hat\gn$. It shall be called a choice of the nesting path. Define by ${\leftarrow\!k}$ the sequence of the first $k$ letters from the nesting path. For instance, if we chose  $2\hat 1\hat 2 1$ for the $\gls{2|2}$ case, then $\leftarrow\!1=2$, $\leftarrow\!2=2\hat 1$, $\leftarrow\!3=2\hat 1\hat2$, $\leftarrow\!4=2\hat 1\hat 21$. 

Nested Bethe equations are the equations on zeros of $Q_{\leftarrow\!k}$, $k=1,\ldots,\gm+\gn-1$, where we mean \eg $Q_{\leftarrow\! 3}=Q_{2|\hat 1\hat 2}$. Note also that, independently of the path choice, $Q_{\leftarrow\!(\gm+\gn)}=Q_{\theta}$ -- this is the fixed Q-function which plays the role of the source term.

The equations are derived as follows. Let $a,b$ denote some indices from the set $\{1,\ldots,\gm\}$, and $i,j$-some indices from the set  $\{\hat 1,\ldots,\hat \gn\}$. Then \eqref{QQrelations} imply
\begin{subequations}
\label{QQk}
\begin{align}
\label{QQka}
Q_{\leftarrow\!(k+1)}Q_{\leftarrow\!(k-1)}&\propto W(Q_{\leftarrow\!k},Q_{\leftarrow\!(k-1)b})\,,&&\text{\hspace{0.25em} for}\leftarrow\!(k+1)=\leftarrow\!(k-1)ab\,,
\\
\label{QQkb}
Q_{\leftarrow\!(k-1)a}Q_{\leftarrow\!(k-1)j}&\propto W(Q_{\leftarrow\!(k+1)},Q_{\leftarrow\!(k-1)})\,,&& \begin{tabular}{ll} \text{for $\leftarrow\!(k+1)=\leftarrow\!(k-1)aj$}\\ \text{\ or $\leftarrow\!(k+1)=\leftarrow\!(k-1)ja$}\end{tabular}\,,
\\
\label{QQkc}
Q_{\leftarrow\!(k+1)}Q_{\leftarrow\!(k-1)}&\propto W(Q_{\leftarrow\!k},Q_{\leftarrow\!(k-1)j})\,,&&\text{\hspace{0.25em} for}\leftarrow\!(k+1)=\leftarrow\!(k-1)ij\,.
\end{align}
\end{subequations}
Take \eqref{QQka} (or \eqref{QQkc}), make a shift $u\to u+\frac{\hbar}{2}$ and evaluate it at a zero of $Q_{\leftarrow\!k}$. Do the same for $u\to u-\frac{\hbar}{2}$ and divide the two evaluated equations by one another. One gets ``bosonic'' nested Bethe equations
\begin{subequations}
\label{BAEgeneric}
\be
\frac{Q_{\leftarrow\!(k+1)}^+Q_{\leftarrow\!(k-1)}^+}{Q_{\leftarrow\!(k+1)}^-Q_{\leftarrow\!(k-1)}^-}=-\frac{Q_{\leftarrow\!k}^{++}}{Q_{\leftarrow\!k}^{--}}\quad\text{at zeros of } Q_{\leftarrow\!k}\,.
\ee
By evaluating \eqref{QQkb} at zeros of $Q_{\leftarrow\!k}$, one gets $W(Q_{\leftarrow\!k},Q_{\leftarrow\!(k-1)j})=0$ which can be written in the ``fermionic'' nested Bethe equations form
\be
\frac{Q_{\leftarrow\!(k+1)}^+Q_{\leftarrow\!(k-1)}^-}{Q_{\leftarrow\!(k+1)}^-Q_{\leftarrow\!(k-1)}^+}=1\quad\text{at zeros of } Q_{\leftarrow\!k}\,.
\ee
\end{subequations}
For $\leftarrow\!k=\leftarrow\!(k-1)\sA_k$, we remind that $Q$-functions are actually (twisted) polynomials and denote them as
$Q_{\leftarrow\!k}\propto\prod\limits_{\sB\in \leftarrow\!k}z_\sB^{(-1)^{\bar\sB} u/\hbar}\prod\limits_{l=1}^{M_{k}}(u-u_l^{(k)})$. Equations (\ref{BAEgeneric}) read
\begin{equation}
\frac{z_{\sA_k}}{z_{\sA_{k+1}}}\prod\limits_{(k',l')\neq(k,l)} \frac{u_l^{(k)}-u_{l'}^{(k')}+\frac{\hbar}2 c_{k,k'}}{u_l^{(k)}-u_{l'}^{(k')}-\frac{\hbar}2 c_{k,k'}}=1\label{eq:WBEgeneric}\,,
\end{equation}
where $k'$ runs from $0$ to $\gm+\gn$ (whereas $1\le k \le \gm+\gn-1$), and
the matrix $c_{k,k'}$ has coefficients $c_{k,k}=(-1)^{\bar\alpha_k}+(-1)^{\bar\alpha_{k+1}}$, $c_{k,k+1}=c_{k+1,k}=-(-1)^{\bar\sA_{k+1}}$ and all other coefficients equal to zero. In the particular case $\sA_k=k$, if we rename the label ${(k)}$ as ${(\sA)}$ where $\sA:=\gm+\gn-k$ and we remember that $Q_{\leftarrow\!0}=1$ and $Q_{\leftarrow\!(\gm+\gn)}\propto Q_{\theta}$, then equation~(\ref{BAEgeneric}) becomes precisely (\ref{eq:Bethedistinguished})~\footnote{
The matrix $c_{\sA,\sB}$ of (\ref{eq:Bethedistinguished}) is obtained from the matrix $c_{k,k'}$ by restricting to $1\le \sA,\sB< \gm+\gn$ and re-ordering of the rows and columns.}.

In particular we see that the choice of the Cartan matrix, or equivalently of the Kac-Dynkin diagram, see \eg \cite{Frappat:1996pb}, is implied by the choice of the nesting path \cite{Kazakov:2007fy}. Namely, the Kac-Dynkin diagram should be a chain of $\gm+\gn-1$ nodes where the $k$'th node is fermionic (crossed) if the $k$'th and the $(k+1)$'th letters have different grading and is bosonic (blank) otherwise.

For twist-less systems, two comments are due. First, a twist-less Q-system is invariant under symmetry transformations \eqref{eq:225} and this ambiguity can propagate to the nested Bethe equations if the nesting path is generic. To avoid this happening, we restrict ourselves only to those paths for which $a$ is to the left from $b$ if $a>b$, and $i$ is to the left of $j$ if $i>j$. Such a choice ensures that if $Q_{A|J}=Q_{\leftarrow\!k}$ for some $k$ then $Q_{A|J}$ is the polynomial of the smallest degree among all $Q_{A'|J'}$ with $|A'|=|A|,|J'|=|J|$. Hence the distinguished subclass of the nesting paths is naturally realised by the paths across Young diagrams, with $\wQ_{a,s}=Q_{\leftarrow\!k}$, $a=\gm-|A|,s=\gn-|J|$, and so we can reformulate \eqref{BAEgeneric} using Q-functions from the Young diagram Q-system.  Second, for short representations, a part of the nesting path lies outside of the Young diagram. We should define $\wQ_{a,s}=Q_{\leftarrow\!k}=1$ if $(a,s)$ is on the path but outside of the Young diagram to get the correct interpretation in terms of the nested Bethe equations.

\noindent\begin{minipage}{\textwidth}
\begin{example}
\label{exmp:p44}
\begin{minipage}{0.4\textwidth}
\begin{center}
\begin{picture}(160,80)(0,-80)
\thinlines
\color{blue!5}
\polygon*(0,0)(160,0)(160,-20)(120,-20)(120,-40)(40,-40)(40,-80)(0,-80)(0,0)
\color{black}
\thinlines
\multiput(0,0)(0,-20){5}{\line(1,0){160}}
\multiput(0,0)(20,0){9}{\line(0,-1){80}}

\color{blue}
\thicklines
\drawline(0,0)(160,0)(160,-20)(120,-20)(120,-40)(40,-40)(40,-80)(0,-80)(0,0)

\color{Black}
\thicklines
\drawline(0,0)(0,-80)
\drawline(0.4,0)(0.4,-80)
\drawline(-0.4,0)(-0.4,-80)

\color{white}
\put(0,-20){\circle*{7}}
\put(0,-40){\circle*{7}}
\put(0,-60){\circle*{7}}

\thinlines
\color{Blue}
\put(0,-20){\circle{7}}
\put(0,-40){\circle{7}}
\put(0,-60){\circle{7}}

\end{picture}
\end{center}

\scriptsize
The Q-system on the depicted Young diagram leads, by the choice of the Kac-Dynkin nesting path, to Bethe equations of $\gl_4$ spin chain of length $L=18$. The momentum-carrying Bethe roots are zeros of $\wQ_{1,0}$, the Bethe roots on the nested levels are those of $\wQ_{2,0}$ and $\wQ_{3,0}$.

\end{minipage}
\hspace{0.1\textwidth}
\begin{minipage}{0.4\textwidth}
{
\begin{picture}(160,80)(0,-80)
\thinlines
\color{blue!5}
\polygon*(0,0)(160,0)(160,-20)(120,-20)(120,-40)(40,-40)(40,-80)(0,-80)(0,0)
\color{black}
\thinlines
\multiput(0,0)(0,-20){5}{\line(1,0){160}}
\multiput(0,0)(20,0){9}{\line(0,-1){80}}

\color{blue}
\thicklines
\drawline(0,0)(160,0)(160,-20)(120,-20)(120,-40)(40,-40)(40,-80)(0,-80)(0,0)

\color{Black}
\thicklines
\drawline(0,0)(20,0)(20,-20)(40,-20)(40,-40)(60,-40)(60,-60)(160,-60)
\put(0.4,0.4){\drawline(0,0)(20,0)(20,-20)(40,-20)(40,-40)(60,-40)(60,-60)(160,-60)}
\put(-0.4,-0.4){\drawline(0,0)(20,0)(20,-20)(40,-20)(40,-40)(60,-40)(60,-60)(160,-60)}

\color{white}
\put(20,0){\circle*{7}}
\put(20,-20){\circle*{7}}
\put(40,-20){\circle*{7}}
\put(40,-40){\circle*{7}}
\put(60,-40){\circle*{7}}
\put(60,-60){\circle*{7}}
\put(80,-60){\circle*{7}}
\put(100,-60){\circle*{7}}
\put(120,-60){\circle*{7}}
\put(140,-60){\circle*{7}}

\thinlines
\color{Blue}
\put(20,0){\circle{7}}
\put(20,-20){\circle{7}}
\put(40,-20){\circle{7}}
\put(40,-40){\circle{7}}
\put(60,-40){\circle{7}}
\put(60,-60){\circle{7}}
\put(80,-60){\circle{7}}
\put(100,-60){\circle{7}}
\put(120,-60){\circle{7}}
\put(140,-60){\circle{7}}

\put(20,0){
\put(-2.6,-2.6){\line(1,1){5.2}}
\put(-2.6,2.6){\line(1,-1){5.2}}
}
\put(20,-20){
\put(-2.6,-2.6){\line(1,1){5.2}}
\put(-2.6,2.6){\line(1,-1){5.2}}
}
\put(40,-20){
\put(-2.6,-2.6){\line(1,1){5.2}}
\put(-2.6,2.6){\line(1,-1){5.2}}
}
\put(40,-40){
\put(-2.6,-2.6){\line(1,1){5.2}}
\put(-2.6,2.6){\line(1,-1){5.2}}
}
\put(60,-40){
\put(-2.6,-2.6){\line(1,1){5.2}}
\put(-2.6,2.6){\line(1,-1){5.2}}
}
\put(60,-60){
\put(-2.6,-2.6){\line(1,1){5.2}}
\put(-2.6,2.6){\line(1,-1){5.2}}
}
\end{picture}
}

\scriptsize
The same Q-system but a different path leading to Bethe equations of $\gl_{3|8}$ spin chain of length $L=18$. The momentum-carrying Bethe roots are zeros of $\wQ_{0,1}$, the Bethe roots on nested levels are those of $\wQ_{1,1}$ and $\wQ_{1,2}$, all other nested levels are not excited, and so these Bethe equations are also those of the $\gl_{2|2}$ chain.
\end{minipage}
\newline
\newline
\scriptsize
Bethe equations are written as
\begin{align}
\prod_{\pa=1}^L\frac{u_k^{(\sA)}-\inhom[\pa]+\frac{c_{1,2}+c_{1,1}}{2}\,\hbar\,\delta_{\sA,1}}{u_k^{(\sA)}-\inhom[\pa]-\frac{c_{1,2}+c_{1,1}}{2}\,\hbar\,\delta_{\sA,1}}
&=(-1)^{\frac{c_{kk}}2}
\prod_{\sB=1}^{3}\prod_{l=1}^{M_\beta}
\frac{u_{k}^{(\sA)}-u_{l}^{(\sB)}+\frac \hbar 2c_{\sA,\sB}}{u_{k}^{(\sA)}-u_{l}^{(\sB)}-\frac \hbar 2 c_{\sA,\sB}}\,,&& 1\le \sA\le 3\,,&&
1\le k\le M_\sA\,;
\end{align}
\begin{minipage}{0.4\textwidth}
\centering
$
c_{\sA,\sB}=\left(\begin{matrix} 2& -1 & 0 \\ -1 & 2 & -1 \\ 0 & -1 & 2  \end{matrix}\right)\,,
$
$
\begin{array}{rcrl}
\wQ_{1,0} &=\prod\limits_{i=1}^{M_1}(u-u_i^{(1)}),& M_1&=10\,,
\\
\wQ_{2,0} &=\prod\limits_{i=1}^{M_2}(u-u_i^{(2)}),& M_2&=4\,,
\\
\wQ_{3,0} &=\prod\limits_{i=1}^{M_3}(u-u_i^{(3)}),& M_3&=2\,.
\end{array}
$
\end{minipage}
\hspace{0.1\textwidth}
\begin{minipage}{0.4\textwidth}
\centering
$
c_{\sA,\sB}=\left(\begin{matrix} 0& -1 & 0 \\ -1 & 0 & 1 \\ 0 & 1 & 0  \end{matrix}\right)\,,
$
$
\begin{array}{rcrl}
\wQ_{0,1} &=\prod\limits_{i=1}^{M_1}(u-u_i^{(1)}),& M_1&=14\,,
\\
\wQ_{1,1} &=\prod\limits_{i=1}^{M_2}(u-u_i^{(2)}),& M_2&=7\,,
\\
\wQ_{1,2} &=\prod\limits_{i=1}^{M_3}(u-u_i^{(3)}),& M_3&=4\,.
\end{array}
$
\end{minipage}
\end{example}
\end{minipage}

It is possible to perform a duality transformation \cite{Woynarovich_1983,PhysRevB.46.14624,Tsuboi:1998ne,Pronko:1998xa,Mukhin2003,Gromov:2007ky} -- to pass from the nested Bethe equations for one nesting path to the equations for another path, typically when the change of path is an elementary permutation. It can be done by lifting \eqref{BAEgeneric} to the QQ-relations and then by descending to another path.

Concerning completeness,  we start by commenting on the relation between \eqref{BAEgeneric} and \eqref{QQk}. The Bethe equations \eqref{BAEgeneric} involve ratios which can become of type $0/0$ for certain class of solutions known as exceptional solutions, see \eg  type {\it  a} in \cite{Avdeev1986}. We should provide some regularisation prescription to treat them properly. Furthermore, if $Q_{\leftarrow\!k}$ contains a double zero, \ie  coinciding Bethe roots, we are losing information when passing from \eqref{QQk} to \eqref{BAEgeneric}. Indeed, we should consider also a derivative of \eqref{QQk} at a double zero which provides an extra constraint in addition to the nested Bethe equations. Double zeros can indeed exist as physical solutions since we can always collide roots by fine-tuning values of inhomogeneities or twist \cite{Volin:2010xz}. Moreover, the fine-tuned points coincide sometimes with the physically-relevant case of $\inhom[\pa]=0$ \cite{Avdeev1986,Hao:2013rza}~\footnote{the observed cases are however for $\gl_2$ chains in a higher spin representation}.

Based on the above comments, \eqref{QQk} look as more appropriate equations than \eqref{BAEgeneric}. The QQ-relations \eqref{QQk} along a nesting path are closely related to the description in terms of  quantum eigenvalues \footnote{The name, to our knowledge, first time appears in the work of Sklyanin \cite{Sklyanin:1992sm}. Historically, the name ``analytic Bethe Ansatz'' (for transfer matrices as sum over $\Lambda$'s) was more often in use following the work of Reshetikhin \cite{Reshetikhin1983}.} \cite{Kuniba:1994na,Krichever:1996qd,Tsuboi:1997iq,Tsuboi:1998ne}. For clarity, we introduce them in the example of $\gl_{2|2}$ spin chain in a concrete grading choice - $2\hat 2\hat 1 1$. Consider a generating functional \cite{Kazakov:2007fy}
\be
\label{genfun} \sum_{a=0}^{\infty} (-1)^a\DD^a\,\wT_{(1^a)}(u)\,\DD^a =
\frac 1{1-\DD\Lambda_1\DD}(1-\DD\Lambda_{\hat 1}\DD)(1-\DD\Lambda_{\hat 2}\DD)\frac 1{1-\DD\Lambda_2\DD}
\ee
that can be viewed as a way to factorise $\Ber\left[\Id-\,\DD\,T(u)G\,\DD\right]$ (see \eqref{eq:genseries}), and thus $\Lambda_{\sA}$ are sometimes called quantum eigenvalues (of the monodromy matrix). $\Lambda_{\sA}$ commute between themselves and are expressed in terms of the $Q$-functions as
\begin{align}
\Lambda_2&=\frac{Q_2^{+}}{Q_2^-}Q_{\theta}\,,&\Lambda_{\hat 2}&=\frac{Q_{2|2}^{[-2]}}{Q_{2|2}}\frac{Q_2^{+}}{Q_2^{-}}Q_{\theta}\,,&\Lambda_{\hat 1}&=\frac{Q_{2|12}^-}{Q_{2|12}^{+}}\frac{Q_{2|2}^{[+2]}}{Q_{2|2}}Q_{\theta}\,,& \Lambda_1&=\frac{Q_{2|12}^-}{Q_{2|12}^{+}}\frac{Q_{12|12}^{[+2]}}{Q_{12|12}}Q_\theta\,.
\end{align}
The general rule is
\be
\Lambda_\sA=\frac{Q_{\leftarrow\!k}^{[\pm 2+\gm-\gn-|A|+|J|]}Q_{\leftarrow\!(k-1)}^{[\mp 2+\gm-\gn-|A'|+|J'|]}}{Q_{\leftarrow\!k}^{[\gm-\gn-|A|+|J|]}Q_{\leftarrow\!(k-1)}^{[\gm-\gn-|A'|+|J'|]}}Q_{\theta}\,,
\ee
where $\sA$ is on the $k$'th position of the nesting path, and $\leftarrow\!k=A|J$, $\leftarrow\!(k-1)=A'|J'$, the upper sign corresponds to $\bar\sA=0$ and the lower sign - to $\bar\sA=1$. While $\Lambda_{\sA}$ depend on the choice of the nesting path, the generating functional \eqref{genfun} does not which follows from the QQ-relations \eqref{QQrelations}~\footnote{In \cite{Huang_2019}, the same concepts and statements are expressed more formally. There, {\it population} is the same as Q-system reviewed on page~\pageref{page:Qsys} and onwards, and {\it reproduction procedure} is the same as the above-mentioned duality transformations.}.

Given that \eqref{genfun} generates transfer matrices, $\Lambda_2-\Lambda_{\hat 2}$ should not have poles at zeros of $Q_2^-$, $\Lambda_{\hat 1}+\Lambda_{\hat 2}$ should not have poles at zeros of $Q_{2|2}$ \etc ., these conditions are another way to generate nested Bethe equations, in the same spirit as they are derived from Baxter equation, see \eg   \cite{Frenkel:2013uda}.

Both QQ-relations along the path and no-poles conditions for combinations of quantum eigenvalues (which are also path-dependent) are less constraining than the Wronskian condition \eqref{mastereq}. What happens is that polynomiality should be ensured for {\it all} choices of paths, that is for all Q-functions. This requirement is achieved by \eqref{mastereq} or equivalent formulations. We remark that even if a solution of \eqref{QQk} looks normal (\ie it has no coinciding Bethe roots or roots separated by $\hbar$, $\frac 12\hbar$) it might be still not physical because of the problems with polynomiality happening when we try to change the path~\footnote{This was observed by C.~Marboe and one of the authors \cite{MarboeUnpub} while computing the AdS/CFT spectrum for \cite{Marboe:2017dmb,Marboe:2018ugv}. Curiously, attempts to mitigate this issue led to the formulation of the Q-system on a Young diagram \cite{Marboe:2016yyn}.}.

\section{Labelling solutions}
\label{sec:label}
In this section we solve Wronskian Bethe equations explicitly in special regimes, when all $\inhom$ are far away from one another. Then one can continuously deform $\inhom$ to any desired values thus obtaining a way to label solutions. A practical application of this labelling approach is demonstrated in Section~\ref{sec:alternative}. The labelling approach provides an alternative physics-style proof of the completeness.

\subsection{Twisted case, labelling with a multinomial expansion}
Consider the regime when $|\inhom-\inhom[\pa']|\sim \Lambda$, and $\Lambda$ is large. This limit appeared previously in the literature including for counting purposes, see \eg  \cite{MTV,Gromov:2019wmz}.

By rescaling
\be
\label{eq:rescaling}
u\to u/\Lambda,\quad \inhom\to\inhom/\Lambda,\quad \hbar\to\hbar/\Lambda\,,
\ee
which is a symmetry of the Wronskian equations,  we can consider the $\hbar\rightarrow 0$ limit with all $\inhom[\pa]$ being finite and distinct instead of $\Lambda\to\infty$. We will show below that when $\hbar\rightarrow 0$, we can neglect shifts of the spectral parameter in the polynomial piece of the Baxter function $Q=z^{-u/\hbar}q(u)$, $q(u+\hbar)\simeq q(u)+\CO(\hbar)$. So all the QQ-relations \eqref{QQrelations} become of type $QQ=QQ$ and equation \eqref{mastereq} simplifies to
\begin{equation}
\label{hbarzeroeq}
\prod_{a=1}^{\gn}q_{a|\es}(u)\prod_{i=1}^{\gm}q_{\es|i}(u)=\prod_{\pa=1}^L(u-\inhom)\,.
\end{equation}
The number of solutions to the last equation which is an equation on $c_{a|\es},c_{\es|i}$ is easily counted to be
\begin{equation}
\label{numofsol}
\frac{L!}{\prod_{a=1}^{\gn}\lambda_a ! \prod_{i=1}^{\gm} \nu_i !}=\mathrm{dim}~\VTw\,.
\end{equation}

The value $\hbar=0$ is quite special, therefore let us ensure that the conclusion about number of solutions holds also for $\hbar\neq 0$. 
\begin{lemma}
\label{thm:61}
For distinct $\inhom$, the number of solutions of \eqref{mastereq} in some neighbourhood of $\hbar=0$ is given by \eqref{numofsol} and the solutions are in one-to-one correspondence, by analytic continuation in $\hbar$, with solutions at $\hbar=0$.
\end{lemma}
\begin{proof}
Let us treat $\hbar$ as a parameter and consider $\SW$ as a smooth map from $\Cdom\times\CC$ to $\Sedom$. For our discussion, it will be safe to use $\inhom$ instead of $\sepa$ to (locally) parameterise $\Sedom$. Take a solution $\bar{c}$ of \eqref{hbarzeroeq} corresponding to a point $\binhom[]$ with pairwise distinct coordinates. It is easy to check that the differential of $\SW$ with respect to the first $L$ coordinates (the $c_\pa$ variables) at the point ${\bar{c},\hbar=0}$ is an invertible $L\times L$ matrix. Indeed it just reduces to the differential of the smooth map from $\CC^L$ to $\CC^L$ defined by \eqref{hbarzeroeq} which is clearly non-degenerate if all $\binhom$ are distinct. Therefore we can apply the analytic implicit function theorem to conclude that all the solutions at $\hbar=0$ can be analytically extended to solutions in a neighbourhood of $\hbar=0$.

It remains to show that all solutions for some $\hbar\neq 0$ can be obtained by extending from $\hbar=0$. This is equivalent to establishing a version of properness, namely that all sequences of solutions $\bar{c}^{(n)}$ corresponding to $\binhom[]$ and $\hbar^{(n)}\rightarrow 0$ as $n\rightarrow \infty$ remain bounded. Assuming the contrary and rewriting \eqref{mastereq} in terms of the roots $u_i$ of the Q-functions, we can extract a subsequence of solutions $\bar{u}^{(n)}$ such that every root either converges to a finite limit or diverges to infinity. But then one of the equations of $\eqref{mastereq}$ will contain a monomial with all the diverging roots and only those. This can be seen by the fact that at $\hbar=0$, \eqref{mastereq} reduces to \eqref{hbarzeroeq}. This monomial will therefore grow faster than any other possible monomial as $n\rightarrow \infty$. Since $\hbar^{(n)}\rightarrow 0$ and $\binhom[]$ is finite we arrive at a contradiction.
\end{proof}

\subsection{Twist-less case, labelling with standard Young tableaux}\label{sec:twist-less-case}

Using the $\hbar\to 0$ limit is not sufficient in the absence of twist as
there are no longer fast-oscillating terms $z_{\sA}^{u/\hbar}$ in Q-functions and so the QQ relations won't simplify to the structure $QQ=QQ$~\footnote{They will reduce to those of the Gaudin model, see Section~\ref{sec:Gaudin}}. Hence we use a stronger regime when 
\be
\label{scaling}
\inhom[L]\gg \theta_{L-1}\gg\ldots \gg\inhom[1]\gg \hbar\,.
\ee
The technical analysis of this limit is given in Appendix~\ref{app:BetheDetails}, and here we describe its combinatorial outcome. Similarly to the twisted case, each root of each Q-function $\mathbb{Q}_{a,s}$ ``sticks to'' a specific inhomogeneity in the sense that, at the leading order of \eqref{scaling}, it is proportional to this inhomogeneity. However, in contrast to the twisted case, the coefficients of proportionality are not equal to one and an arbitrary distribution of the roots between the inhomogeneities is not allowed. Namely, when $\theta_L$ is large compared to the other inhomogeneities, the Q-functions must exhibit the behaviour~\footnote{$\sim$ designates an equality at the leading order of the corresponding expansion, for this case -- the large-$\theta_L$ expansion. Equality is verified by comparing coefficients of polynomials in $u$.}
\begin{equation}\label{eq:QSplit}
  \wQ_{a,s}(u)\sim (u-\N{L}as\, \inhom[L])\, \tilde \wQ_{a,s}(u)
\end{equation}
for certain $(a,s)$ and $\wQ_{a,s}(u)\sim \tilde\wQ_{a,s}(u)$ for the other $(a,s)$, such that $\tilde\wQ_{a,s}$   form a Q-system on a Young diagram which is obtained from $\LD$ by removing one box. $\N{L}as$ are numerical coefficients fixed below. It is clear that \eqref{eq:QSplit} will hold precisely for those $(a,s)$ for which $\deg \wQ_{a,s}=\deg \tilde\wQ_{a,s}+1$, and these are the points satisfying $a < \ao,s < \so$ if the box $(\ao,\so)$ is being removed.

After removing one box, we end up with a Q-system for a spin chain of length $L-1$. Now we repeat the argument with $\theta_{L-1}\to \infty$ and so on and recursively fully disentangle the Young diagram. We associate a number $\SYT_{\ao,\so}$ to each box of the Young diagram which is equal to the length of the spin chain at which the box $(\ao,\so)$ decouples. These numbers range from $1$ to $L$ and increase across each row and each column, \ie they form a standard Young tableau (SYT) $\SYT$ of shape $\LD$. An equivalent statement on the level of NBAE was made for $\gln$ in \cite{Kirillov1986}, although the $\hbar\to0$ regime was suggested (which we know is not sufficient) and details were not given for $\gn>2$.

For a given SYT $\SYT$, the solution $\wQ_{a,s}$ at the leading order of \eqref{scaling} is then given by
\be\label{eq:wQlead0}
\wQ_{a,s}\sim\wQ_{a,s}^{\rm lead}=\prod_{\ao> a,\so> s}\left(u-\N{\SYT_{\ao,\so}}as\theta_{\SYT_{\ao,\so}}\right)\,,
\ee
see Figure~\ref{fig:6}.
\begin{figure}
  \centering
  \begin{tikzpicture}[rotate=-90,scale=.5,node distance=0pt]
    \fill [red!25!white] (0,2) +(2,1) |- +(1,2) |- +(0,3) |- cycle;
    \fill [green!25!white] (2,1) rectangle ++ (2,2);
    \draw (1,0) |- (0,5) |- (5,0) |- (0,1)
    (2,0) |- (0,4) (4,0) |- (0,3) (3,0) -- (3,3) (0,2) -- (4,2);
\foreach \l[count=\y] in
    {{1,2,4,7,10},{3,6,8,13},{5,9,11},{12,15,16},{14}}
    { \foreach \t [count=\x] in \l
      {\path (-0.5,-0.5) +(\y,\x) node {\t};}}
    \draw [<-, shorten <= .15cm,>=stealth, green!50!black] (2,1) node [circle, fill=green!50!black,inner sep=.05cm] {}
    to [bend right=-10] ++(1.,3) node [anchor=175,text=black] 
    {$\mathbb
        Q_{2,1}=(u-\underbrace{ u^{(2,1)}_1}_{\propto \theta
          _{{\color{green!60!black}{9}}}}\color{black})
        (u-\underbrace{ u^{(2,1)}_2}_{\propto \theta
          _{\color{green!60!black}{11}\color{black}}})
        (u-\underbrace{ u^{(2,1)}_3}_{\propto \theta
          _{\color{green!60!black}{15}\color{black}}})
        (u-\underbrace{ u^{(2,1)}_4}_{\propto \theta
          _{\color{green!60!black}{16}\color{black}}}) $};
    \draw [<-, shorten <= .15cm,>=stealth, red!50!black] (0,3) node [circle, fill=red!70!black,inner sep=.05cm] {}
    to [bend right=-20] ++(0,2.5) node [anchor=173,text=black] 
    {$\mathbb
        Q_{0,3}=(u-\underbrace{ u^{(0,3)}_1}_{\propto \theta
          _{{\color{red!80!black}{7}}}}\color{black})
        (u-\underbrace{ u^{(0,3)}_2}_{\propto \theta
          _{\color{red!80!black}{10}\color{black}}})
        (u-\underbrace{ u^{(0,3)}_3}_{\propto \theta
          _{\color{red!80!black}{13}\color{black}}})$};
  \end{tikzpicture}
  \caption{\label{fig:6}Description of a solution of twist-less Q-system \via a standard Young tableau:
    in the regime $\inhom[L]\gg \theta_{L-1}\gg\ldots \gg\inhom[1]\gg
    \hbar$, the scaling of the roots of the $\mathbb Q$
    function at any given node is given by the set of boxes to the
    bottom right of this node (for instance the red resp. green boxes
    for $\mathbb Q_{0,3}$ resp. $\mathbb Q_{2,1}$).}
  \label{fig:Q_tab}
\end{figure}
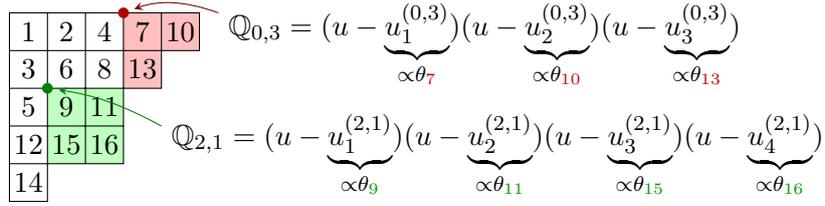

Let us now fix the coefficients $\N{L}as$. To this end recall that $\wQ_{a,s}$ are defined as monic polynomials and restore the normalisation in \eqref{eq:58}
\begin{equation}
  \label{eq:QQ_norm} \wQ_{a+1,s}\wQ_{a,s+1}=\frac{\wQ_{a,s}^+\wQ_{a+1,s+1}^--\wQ_{a,s}^-\wQ_{a+1,s+1}^+}{\hbar\left(\deg
    \wQ_{a,s}-\deg \wQ_{a+1,s+1}\right)}\,.
\end{equation}
The key identity we will need is the relation of the polynomial degrees of $\wQ_{a,s}$ to the hook length $h_{a,s}$ of the Young diagram box $(a,s)$
\be
\deg\wQ_{a,s}-\deg\wQ_{a+1,s+1}=h_{a+1,s+1}\,.
\ee
Note that $h_{a,s}$ will change when we remove certain boxes from the Young diagram. Denote therefore by $h_{a,s}^{(\pa)}$ the hook length in the diagram with $\pa$ boxes which appears in the recursive procedure. Also, denote by $\wQ_{a,s}^{(\pa)}$ the Q-functions on this diagram.

Let us remove the box $(\ao,\so)$ in the recursive procedure. This means that there are currently $\pa=\SYT_{\ao,\so}$ boxes in the diagram, and one has $\wQ_{a,s}^{(\pa)}\sim(u-\N{\pa}{a}{s}\,\theta_{\pa})\wQ_{a,s}^{(\pa-1)}$ for all pairs $(a,s)$ with $a< \ao,s< \so$ and $\wQ_{a,s}^{(\pa)}\sim\wQ_{a,s}^{(\pa-1)}$ otherwise. Consider the regime $u\ll\inhom$ for which $(u+\kappa\,\hbar-\N{\pa}{a}{s}\,\theta_{\pa})\simeq -\N{\pa}{a}{s}\,\theta_{\pa}$ for any finite $\kappa$ and so \eqref{eq:QQ_norm} simplifies providing consistency relations between $\N{\pa}{a}{s}$:
\be
\begin{array}{rlcl}
\N{\pa}{a+1}{s}\N{\pa}{a}{s+1} &=\N{\pa}{a}{s}\N{\pa}{a+1}{s+1}\,, & & a<\ao-1\,,\ s<\so-1\,,
\\
\N{\pa}{a}{s+1} &=\N{\pa}{a}{s}\frac{h_{a+1,s+1}^{(\pa)}-1}{h_{a+1,s+1}^{(\pa)}}\,, & & a=\ao-1\,,\ s<\so-1\,,
\\
\N{\pa}{a+1}{s} &=\N{\pa}{a}{s}\frac{h_{a+1,s+1}^{(\pa)}-1}{h_{a+1,s+1}^{(\pa)}}\,, & & a<\ao-1\,,\ s=\so-1\,.
\end{array}
\ee
There is no constraint at the point $a=\ao-1,s=\so-1$ but instead we have to set $\N{\pa}00=1$  since $\wQ_{0,0}^{(L)}=Q_{\theta}$ and so $\wQ_{0,0}^{(\pa)}=(u-\inhom)\wQ_{0,0}^{(\pa-1)}$. This normalisation allows finding all $\N{\pa}{a}{s}$ with no ambiguities:
\be
\N{\pa}{a}{s}=\prod_{a'=1}^{a}\frac{h_{a',\so}^{(\pa)}-1}{h_{a',\so}^{(\pa)}}\prod_{s'=1}^{s}\frac{h_{\ao,s'}^{(\pa)}-1}{h_{\ao,s'}^{(\pa)}}\,,\quad a< \ao, s< \so\,.
\ee
We emphasise that this solution depends on the choice of $\SYT$ through the condition $\SYT_{\ao,\so}=\pa$. Therefore, two distinct $\wQ_{a,s}^{\rm lead}$ differ by the scaling of at least one Bethe root.

We hence confirmed that $\wQ_{a,s}^{\rm lead}$ is explicitly and bijectively fixed by standard Young tableaux. We remind that the number of SYT is the dimension over $\mathbb{C}$ of $\VTl$ on which the Bethe algebra is restricted. 

Since inhomogeneities are not bounded in the regime \eqref{scaling} one still needs to perform work to show that solutions in this regime are bijectively linked to solutions at finite values of inhomogeneities. This is done in Appendix~\ref{sec:enum-solut-with}. In summary, we have the following result
\begin{lemma}
\label{thm:62}
For $\Lambda_\pa:=\frac{\binhom[\pa+1]}{\binhom[\pa]}$, $\pa=1,\ldots,L-1$, being large enough but finite, solutions of the Q-system on a Young diagram $\LD$ and hence of the Wronskian Bethe equations \eqref{mastereq} at point $\binhom[]$ are bijectively labelled with standard Young tableaux, where the solution associated with a tableau $\SYT$ approaches
\be\label{eq:wQlead}
\wQ_{a,s}\sim\prod_{\ao>a,\so> s}\left(u-\binhom[\SYT_{\ao,\so}]\prod_{a'=1}^{a}\frac{h_{a',\so}^{(\SYT_{\ao,\so})}-1}{h_{a',\so}^{(\SYT_{\ao,\so})}}\prod_{s'=1}^{s}\frac{h_{\ao,s'}^{(\SYT_{\ao,\so})}-1}{h_{\ao,s'}^{(\SYT_{\ao,\so})}}\right)\,
\ee
when $\Lambda_\pa$ approaches infinity. \qed
\end{lemma}

\section{Summary and applications}
\label{sec:summary}
\subsection{Completeness, faithfulness, and maximality of the Bethe algebra}
In this paper we proved completeness of the Wronskian Bethe equations (WBE) and faithfulness of the map from the Wronskian to the Bethe algebra, for the case of both twisted and twist-less supersymmetric spin chains.

Completeness on the level of equations is the statement that the algebraic number of solutions of the WBE is the "right one", \ie it is equal to the dimension of the weight space $\VV$ (as a vector space over $\CC$). We proved the statement for arbitrary numerical values of $\sepa$ -- elementary symmetric polynomials in inhomogeneities $\inhom$. The paper actually contains two independent proofs. The first one is based on a character computation presented in Section~\ref{sec:Hilbert} which is valid because the Wronskian algebra $\WAL$ is a free $\CC[\sse]$-module, by Lemma~\ref{freeness}. The second proof is based on the explicit solution counting in the limits $\left|\frac{\inhom[\pa+1]-\inhom[\pa]}{\hbar}\right|\gg 1$ (twisted case) and $\frac{\inhom[\pa+1]}{\inhom[\pa]}\gg 1$ (twist-less case). The fact that this counting is valid for finite (but probably large) values of inhomogeneities is summarised in Lemma~\ref{thm:61} and Lemma~\ref{thm:62}; the fact that the algebraic number of solutions remains the same for any values of inhomogeneities is a consequence of freeness but we also show this using more elementary arguments in Lemma~\ref{thm:counting}.

Faithfulness and hence bijectivity of the map established in Theorem~\ref{thm:isomorphism} allows one to transfer algebraic properties of the Wronskian algebra $\WAL$ to the Bethe algebra $\BAL$. The Bethe algebra over $\CC[\sse]$ and restricted to the weight subspace $\VV$ can then be viewed as a polynomial ring defined by WBE. Furthermore, for the twist-less case, $\BAL$ depends  on the Young diagram alone and does not depend on the rank of $\glmn$. Its description in terms of a Q-system on a Young diagram directly follows from the results of \cite{Marboe:2016yyn,Kazakov:2015efa} although this fact was not explained there and we filled in the gap in Section~\ref{sec:variouspar}. Using the bosonisation trick, we also found a novel very explicit way \eqref{eq:reconsQ} to parameterise functions $\wQ_{a,s}$ using Wronskian determinants which is the main technical tool for analysing the $\frac{\inhom[\pa+1]}{\inhom[\pa]}\gg 1$ regime.

The faithfulness property holds also for specialisation of $\sepa$ to any numerical value $\bsepa=\sepa(\binhom[1],\ldots,\binhom[L])$, whereas inhomogeneities should probably satisfy the constraint $\binhom+\hbar\neq\binhom[\pa']$ for $\pa<\pa'$~\footnote{If solutions of WBE are non-degenerate this constraint is not needed. For $\sse\in\Secrit$, it might be needed but we did not analyse precisely when, so we keep it as a sufficient requirement. Analysing when it is necessary would probably require exploration of Yangian representation theory beyond techniques developed in the paper.}. For almost any values of $\bsepa$, this follows already from Theorem~\ref{eq:alwaysiso} which uses very general properties of the WBE. However, to get really arbitrary values of $\bsepa$, a  more refined analysis is performed in Appendix~\ref{sec:cyclic} which relies  on properties of the Yangian and its representations. This analysis builds on ideas of \cite{MTV} generalising them to the supersymmetric case, with notable exception of Lemma~\ref{thm:nonzero}. %

Completeness and faithfulness combined ensure that $\BAL$ is a maximal commutative subalgebra of $\End(\VV)$ which should be viewed as the completeness property on the Bethe algebra level. Each {\it distinct} solution of WBE bijectively corresponds to a joint eigenstate of commuting charges. The word ``distinct'' means that even in the case when solutions degenerate, the eigenspace corresponding to the coinciding solutions is still one-dimensional. This does not contradict maximality of the Bethe algebra as the latter becomes non-diagonalisable in the degenerate case with the size of the corresponding trigonal block  equal to the degree of degeneration.
\newline
\newline
The results of this paper are likely to be generalisable for spin chains in an arbitrary highest-weight representation of $\glmn$~\footnote{A unified approach for twisted and non-twisted Bethe subalgebras of $\Ygln$ has been put forward in \cite{Ilin_2017}. Using these ideas and representation-theoretic arguments, powerful completeness-type results have been proven for the $\gl_2$ case in \cite{mashanovagolikova2019simplicity}.}. Given our preliminary studies, an analog of the quantisation condition \eqref{mastereq} will not be sufficient but we can build on a generalisation of Lemma \ref{thm:pollemma} for Q-systems on Young diagrams. Q-systems on Young diagrams can be also defined for non-compact spin chains \cite{Marboe:2017dmb} and they suggest an explicit isomorphism map between restricted Bethe algebras $\BAL$ for non-compact spin chains and compact spin chains. The isomorphism class of $\BAL$ should depend only on the extended Young diagram introduced in \cite{Marboe:2017dmb,Gunaydin:2017lhg}. Performing the suggested program should prove completeness of the quantum spectral curve for $\mathcal{N}$=4 SYM \cite{Gromov:2013pga,Gromov:2014caa} which is confirmed so far by an extensive analysis of its solutions in \cite{Marboe:2017dmb,Marboe:2018ugv}.

\subsection{Simplicity of the spectrum and controlled numerical solution}
\label{sec:alternative}
Choose normalisation $\hbar=\ii$ and consider the twist-less case and real values of inhomogeneities $\inhom$. The Bethe algebra is invariant under Hermitian conjugation in this case and hence is diagonalisable. On the other hand, diagonalisation is impossible if there are coinciding solutions which immediately implies that restriction of the Bethe algebra to a weight subspace $\VTl$ has {\it simple spectrum}, \cf Corollary~\ref{thm:49}.

The considered scenario contains both the homogeneous spin chain with $\inhom=0$ and the spin chain decoupling limit $|\inhom[\pa+1]/\inhom[\pa]|\gg 1$ where WBE can be solved explicitly and labelled with standard Young tableaux. We can continuously connect the homogeneous spin chain and the spin chain in the decoupling limit while keeping $\inhom$ real and in this way {\it unambiguously} label solutions of homogeneous Bethe equations by SYT \footnote{If we allow $\inhom$ to be complex then we wan connect any two solutions of WBE by varying $\inhom[]$. This follows from the path connectivity argument of Section \ref{sec:basicproperties} and allows in turn to connect any two SYT proving a conjecture made in \cite{Kirillov1986}.}.

We have realised this idea numerically. Parameterise inhomogeneities as $\binhom=\Lambda^{\pa-1}-1$, $\pa=1,\ldots,L$. For a chosen tableau $\SYT$, start with the solution \eqref{eq:wQlead} in the decoupling regime $\Lambda\gg 1$ and then incrementally decrease $\Lambda$ until it reaches the point $\Lambda=1$. While changing $\Lambda$ we require that \eqref{eq:QQ_norm} are always satisfied. The numerical realisation turned out to be very stable for any choice of $\SYT$ that we tried. For $L\lesssim 20$, we are able to produce, for a given $\SYT$, a numerical solution with an 80-digit precision in less than three minutes, the speed is obviously much faster for shorter chains. Further substantial optimisation of the code should be possible.

The details of implementation and the code will be published in a separate work. Here we give one illustration. For the example of Q-system on Young diagram on page~\pageref{exmp:p44}, choose a standard Young tableau
$$
\includegraphics[width=0.4\textwidth]{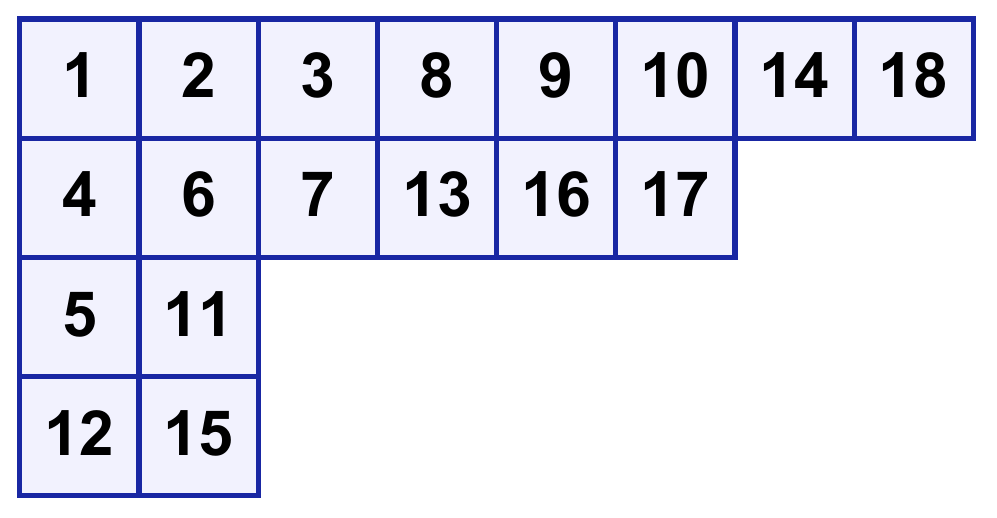}
$$
Then we get the following numerical solution

\definecolor{mblue}{rgb}{0, 0.1, 0.96}
\definecolor{mgreen}{rgb}{0.30, 0.55, 0.19}
\definecolor{mred}{rgb}{0.61, 0.12, 0.08}
\noindent\begin{minipage}{0.45\textwidth}
\begin{center}
{
\includegraphics[width=\textwidth]{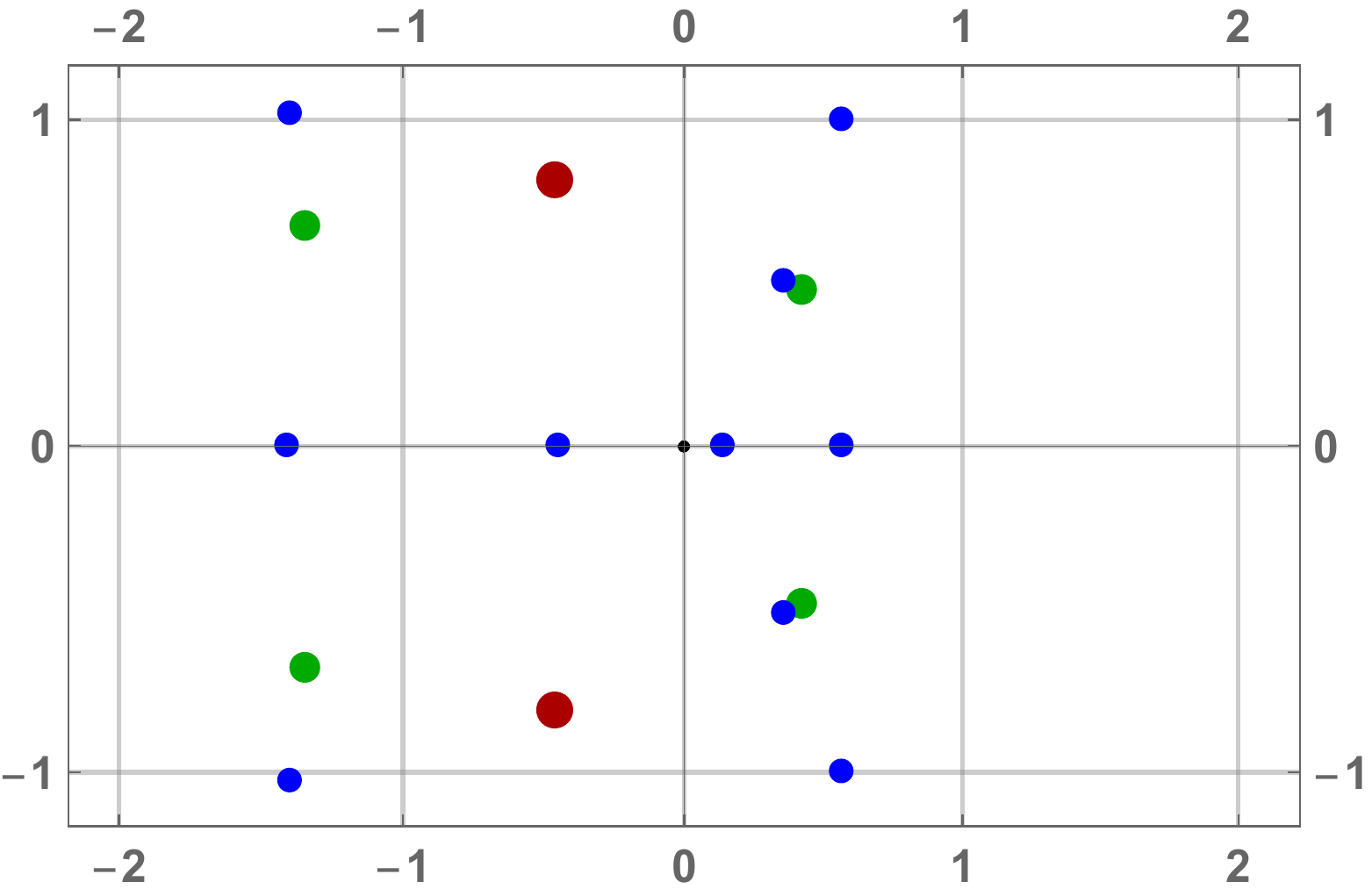}
}
\end{center}
\scriptsize
Roots of $\wQ_{1,0}$  ({\color{mblue} blue}), $\wQ_{2,0}$  ({\color{mgreen} green}), $\wQ_{3,0}$ ({\color{mred} red}) are shown
\end{minipage}
\hspace{0.1\textwidth}
\begin{minipage}{0.45\textwidth}
\begin{center}
{
\includegraphics[width=\textwidth]{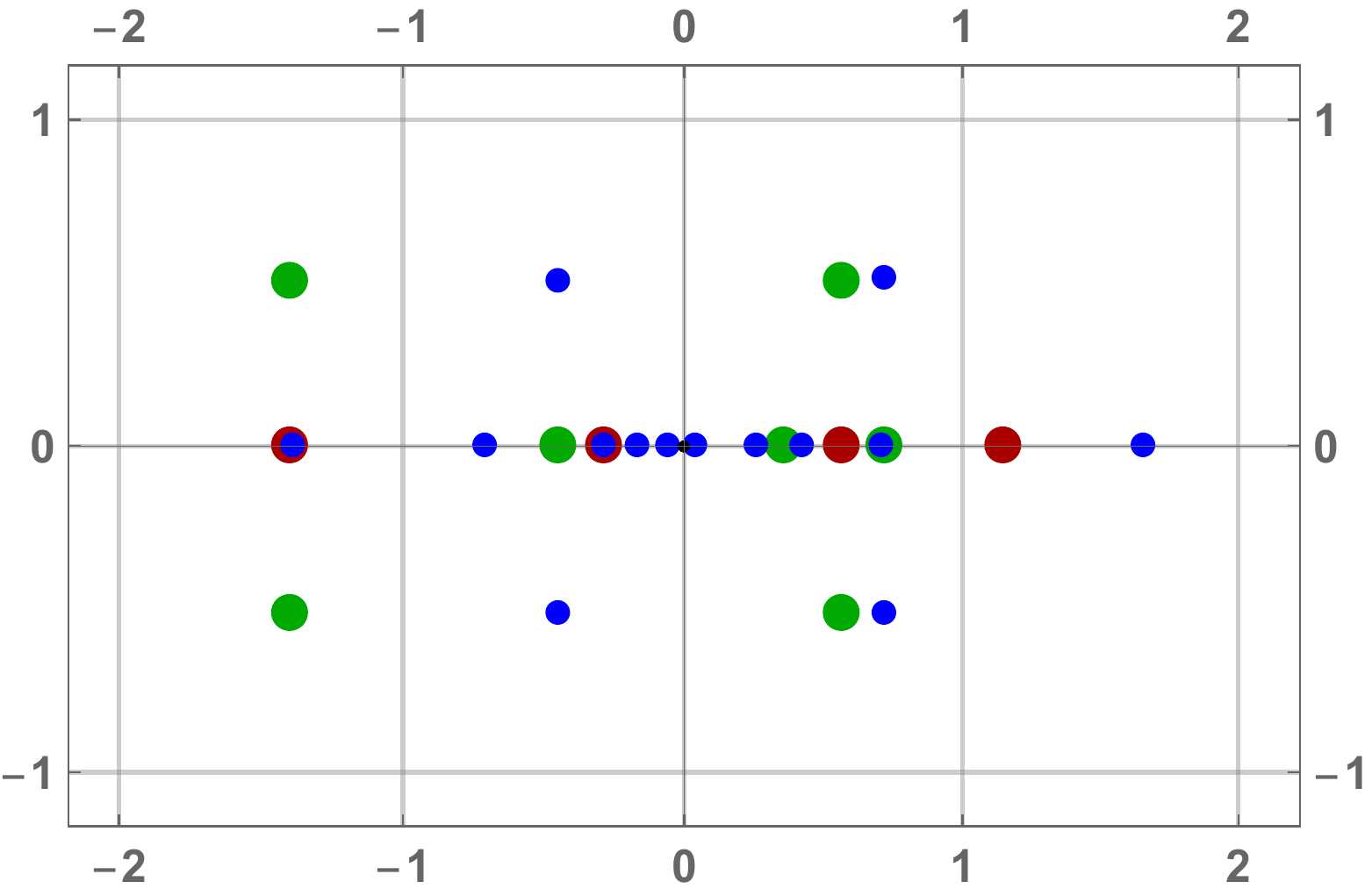}}
\end{center}
\scriptsize
Roots of $\wQ_{0,1}$ ({\color{mblue} blue}), $\wQ_{1,1}$ ({\color{mgreen} green}), $\wQ_{1,2}$ ({\color{mred} red}) are shown
\end{minipage}
\newline
\newline
To our knowledge, the proposed approach is the first example when we have a systematic and proven to be unambiguous way to control all solutions of the Bethe equations for systems of size where direct brute-forcing (\eg by numerical diagonalisation of Hamiltonian matrix) is unlikely to be practical. For instance, there are $2\,148\,120$ different standard Young tableaux of the same shape (and hence distinct solutions of the Bethe equations with the same magnon numbers) as in the above example. Overall, the length $L=18$ $\gl_4$ chain has a Hilbert space of dimension $\sim 6.8\times 10^{10}$ comprising $81\,662\,152$ irreps of $\gl_4$ with one solution of WBE per irrep.

To compare with other approaches, all solutions of the $\gl_2$ chain for $L=14$ were reported in \cite{Hao:2013jqa}. Yet, at this length the Hilbert space is of dimension $16384$ and the Hamiltonian being a sparse matrix \cite{Sandvik:2011bd} can be diagonalised numerically. In \cite{Bargheer:2008kj}  solutions with large magnon numbers were studied quite systematically but only particular classes of solutions were controlled. For spin chains of rather large length, the low-energy excitations around the antiferromagnetic vacuum are also of numerical interest, see \eg \cite{Caux:2005xx}. It would be interesting to explore whether we can apply the proposed techniques in this regime and improve the existing tools which rely on the string hypothesis.
 \subsection{Gaudin model}
\label{sec:Gaudin} 

The $\gln$ Gaudin model \cite{gaudin:jpa-00208506, GaudinBook} can be obtained as the $\hbar\to 0$ limit of the non-twisted $\gln$ spin chain. Formally, this just amounts to replacing the discrete Wronskian by an actual Wronskian. On the representation theory side the spin chain will no longer be a representation of the Yangian $\Ygln$ but of the current algebra $\gln[u]$ (\eg in the terminology of \cite{Mukhin_2009.2}). Completeness of the Bethe Ansatz for the $\gln$ Gaudin model has been proven in \cite{Mukhin_2009.2} under the assumption that all inhomogeneities are pairwise distinct. The philosophy of the proof and the end result are similar to \cite{MTV}. As far as we know, for the supersymmetric $\glmn$ Gaudin model, completeness is  proven for generic values of $\inhom$ \cite{Mukhin_2015} in the twist-less case, and for any pair-wise distinct $\inhom$ for the twisted case \cite{HUANG20201}.

In our construction, the Bethe equations and Q-operators admit a well-defined ${\hbar\to 0}$ limit and an analogue of Theorem~\ref{thm:isomorphism} can be proven along the same lines. Similarly, all the results relying solely on the analytic properties of the map $\SW$ will also be true in the Gaudin case. In particular, the algebraic number of solutions will not depend on $\bsse$ and will be equal to $d_\Lambda$. Specialisation of the isomorphism will also hold generically using the same arguments as in Theorem~\ref{eq:alwaysiso}. 

To prove further constraints on specialisation as in Theorem~\ref{isotheorem}, one way would be to investigate the representation theory of $\glmn[u]$. But instead of doing that, we can reformulate statements for the Bethe algebra $\BAL$ of the $\glmn$ system as statements for the non-supersymmetric $\gl_{h_{\Lambda^+}}$ system. This is based on the results of Section~\ref{sec:variouspar} for isomorphisms of Bethe algebras and the discussion on page~\pageref{p65} for the existence of a cyclic vector. They apply in the $\hbar\to 0$ limit as well. Using \cite{Mukhin_2009.2}, one then confirms that the specialisation of the isomorphism also holds in the case of the supersymmetric Gaudin model for pairwise distinct inhomogeneities which is a naive $\hbar\to 0$ limit of Theorem~\ref{isotheorem} and related statements.

\subsection{Separation of variables}
\label{sec:SoV}
To construct a basis that factorises wavefunctions of eigenstates of the Bethe algebra, Maillet and Niccoli proposed \cite{Maillet:2018bim} to repeatedly act with transfer matrices on a reference state. One can reach factorisation also by choosing other Bethe algebra elements that depend on the spectral parameter $u$. This idea was fruitfully used recently alongside with other related tools in application to rational spin chains \cite{Ryan:2018fyo,Maillet:2018czd,Maillet:2019nsy,Maillet:2019ayx,Gromov:2019wmz,Ryan:2020rfk,Maillet:2020ykb}. Currently, an SoV basis was constructed for $\glm$ spin chains in arbitrary finite-dimensional representation as \cite{Ryan:2020rfk} 
\be
\bra{\bf x}=\bra{0}\prod_{\pa=1}^L\prod_{k=1}^{\gm-1}\det\limits_{1\leq i,j\leq k}Q_i(x_{kj}^{\pa})
\ee
and for $\glmn$ spin chains in the defining representation as \cite{Maillet:2019ayx}
\be
\bra{\bf x}=\bra{0}\prod_{\pa=1}^L \wT_{(1)}(\inhom)^{d_{\pa}}\,.
\ee
Here $x_{kj}^{\pa}=\inhom+\hbar\,m_{kj}^{\pa}$, where $m_{kj}^{\pa}$ are integers forming Gelfand-Tsetlin patterns and defining what is $\bra{\bf x}$, and $d_{\pa}$ are integers from $0\leq d_{\pa}\leq \gm+\gn-1$ also defining what is $\bra{\bf x}$. There are exactly as many choices for $m_{kj}^{\pa}$ and $d_{\pa}$ as the dimension of the corresponding Hilbert space.

An important technical challenge of this approach is to prove that a construction as above indeed produces a basis of the Hilbert space. It was resolved in the mentioned works, however only spin chains with generic twist were considered and certain restrictions on admissible values of inhomogeneities were used.

We can now give an alternative insight on resolving this challenge. By Theorem~\ref{repisotheorem}, the representation of the Bethe algebra is isomorphic to the regular representation of the Wronskian algebra. In particular there is always a cyclic vector. One can choose the cyclic vector as a reference state $\bra{0}$. Then $\bra{\bf x}=\bra{0}\varphi(\bas_x)$ form a basis as long as $\bas_x$ form a basis in the Wronskian algebra. The Wronskian algebra is a polynomial algebra with Pl{\"u}cker-type relations. Then the basis question reduces to questions very similar to those of projective geometry. This naturally links to the last topic we would like to review.

\subsection{Geometric representation theory and Bethe/Gauge correspondence}

So far we have concentrated on the isomorphism between the restricted Bethe algebra $\BAL$ and the polynomial ring $\WAL=\CC[\sse][c]/\mathcal{I}_\Lambda$. It turns out that for $\gl _\gm$ spin chains the Bethe algebra can be realised in a third, purely geometric way \cite{maulik2012quantum, Nekrasov_2009, Nekrasov:2009ui, Gorbounov_2013}. Consider the manifold $\mathcal{F}_\Lambda$ of all the partial flags associated to the weight $\Lambda=[\lambda_1,\ldots,\lambda_{\gm}]$, that is chains of vector spaces
\begin{equation}
\{0\}=V_0\varsubsetneq V_1 \varsubsetneq \ldots \varsubsetneq V_{\gm-1} \varsubsetneq V_\gm=\CC^L
\end{equation}
such that $\mathrm{dim}~V_a/V_{a-1}=\lambda_a$ for all $1\leq a\leq \gm$. $\mathcal{F}_\Lambda$ admits a natural action of $\GL (L)$. Now consider its cotangent bundle $T^*\mathcal{F}_\Lambda$ with an action of $\GL (L)\times\CC^*$, where $\CC^*$ acts on the cotangent spaces by multiplication. It is known \cite{Gorbounov_2013} that the equivariant cohomology ring of $T^*\mathcal{F}_\Lambda$ with this action of $\GL (L)\times\CC^*$ is given by
\begin{equation}
H^\bullet_{\GL (L)\times\CC^*}(T^*\mathcal{F}_\Lambda,\CC)=\CC[c, \sse, \hbar]/\langle\prod_{a=1}^\gm q_a(u)-Q_\theta(u)\rangle.
\end{equation}
In this identification the $L+1$ parameters $(\sepa)_{1\leq \pa \leq L}$ and $\hbar$ come from the standard $H^\bullet_{\GL (L)\times\CC^*}(\mathrm{pt},\CC)=\CC[\chi,\hbar]$-module structure of $H^\bullet_{\GL (L)\times\CC^*}(T^*\mathcal{F}_\Lambda,\CC)$. Treating $\hbar$ as an additional parameter now, $H^\bullet_{\GL (L)\times\CC^*}(T^*\mathcal{F}_\Lambda,\CC)$ is thus $\CC[\sse,\hbar]$-isomorphic to $\BAL$ in the singular twist limit limit $x_1\ll x_{2}\ll\ldots\ll x_\gm$. In general, when $x_a$ are arbitrary pairwise distinct complex numbers, the twisted Bethe algebra $\BAL$ can be identified with a quantum deformation of $H^\bullet_{\GL (L)\times\CC^*}(T^*\mathcal{F}_\Lambda,\CC)$, the so-called equivariant quantum cohomolgy ring $QH^\bullet_{\GL (L)\times\CC^*}(T^*\mathcal{F}_\Lambda,\CC)$. 

This connection is actually a particular case of a more general construction \cite{maulik2012quantum} which first appeared in the context of Bethe/Gauge correspondence \cite{Nekrasov_2009, Nekrasov:2009ui}. Starting from any so-called Nakajima quiver variety one can build a Yangian action~\footnote{The Yangian in question is not necessarily $\Y(\glm)$. An identification of these geometrically-realised Yangians with known integrable systems and solutions of the Yang-Baxter equation is an open question.} on its equivariant quantum cohomology ring considered as a Hilbert space. Moreover, the action of the Yangian generators can be expressed as some geometric operations on the classes of the variety. In particular the Baxter Q-operator can be constructed in a purely geometric way \cite{Pushkar:2016qvw}.

It is still unclear how to properly extend this construction to supersymmetric Yangians \cite{Nekrasov:2018gne}. We hope that results of our paper, in particular the isomorphism between the bosonic and supersymmetric case elucidated in Section~\ref{sec:variouspar}, will be useful for the advancement of the subject.

\acknowledgments
We are grateful to Luca Cassia, Vladimir Dotsenko, Antoine Ducros, Arnaud Eteve, David Hernandez, Atsuo Kuniba, Jules Lamers, Fedor Levkovich-Maslyuk, Maksim Maydanskiy, Lucy Moser, Sergey Mozgovoy, Rafael Nepomechie, Antoine Picard, Paul Ryan, Peter Schauenburg, Didina Serban, Samson Shatashvili, Pedro Tamaroff, Ronan Terpereau, and Emmanuel Wagner for useful discussions. 

The work of D.V. and S.L. was partially supported by the Knut and Alice Wallenberg Foundation under grant ``Exact Results in Gauge and String Theories''  Dnr KAW 2015.0083. The work of S.L. was partially supported by the European Union (through the PO FEDER-FSE Bourgogne 2014/2020 program) and the EIPHI Graduate School (contract ANR-17-EURE-0002) as part of the ISA 2019 project, and by the région Bourgogne-Franche-Comté as part of the MolQuan project.

D.V. is very grateful to Institut de Math\'ematiques de Bourgogne and D\'epartement de Physique de l'Ecole Normale sup\'erieure for hospitality where a part of this work was done. S.L. is very grateful to Nordiska institutet f{{\"o}}r teoretisk fysik, where a part of this work was done.

\appendix

\section{Wronskian algebra - a pedestrian approach}
\label{sec:pedestrian}
\subsection{Some facts from commutative algebra and algebraic geometry}
\label{sec:alggeom}
\subsubsection*{Basic definitions}
The study of the polynomial equations can be done in an analytic (geometric) or in an algebraic way. In the analytic approach, to a set of $m$ polynomial equations in $n$ variables $P_{i}(x_1,\ldots,x_n)=0$, $i=1,\ldots, m$ is assigned an algebraic variety $\mathcal{A}$ - a set of points  $x\equiv (x_1,\ldots,x_n)\in \mathbb{C}^n$ where these equations hold. The algebraic approach attaches to the equations an ideal generated by $P_{\pa}$, $\CI=\langle P_{\pa} \rangle$ which is the set of all possible polynomials in $n$ variables $Q\in \CC[x_1,\ldots, x_n]$  that can be written in the form $Q=\sum_{\pa}q_{\pa}P_{\pa}$ for some polynomials $q_{\pa}$.

The relation between the two approaches is established by Hilbert's Nullstellensatz: if $Q$ vanishes on $\mathcal{A}$ then $Q^{r}\in \CI$ for some integer $r$. The ideal constructed by all polynomials that vanish on $\mathcal{A}$ is called the radical of $\CI$ and is denoted by $\sqrt{\CI}$. 

The algebraic description is more abstract and is less used in physics but it allows one to more accurately formulate some of the properties of the Wronskian algebra. In particular, we can work over  fields different to $\CC$, \eg the field of fractions $\CC(\inhom[])$. 

The next concept is to consider functions on $\mathcal{A}$ formalised as the quotient ring
\be
\label{Qring}
\CR=\CC[x_1,\ldots,x_n]/\CI\,.
\ee
If $\CI=\sqrt{\CI}$ then $\CR$ is called the coordinate ring of $\mathcal{A}$. Note that in this paper not all ideals are equal to their radicals.

Finally, recall that a ring is said to be an integral domain if $ab=0$ implies $a=0$ or $b=0$. The corresponding ideal is then called prime ($ab\in \CI$ implies $a\in \CI$ or $b\in \CI$).

\subsubsection*{Polynomial division}
Easiness of the study of polynomials in {\it one variable} exists mainly due to the unambiguous polynomial division procedure. Recall how it works: let $P$ be a polynomial in $x$ of degree $b$ which is one (of those polynomials) that generates the ideal $\CI$ in $\CC[x]$. Let  $Q$ be any polynomial in $x$. If $Q$ contains a monomial $c\, x^a$ with $a\geq b$, we represent $Q$ as a combination $Q=c\, x^{a-b} P+(Q-c\, x^{a-b}P)$ in which the first term is divisible by $P$ and the second term has no monomial of degree $a$. One performs the same procedure with $P'=Q-c\, x^{a-b}P$ and continue it recursively until no monomials divisible by $x^b$ remain. So one obtains a representation $Q=q\,P+r$\,, where the degree of $r$ is strictly smaller than $b$. Both $q$ and $r$ are fixed uniquely. Furthermore, one can guarantee the B{\' e}zout's lemma, that is, one can find such $\alpha,\beta\in\CC[x]$ that $\alpha P_1+\beta P_2={\rm GCD}(P_1,P_2)$ for any polynomials $P_1,P_2$, and hence conclude that any ideal in one variable is principal, \ie   it is  generated by a single polynomial -- the GCD of polynomials $P_1,\ldots,P_m$ that generate the ideal.

A practical application in our case would be: if a Bethe algebra is generated by a single operator $\hat x$ then this algebra is guaranteed to be isomorphic to a quotient $\mathbb{C}[x]/\CI$ where $\CI$ is the ideal generated by the minimal polynomial of $\hat x$. As we have $L$ generating operators $\hat c_{\pa}$, things are not that simple.
\subsubsection*{Gr\"obner bases}
Many problems in systems with multiple variables arise from difficulties with the polynomial division. First, to even define a division algorithm one needs to introduce a total order on the set of monomials $x^d\equiv \prod\limits_{i=1}^n x_i^{d_i}$ that should be an order in which $1$ is the smallest monomial and $a<b$ implies $a\,c<b\,c$ for any $a,b,c$. A diversity of monomial orders is available in contrast to only one option for the single-variable case. We shall use below only lexicographic orders which form a small subset of all possibilities.

After fixing a monomial order and denoting by $P_{i}$ the generators of the ideal $\CI$, one can perform long polynomial division (exclusion of all monomials that are divisible by leading monomials of $P_{i}$) to represent any polynomial $Q$ as
\be
Q=\sum_{i} q_{i}\,P_{i}+r\,.
\ee
Unfortunately, neither the procedure nor its result are unique if $P_{i}$ are arbitrary generators, so the division is essentially meaningless.

However, if $P_{i}$ form a special set called  Gr\"obner basis then $r$ is uniquely defined by the polynomial division~\footnote{Note however that, in contrast to the one-dimensional case, monomials comprising $r$ can be still bigger than the leading monomials of $P_{i}$ and yet not divisible by the latter. Hence we might be unable to perform a chain of divisions that leads to the B{\' e}zout's lemma and it generically does not hold in the multivariable case.}. Hence $Q\in \CI$ iff $r=0$. $q_{i}$ are not unique though, but  uniqueness of $r$  suffices for the study of the quotient ring \eqref{Qring}.

A set of polynomials $P_i$ forms a Gr\"obner basis of an ideal $\CI$ if $i)$ they generate $\CI$, $ii)$ the set is closed under computation of S-polynomials, see \eg \cite{becker1998grobner} for further explanations. If moreover, for all $i\neq i'$,  $P_i$ does not contain monomials divisible by the leading monomial of $P_{i'}$  then such a Gr\"obner basis is called the reduced one and it is unique for the given choice of a monomial order. By a Gr\"obner basis we mean the reduced basis in the following.

\subsubsection*{Monomial basis}
Let us fix a Gr\"obner basis. The set of monomials that can arise in the remainders of polynomial divisions forms a basis in the quotient ring $\CR$ considered as a vector space. This basis shall be called the monomial basis.

We can use the monomial basis to realise the regular representation of an algebra in terms of explicit matrices, see the example on page \pageref{ex:36}. Such a basis has an important advantage -- all computations in it are performed in the original field, and so the coefficients of $\check x$  will belong to the same field, where $\check x$ is a matrix in the regular representation, see Section \ref{sec:algedes}.

\subsection{\texorpdfstring{$\mathbb{C}(\sse)$}{C(chi)}-module and invariance of solutions multiplicity}
\label{sec:multiplicity}
Consider the ring of polynomials in $2L$ variables $\CC[\sse][c]\equiv \CC[\se{1},\ldots\se{L}][c_1,\ldots,c_L]$. We define the Wronskian algebra as $\WAL:=\CC[\sse][c]/\CI_\Lambda$, where $\CI_{\Lambda}=\langle \SW_{\pa}(c)-\se{\pa}\rangle$ is the ideal generated by Wronskian relations. As we can simply exclude $\se{\pa}$ using equations $\se{\pa}=\SW_{\pa}(c)$, $\WAL$ is isomorphic over $\CC$ to $\CC[c]$ -- the ring of polynomials in $L$ variables. Hence, in particular, $\WAL$ is an integral domain and $\CI_{\Lambda}$ is a prime ideal.

In the case of prime ideals, it is quite easy to promote rings to fields. In this subsection, we shall consider $\se{\pa}=\SW_{\pa}(c)$ as an equation on $c_{\pa}$ in the field of fractions $\CC(\sse)$ and $\WAL$ as a ring over $\CC(\sse)$. ``Easiness'' of promotion lies in the following statement: any polynomial in variables $c_{\pa}$ and $\se{\pa}$ that belongs to $\WAL$ considered as an object in a ring over $\CC(\sse)$ would also belong to $\WAL$ considered as an object in a ring over $\CC[\sse]$.

When we work over a field of fractions, we can compute a  Gr\"obner basis. Simply, instead of conventional computation in $\CC[c_1,\ldots,c_L]/\langle \SWe_\pa(c)-\bsepa\rangle$ with numerical $\bsepa\in\CC$, we do a computation in $\CC(\sse)[c_1,\ldots,c_L]/\langle \SWe_\pa(c)-\sepa\rangle$ with symbolic $\sepa\in\CC(\sse)$. When the Gr\"obner basis is computed, we can construct the corresponding monomial basis and conclude what is the dimension of $\WAL$ (as a vector space over $\CC(\sse)$) and hence what is the number of solutions of the Wronskian equations. Note that the solutions themselves would typically only exist in an algebraic closure of $\CC(\sse)$. However, computation of the monomial basis can be performed directly in $\CC(\sse)$ and this is the only thing needed.

Working over $\CC(\sse)$ is equivalent to considering $\sepa$ in generic position, when no accidental relations happen. When we specialise to a concrete numerical value $\bsepa$ of $\sepa$, we are interested whether the number of solutions changes. We can formulate (a bit stronger) question from the point of view of the Gr\"obner basis: does it remain a Gr\"obner basis upon specialisation?
\begin{lemma}
\label{lemma:specialisation}
Let the Gr\"obner basis of the ideal $\CI_\Lambda=\langle \SW_{\pa}-\se{\pa}\rangle$ in $\CC(\sse)[c]$ \wrt some monomial order $<$ be given by polynomials 
\be
s_m=c^{m}+\sum_{m'<m}p_{mm'}(\sse)\,c^{m'}\,,\quad m\in M\,,
\ee
where $m:=(m_1,\ldots,m_L)$, $c^{m}:=\prod\limits_{\pa=1}^L c_{{\pa}}^{m_{\pa}}$,  $M$ is a set of tuples $m$, and  $p_{mm'}\in\CC(\sse)$.

Let $p_{mm'}$ be finite numbers when evaluated at $\sepa=\bsepa\in \mathbb{C}$. Then ${\bar s}_m=c^{m}+\sum\limits_{m'<m}p_{mm'}(\bsse)c^{m'}, m\in M$,  form the Gr\"obner basis  of the ideal $\CI_\Lambda(\bsse)=\langle \SW_{\pa}-\bsepa\rangle$  in $\CC[c]$ for the same monomial order.   
\end{lemma}
In other words, it is safe to specialise a Gr\"obner basis at those values of $\se {\pa}$ where denominators of $p_{mm'}$ do not vanish.
\begin{proof}
To verify the statement first we check that the declared set of ${\bar s}_{m}$ generates $\CI_\Lambda(\bsse)$. To this end, use long division in $\CC(\sse)[c_1,\ldots,c_L]$ to write $\SW_{\pa}-\se {\pa}=\sum_m q_m(\sse) s_m$. From the algorithm of long division it is clear that $q_m(\sse)$ are not singular at $\sse=\bsse$ if $p_{mm'}(\sse)$ are not singular which is the case by the condition of the theorem. Hence $\SW_{\pa}-\se {\pa}=\sum_m q_m(\sse) s_m$ can be evaluated and still holds at $\sse=\bsse$. To check that ${\bar s}_{m}$ form a Gr\"obner basis we need to \eg compute S-polynomials but this is combinatorially the same exercise as for $s_m$ since the leading monomials are not affected by specialisation.
\end{proof}

Wronskian equations can be obviously specialised at arbitrary point $\bsse$ and so the ring  $\WAL(\bsse)$ is always a well-defined object. Now we would like to show that, for a given $\bsse$, one can find a Gröbner basis that can be specialised at this point and its vicinity. This is the key point to prove the following theorem:
\begin{theorem}
\label{thm:counting}
$d_{\Lambda}:= \dim_{\CC}\WAL(\bsse)$ does not depend on $\bsse$.
\end{theorem}
In other words, number of solutions of Wronskian equations counted with multiplicities is always the same, even on the degeneration set $\Secrit$.

\begin{proof}
We know that the theorem holds for all points $\bsse\notin \SW(D)$ since all the solutions of the Wronskian equations are distinct there and so the dimension of the quotient ring coincides with the number of solutions that we denote as $d_{\Lambda}$. We can path-connect any two regular points and the number of solutions cannot change along the path, see Section~\ref{sec:basicproperties}. 

Take $L$ linearly independent constant vectors  $w_{\pa}=(w_{\pa1},\ldots,w_{\pa L})$ and define $x_{\pa}=\sum_{\pa'}w_{\pa\pa'}c_{\pa'}$. For almost any choice of $w_{\pa}$, the Gr\"obner basis of $\CI_{\Lambda}$ in $\CC(\sse)$ \wrt the monomial order $x_1<x_2<\ldots x_L$ should have the form
\begin{subequations}
\label{GBeq}
\be
\label{firstGBeq}
&& x_{1}^{d_{\Lambda}}+a_{1}^{(d_{\Lambda}-1)}(\sse)x_{1}^{d_{\Lambda}-1}+\ldots a_{1}^{(0)}\,,\\
&& x_{2}-\sum_{k=0}^{d_{\Lambda}-1}b_{2k}(\sse) x_1^{k}\,,\nonumber\\
\label{secondGBeq}\ldots \\
&& x_{L}-\sum_{k=0}^{d_{\Lambda}-1}b_{Lk}(\sse) x_1^{k}\,.\nonumber
\ee
\end{subequations}
Indeed, take a point $\bsse\notin \Secrit$  for which the conditions of Lemma~\ref{lemma:specialisation} hold. At such a point, leading monomials of the Gr\"obner basis are the same before and after specialisation, and so we can judge about the Gr\"obner basis from its specialised version. Since $\bsse\notin \Secrit$,  $\check x_{\pa}$ (regular representation of $x_\pa$, written as a matrix in the monomial basis) should have $d_{\Lambda}$ distinct eigenvalues for almost any choice of $\omega_{\pa}$,  and therefore the minimal polynomial equation it  satisfies is of degree $d_{\Lambda}$ which is \eqref{firstGBeq}. In the chosen lexicographic order this equation should belong to the Gr\"obner basis. Other variables $x_2,\ldots,x_L$ should satisfy \eqref{secondGBeq} (\ie  they are uniquely fixed if $x_1$ is fixed) otherwise dimension of $\WAL(\bsse)$ would exceed $d_{\Lambda}$.

By the properness of WBE $a_1^{(a)}(\sse)$ cannot have singularities, hence they are simply polynomials in $\se{\pa}$. Coefficients $b_{\pa k}$ however are rational functions of $\se{\pa}$ that can contain poles. Everywhere outside of these poles, the conditions of Lemma~\ref{lemma:specialisation} hold and we can perform specialisation asserting that the dimension of the specialised polynomial ring is $d_{\Lambda}$.

It remains to show that for any $\bsse\in \CC^L$, one can choose $\omega_{\pa}$ such that $b_{\pa k}$ are not singular at $\bsse$. To this end, we can actually explicitly express $b_{\pa k}$ in terms of solutions of the Wronskian system. Let $x_\pa=x_{\pa}^{(i)}$ be the $i$-th solution. Then polynomials \eqref{secondGBeq} can be rewritten as
\be
\label{mainpoleq}
x_{\pa}-\sum_{k=0}^{d_{\Lambda}-1}b_{\pa k}(\sse) x_1^{k}=
\frac{\det\left|\begin{matrix} 
x_{\pa} & 1 & x_1 & x_1^2 &\ldots
\\
x_\pa^{(1)} & 1& x_1^{(1)} & (x_1^{(1)})^2 & \ldots \\
 x_\pa^{(2)} & 1& x_1^{(2)} & (x_1^{(2)})^2 & \ldots\\
x_\pa^{(3)} & \ldots \\
\ldots
\end{matrix}\right|}{\det\left|\begin{matrix} 1 & x_1^{(1)} & (x_1^{(1)})^2 & \ldots \\
 1 & (x_1^{(2)}) & (x_1^{(2)})^2 & \ldots\\
1 & \ldots
\end{matrix}\right|}\,.
\ee
Indeed, equality of the above polynomials to zero implies $x_\pa=x_{\pa}^{(i)}$ precisely when $x_1=x_1^{(i)}$. 

While $x_{\pa}^{(i)}$ belong to an algebraic closure of $\CC(\sse)$, the above ratio of determinants is symmetric under permutations $x_{\pa}^{i}\to x_{\pa}^{\sigma(i)}$ and hence should be a polynomial in $x_1$ with coefficients in the base field, \ie  $\CC(\sse)$. This follows for instance from $\sum_{i=1}^{d_{\Lambda}} f(x^{(i)})={\rm Tr}\, f(\check x)$ and basic combinatorial arguments. Of course, one can conclude the same from the fact that \eqref{secondGBeq} are obtained in the process of computation of the Gr\"obner basis.

At points $\bsse$ where all $x_{1}^{(i)}$ are distinct, the denominator of \eqref{mainpoleq} is non-zero and hence $b_{\pa k}$ are non-singular. As discussed, for a given regular $\bsse$, we can adjust $\omega_{\pa}$ in a way that $x_1$ has non-degenerate solutions.

When $\bsse\in \SW(D)$ all $x_{\pa}$ degenerate. Then consider a one-parametric smooth path $\sse(t)$ in the space of parameters such that $\sse(t=0)=\bsse$ is the degeneration point of interest and $\bsse(t\neq 0)\notin \Secrit$. Moreover, one chooses such a path that all $x_{\pa}^{(i)}$ are distinct  along the path  for sufficiently small $t$, except for the point $t=0$ itself. 

The value of the ratio of determinants in \eqref{mainpoleq} is not well-defined at $t=0$ but it can be computed as the limit $t\to 0$. Since this ratio is a rational function of $\se{\pa}$, the limit, if finite, should produce polynomials \eqref{secondGBeq} specialised at $t=0$.

To compute the limit, note that all $x_{\pa}$, for generic enough choice of $\omega_{\pa}$, satisfy one-variable equations  $x_{\pa}^{d_{\Lambda}}+a_{\pa}^{(d_{\Lambda})}x_{\pa}^{d_{\Lambda}-1}+\ldots =0$, where $a_{\pa}^{(k)}$ are polynomials in $\sse$ and hence are well-defined even at $t=0$. Define $x_{\pa}^i(t)$ as solutions of these equations that coincide with solutions of Wronskian equations for $t\neq 0$; their $t=0$ value is then defined as the continuation $t\to 0$. If $\mu_{i\pa}$ is the degree of degeneration of solution $x_{\pa}^{(i)}$ at $t=0$ (\ie  $\mu_{i\pa}$ solutions of the one-variable equation on $x_{\pa}$ coincide at this point) then $x_{\pa}^{(i)}(t)$ is expanded in the Puiseux series
\be\label{eq:Puiseux}
x_{\pa}^{(i)}(t)=x_{\pa}^{(i)}(0)+r_{\pa i,1}t^{1/\mu}+r_{\pa i,2}t^{2/\mu}+\ldots\,,
\ee 
where $\mu={\rm LCM}(\mu_{11},\ldots,\mu_{LL})$. 

One should know finitely many terms in the series \eqref{eq:Puiseux} to compute the determinants ratio in \eqref{mainpoleq} in the limit $t\to 0$. We require that for these finitely many terms, for each $k$ and $i$, if at least one $r_{\pa i,k}$ is non-zero then all $r_{\pa i, k}$, $\pa=1,\ldots, L$, are non-zero. It is sufficient to guarantee that the ratio is finite, while imposing such a requirement excludes a measure zero subspace from acceptable values of $\omega_{\pa}$. Recall that we already excluded the space of $\omega_{\pa}$ where the degree of the minimal polynomial of $\check x_{\pa}$ is less than $d_{\Lambda}$ which is of measure zero as well. The majority of $\omega_{\pa}$ are outside of the stated restrictions, and we can choose any valid option to guarantee the regularity of $b_{\pa k}$ at $\sse=\sse^{(0)}$ and hence the possibility to specialise the Gr\"obner basis \eqref{GBeq} at this point thus concluding that $\dim_{{\CC}}\WAL(\bsse)=d_{\Lambda}$.
\end{proof}
The proposed proof shows that there is a close analogy between WBE and a polynomial equation in a single variable. Indeed, for any point $\bsse\in\Se$, regular or not, we can choose a variable $x_1$ that satisfies \eqref{firstGBeq} and such that there is a neighbourhood $\CO_{\bsse}$  where $b_{\pa,k}$ are non-singular which allows one to compute all elements of $\WAL$ using \eqref{secondGBeq}. So a single-variable equation \eqref{firstGBeq} contains all information about $\WAL$ in the selected neighbourhood.

\subsection{Freeness of \texorpdfstring{$\WAL$}{W\_Lambda} and trivialisation of a vector bundle}

For each $\CO_{\bsse}$, we have a basis generated by powers of $x_1$. Two different bases constructed at $\bsse$ and $\bsse'$  are related by a transition matrix which is regular together with its inverse on  the intersection of $\CO_{\bsse}$ and $\CO_{\bsse'}$. Hence we get a structure of a holomorphic bundle with fibers being $d_{\Lambda}$-dimensional vector spaces over the field $\CC$ and with base $\Se$. The existence of this holomorphic bundle is the same as saying that $\WAL$ is a projective $\CC[\sse]$-module. This is the so-called Serre-Swan correspondence \cite{serreswan}~\footnote{We are grateful to L.~Cassia for pointing out this relation to us}.

Because the base $\Se\simeq \CC^L$ is contractible, this bundle must be topologically trivial, that is, we can find $d_{\Lambda}$ global holomorphic sections forming a basis of the fiber at each point. A much more complicated question, already asked by Serre \cite{serreswan}, is whether we can choose these global sections to be polynomials of $\WAL$. A positive answer was given by the Quillen-Suslin theorem \cite{Quillen1976}. This theorem requires that $\WAL$ is a finitely generated $\CC[\sse]$-module, \ie that there exist finitely many elements $\tilde\bas_1,\ldots,\tilde\bas_{\tilde d}$ such that any element of $\WAL$ is their linear combination with coefficients from $\CC[\sse]$. This is easy to see to be the case. Take for instance the finite set of $\tilde d=L\times d_{\Lambda}$ monomials $x^n:=x_1^{n_1},x_2^{n_2},\ldots,x_L^{n_L}$ with $0\leq n_i< d_{\Lambda}$, where $x_i$ are the ones from the proof of Theorem~\ref{thm:counting}. Due to properness, $x_i$ satisfy a degree-$d_{\Lambda}$ equations with polynomial coefficients, \cf \eqref{firstGBeq}, and hence any higher powers of $x_i$ are expressible as linear combinations of the first $d_{\Lambda}$ powers.

The Quillen-Suslin theorem establishes that there are no non-trivial algebraic vector bundles over $\CC^L$ or equivalently by the Serre-Swan correspondence, that any finitely generated projective $\CC[\sse]$-module is free, with a basis given by the aforementioned global sections. Applied to $\WAL$, this is precisely the statement that it is a free module over $\CC[\sse]$, see Section \ref{sec:wronsk-algebra-free}.

\subsection{Non-symmetric functions}
\label{sec:cth}
Most of the time we work with only symmetric combinations $\sepa$ of inhomogeneities. However, the Baxter operators as explicit matrices acting on the spin chain have coefficients from $\CC[\theta]\equiv\CC[\inhom[1],\ldots,\inhom[L]]$. This prompts us to understand some properties of $\CC[\theta]$-modules as compared to $\CC[\sse]$-modules. Also, one can consider equations
\be\label{eq:A7}
\sepa(\inhom[1],\ldots,\inhom[L])=\sepa
\ee
as a toy model for \eqref{mastereq2} with $c_{\pa}=\inhom$ and $\SW_{\pa}(c)=\sepa$.

First, we demonstrate how to use the Gr\"obner basis techniques to conclude that the polynomial ring $\CC[\theta]$ is a free $\CC[\sse]$-module and count the number of solutions to \eqref{eq:A7}. To this end denote the elementary symmetric polynomial of degree $\pa$ in $k$ variables $\theta_{1},\ldots,\theta_{k}$ as $\sepa^{(k)}$. Being roots of $\prod\limits_{\pa=1}^k(u-\theta_{\pa})$, inhomogeneities satisfy the characteristic equations
\be\label{eq:ch}
s_k :=\theta_{k}^{k}+\sum_{\pa=1}^{k}(-1)^{n}\sepa^{(k)}\theta_k^{k-\pa}=0\,,\quad k=1,\ldots,L\,.
\ee
Now note that the polynomials $\sepa^{(k)}$ can be rewritten as polynomials in $\se{\pa'}\equiv \se{\pa'}^{(L)}$, with ${\pa'}\leq \pa$, and  $\theta_{m}$, $m>k$. As a result, $s_k$ become
\begin{subequations}
\label{eq:chall}
\begin{align}
\label{eq:ch2}
s_1&=\inhom[1]+(-\se{1}+\sum_{i=2}^L\theta_i)\,,\\
s_2&=\theta_2^2+\theta_2(-\se{1}+\sum_{i=3}^L\theta_i)+(\se{2}-(\se{1}-\sum_{i=3}^L\theta_i)\sum_{i=3}^L\theta_i-\sum_{3\leq i<j\leq L}\theta_i\theta_j)\,,\\
&\ldots\,,\nonumber\\
\label{eq:chn}
s_L&=\theta_{L}^{L}+\sum_{\pa=1}^{L-1}(-1)^{n}\sepa\,\inhom[L]^{L-\alpha}+(-1)^L \se{L}\,.
\end{align}
\end{subequations}

Now let's make a small formalisation: treat $\inhom$ and $\sepa$ as independent variables and consider an ideal $\mathcal{I}=\langle s_1,\ldots,s_L\rangle$ as an ideal in $\CC[\sse][\inhom[]]$. In the quotient ring $\CC[\sse][\inhom[]]/\CI$, excluding $\sepa$ in favour of $\inhom$ is easy. However, we are interested in the opposite -- to solve for $\inhom$ in terms of $\sepa$. We won't do this explicitly because this requires an algebraic closure but compute a Gr\"obner basis instead.

\begin{lemma} \label{thm:A3} The above-introduced polynomials $s_1,\ldots,s_L$ form the Gr\"obner basis of the ideal $\mathcal{I}=\langle s_1,\ldots,s_L\rangle$ \wrt a lexicographic order for which $\inhom[1]>\theta_2>\ldots>\inhom[L]>\sepa$, for any $\pa$.
\end{lemma}
\begin{proof}
First, by definition, $s_k$ generate the ideal $\CI$. Then, $s_k$ has $\theta_k^k$ as its leading monomial. Indeed, the other monomials are products of $\theta_{k}^{k'}$ with $k'<k$, powers of $\theta_{k'}$ with $k'>k$, and $\sepa$ which hence are lexicographically smaller than $\theta_k^k$. Finally, as the leading monomials enjoy the property ${\rm GCD}(\theta_k^k,\theta_{k'}^{k'})=1$, the S-polynomials between $s_{k}$ and $s_{k'}$ do not produce new relations and so this set of ideal generators is indeed a Gr\"obner basis.
\end{proof}
Conceptually the Gr\"obner basis tells us how to algorithmically find $\theta_{\alpha}$ from the values of their symmetric combinations $\sepa$. First one needs to solve \eqref{eq:chn} for fix $\theta_{L}$  ($L$ solutions), then one needs to substitute the found value of $\theta_{L}$ to the equation $s_{L-1}=0$ and solve it for $\theta_{L-1}$ ($L-1$ solutions) \etc.

Very similarly to the analysis of the Wronskian algebra, we note that the ring $\CC[\sse][\inhom[]]/\CI$ is isomorphic (over $\CC$) to $\CC[\inhom[]]$, but it is also naturally endowed with the structure of a $\CC[\sse]$-module. The computation of the Gr\"obner basis above immediately implies that this module is free and of rank $L!$. Indeed, the corresponding monomial basis are given by monomials $\theta_2^{n_2}\theta_3^{n_3}\ldots\theta_L^{n_L}$, with $n_\pa<\pa$. Any relations between these monomials is impossible precisely because $s_k$ form a Gr\"obner basis and leading monomials of $s_k$ do not belong to the monomial basis. Of course we know that $L!$ is an expected number, if to count with multiplicities: the equation $\theta^L-\se{1}\theta^{L-1}+\ldots+(-1)^L\se{L}=0$ has $L$ solutions, and any permutation of solutions is allowed as well.
\newline
\newline
Finally, let us extend Theorem~\ref{thm:isomorphism} to the case of non-symmetric polynomials in $\inhom$. Consider first the following example
\begin{example} 
Let the generators of a Wronskian algebra $\WA$ satisfy equations $c_1+c_2=\se{1}$, $c_1c_2=\se{2}$~\footnote{Up to normalisations, it is the $\hbar=0$ version of \eqref{eq:229}}, and (a hypothetical) Bethe algebra $\BA$ is generated by $2\times 2$ diagonal matrices $\hat c_1=\theta_1\times \Id_2$, $\hat c_2=\theta_2\times \Id_2$. These Wronskian and Bethe algebras are isomorphic as $\CC[\sse]$-modules. Let us now consider the extension of the Wronskian algebra $\WA^{\theta}\simeq\WA\otimes_{\CC[\sse]}\CC[\inhom[]]$, \ie consider generators satisfying $c_1+c_2=\theta_1+\theta_2$, $c_1c_2=\theta_1\theta_2$ and treat this algebra as a $\CC[\theta]$-module. This is a rank-two $\CC[\theta]$-module. In contrast, the Bethe algebra considered as a $\CC[\theta]$-module is of rank one.
\end{example}
By Lemma~\ref{thm:1} we actually know that generators of type $\inhom\times\Id$ cannot appear in polynomial combinations of $c_\pa$, and so the hypothetical Bethe algebra in the above example cannot exist. 

More generally, we can show that all polynomial relations satisfied $\hat c_{\pa}$, even with non-symmetric coefficients, should follow from the Wronskian algebra in the following sense. We can add non-symmetric polynomials by hand to the Wronskian algebra by considering $\WA_{\Lambda}^{\theta}\simeq\WAL\otimes_{\CC[\sse]}\CC[\inhom[]]$. Likewise, non-symmetric polynomials (times the identity operator) are not elements of the Bethe algebra by Lemma~\ref{thm:1}, and hence appending them as extra generators is also realised as $\BA_{\Lambda}^{\theta}\simeq\BAL\otimes_{\CC[\sse]}\CC[\inhom[]]$. Isomorphism between $\WA_{\Lambda}^{\theta}$ and $\BA_{\Lambda}^{\theta}$ as $\CC[\inhom[]]$-algebras is then obvious from the isomorphism between $\BAL$ and $\WAL$ as $\CC[\sse]$-algebras. We also note that $\WA_{\Lambda}^{\theta}$ and $\BA_{\Lambda}^{\theta}$ are free as $\CC[\inhom[]]$-modules and $\CC[\sse]$-modules as follows \eg from Lemma~\ref{thm:A3}.

\section{Cyclicity of representations}
\label{sec:cyclic}

The goal of this appendix is to build all the formalism necessary for the proof of Theorems~\ref{repisotheorem}~and~\ref{isotheorem}. There are two reasons why proving isomorphism of the specialised map $\varphi_{\binhom[]}$ \eqref{specmorphism} is problematic. First, setting $\inhom$ to numerical values, which is done for the Bethe algebra, is more restrictive than setting their symmetric combinations $\sepa$ to numerical values, which is done for the Wronskian algebra. Second, the specialisation procedure is actually native to the representation of an algebra, not to the algebra alone. Namely, we set to numerical values coefficients of a matrix which is more restrictive than setting to numerical values only the factors that multiply this matrix as a whole.

To overcome these difficulties, we want to ``rigidify'' the algebra isomorphism \eqref{isomorphism} by also proving isomorphism between certain representations of these algebras. As was already mentioned in Section~\ref{sec:specialisationofiso}, the only natural choice of a representation for the Wronskian algebra $\WAL$ is its regular representation. As for the Bethe algebra, it acts on the {\it a priori} unrelated physical space $\End(U_\Lambda)\otimes\CC[\inhom[]]$. These two representations are not isomorphic. This is why we need to introduce an alternative Yangian representation dubbed symmetrised representation. Using a cyclic vector argument we prove that its weight subspaces $\lVVS$ are indeed isomorphic to $\WAL$ as representations of $\BAL\simeq\WAL$, which resolves the second difficulty. This symmetrised representation has the virtue to manifestly depend only on symmetric combinations of inhomogeneities. We show that, under some explicit restriction on $\binhom[]$, its specialisation at a point $\bar\sse$ is isomorphic to the spin chain representation at a point $\binhom[]$ which resolves the first difficulty.

The discussed approach was developed in \cite{MTV} for $\glm$ spin chains. The below-presented generalisation to the supersymmetric case is conceptually very straightforward. The only difference, apart from the way we present the results, is in the proof of Lemma~\ref{thm:nonzero} which is in line with the ideas of Theorem~\ref{thm:isomorphism}.

\subsection{Symmetrised Yangian representation}

Consider the Yangian spin chain representation at point $\binhom[]=(\binhom[1],\ldots,\binhom[L])$ defined in Section~\ref{sec:Yangian}. We note that the order of inhomogeneities  in $\binhom[]$ is often superfluous. Indeed, the operator $r_{\pa}(\binhom[])=(\binhom-\binhom[\pa+1])\mathcal{P}_{\pa,\pa+1}+\hbar\Id$ satisfies \eqref{eq:braiding} evaluated at $\inhom[]=\binhom[]$.  If $\binhom\neq \binhom[\pa+1]\pm\hbar$, it is invertible and hence an intertwiner between the two representations that differ by the permutation of $\binhom$ and $\binhom[\pa+1]$. From here we conclude that the isomorphism class of the representation at point $\binhom[]$ is decided only by $\bsepa$ if there is no $\pa,\pa'$ such that $\binhom-\binhom[\pa']=\pm\hbar$.

More generally, the following facts hold for supersymmetric representations of  Yangians:

\begin{proposition} 
\label{thm:cyclicity}
If $(\binhom[1],\ldots,\binhom[L])$ satisfy $\binhom+\hbar\neq\binhom[\pa']$ for $\pa<\pa'$ then the spin chain representation of $\Yglmn$ at point $\binhom[]$ is cyclic with cyclic vector $\vve^+=\vve_1\otimes \ldots \otimes \vve_1$, where $\vve_1$ is the highest-weight vector of the defining $\glmn$ representation~\footnote{In the choice of ordering when bosonic indices are considered smaller than fermionic indices.}.\qed
\end{proposition}
This statement follows from Theorem~5.2 of \cite{Zhang:2014rda}.

$\vve^+$ is obviously a highest-weight vector of the Yangian representation, \ie  it satisfies the condition $T_{ij}\vve^+=0$ for $i<j$. Its weight is given by
\be
\label{hiwe}
T_{ii}\,\vve^+=Q_{\theta}(u+\hbar\,\delta_{i,1})\,\vve^+\,.
\ee

\begin{proposition} 
The spin chain representation at point $\binhom[]$ is irreducible if there are no $\pa,\pa'$ such that $\binhom-\binhom[\pa']=\hbar$.
\end{proposition}
For the $Y(\glm)$ case, this is a standard result appearing in the study of Kirillov-Reshetikhin modules \cite{Kirillov1990}. For $Y(\gl_{1|1})$ it was proven in \cite{Zhang:1995qt}, Theorem~5. For the $Y(\gl_{\gm|\gn})$ case, it apparently follows from \cite{Zhang:2014rda}, Proposition~5.4. But as this was not stated explicitly we give an alternative argument for irreducibility in the style of statistical lattice models.
\begin{proof}
Let $\CC^{\gm|\gn}$ be the Hilbert space of the $\pa$-th node of the spin chain, consider also the auxiliary space $\CC^{\gm|\gn}$ with basis vectors $\vve_{\sA}^{\rm aux}$, $\sA=1,2,\ldots,\gm+\gn$. The dual basis vectors of $\vve_{\sA}^{\rm aux}$ shall be denoted $\vve^{\sA}_{\rm aux}$. Define the Lax matrix acting on the tensor product of the mentioned spaces as $\mathcal{L}(u-\inhom)=(u-\inhom)\Id+\hbar\, P$, where $P$ is the graded permutation. In the notations of \eqref{eq:ev1}, $\mathcal{L}(u-\theta_\pa):={(u-\theta_{\pa})}\sum\limits_{\alpha,\beta}ev_{\theta_\pa}(t_{\alpha\beta})\otimes(\vve_\beta^{\rm aux}\otimes\vve_{\rm aux}^{\alpha})$. The key property we use is that the Lax matrix becomes, up to normalisation, the graded permutation if $u=\inhom$.

Introduce $V^{\rm aux}$ -- the tensor product of $L$ auxiliary spaces spanned by $\vve_{A}^{\rm aux}=\vve_{\sA_1}^{\rm aux}\otimes\ldots\otimes\vve_{\sA_L}^{\rm aux}$ and define  $B=\sum\limits_{A}\vve_{A}^{\rm aux}T_{1\sA_L}(\binhom[L])\times\ldots T_{1\sA_2}(\binhom[2])\times T_{1\sA_1}(\binhom[1])$. Given the above-mentioned key property, $B$ maps $V$ to $\vve^+\otimes V^{\rm aux}$ which is easiest to see by a graphical representation of how $B$ acts:
\be\label{eq:B2}
\includegraphics[width=0.3\textwidth]{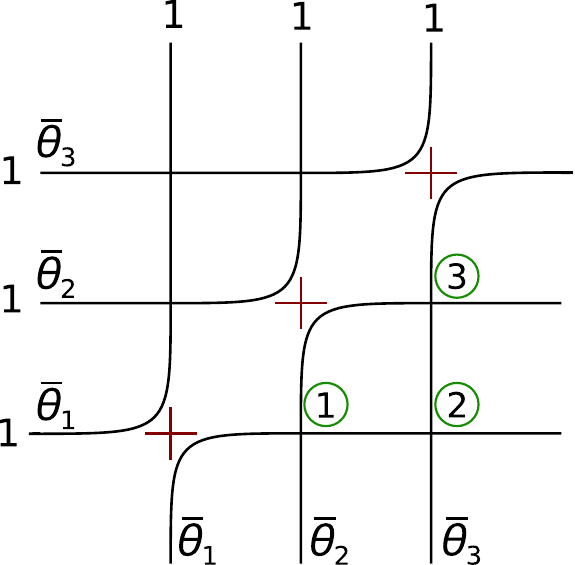}\,.
\ee
Here each vertical direction corresponds to a node $\CC^{\gm|\gn}$ of the spin chain  and each horizontal direction corresponds to a tensor factor $\CC^{\gm|\gn}$ of $V^{\rm aux}$. Intersections are the places where the Lax matrices should be applied (considered as maps from South-West to North-East spaces). Red crosses are the places where the corresponding Lax matrix becomes the permutation.

The map $B$ is also invertible. Indeed, up to a non-zero factor it reduces to an ordered product of $\frac{L(L-1)}{2}$ Lax matrices, \eg these are the three Lax matrices marked by the encircled numbers in the image above. Each of these Lax matrices is invertible since $\binhom-\binhom[\pa']\neq\hbar$.

Because $B$ is invertible, for any $v\in V$ one can find a vector $v^*\in (V^{\rm aux})^*$ such that $(v^*,B)v=\vve^+$. Then irreducibility follows from Proposition~\ref{thm:cyclicity}.
\end{proof}

\paragraph{Remark} We expect that a further refinement of argument \eqref{eq:B2} can be used to prove  Propostion \ref{thm:cyclicity} about cyclicity.

\vspace{.5cm}

For finite-dimensional irreducible $\Yglmn$ representations, the highest-weight vector exists and is unique and the representation is fully determined, up to an isomorphism, by the vector's weight \cite{Zhang:1995uh}. Note that the weight of $\vve^+$  depends only on symmetric combinations of inhomogeneities according to \eqref{hiwe}. Then we can consider the induced representation from $\vve^+$ which is isomorphic to the spin chain one by the above-mentioned uniqueness but, in contrast to the spin chain realisation is manifestly invariant under permutations of inhomogeneities. 
\newline
\newline
We shall now introduce a different permutation-invariant realisation which does not require the irreducibility argument and will be formulated for inhomogeneities being abstract variables.

\paragraph{Yangian centraliser}
Consider the vector space $\mathcal{V}\simeq (\CC^{\gm|\gn})^{\otimes L}\otimes\CC[\inhom[1],\ldots,\inhom[L]]$ on which the Yangian representation $ev_{\inhom[]}$ \eqref{eq:steps} is realised. An interesting question is what is the centraliser of the Yangian action on $\mathcal{V}$.

Define operators $S_{\pa}$ acting on $\mathcal{V}$ by
\be
\label{eq:Spa}
S_{\pa}={\mathcal{P}_{\pa,\pa+1}}\Pi_{\pa,\pa+1}-\frac{\hbar}{\inhom-\inhom[\pa+1]}\left(\Pi_{\pa,\pa+1}-\Id\right)\,,
\ee
where $\mathcal{P}_{\pa,\pa+1}$ is the graded permutation in $(\CC^{\gm|\gn})^{\otimes L}$, and $\Pi_{\pa,\pa+1}$ permutes variables $\inhom$ and $\inhom[\pa+1]$. These permutations were already used in the proof of Lemma~\ref{thm:1} on page \pageref{thm:1}.

Although $S_{\pa}$ contains $\inhom$ in denominator, its action on polynomials in $\inhom$ yields again polynomials and hence its action on $\mathcal{V}$ is well-defined.

\begin{lemma}
\label{thm:2}
For $\pa=1,\ldots,L-1$, $S_{\pa}$ commutes with the $\Yglmn$ action, (i.e. $[S_\pa,T_{\sA\sB}]=0$) and they form a representation of the symmetric group $\SG_L$ on $\mathcal{V}$.
\end{lemma}
\begin{proof}
The commutativity follows from $(\inhom-\inhom[\pa+1])S_{\pa}=-\Pi_{\pa,\pa+1}r_{\pa}+\hbar \Id$ and \eqref{eq:braiding}. Then, by explicit computation one checks $S_{\pa}^2=\Id$ and $(S_{\pa}S_{\pa+1})^3=\Id$ -- the defining relations of $\SG_L$.
\end{proof}

\paragraph{dAHA}\label{pdAHA} In the limit $\hbar\to 0$, $\SG_{L}$ becomes an explicit permutation defined on a graded space that commutes with the action of $\glmn$. Hence $\hbar\neq 0$ should be considered as a generalisation of the Schur-Weyl duality to the case of the Yangian algebra. This statement was made mathematically precise for the bosonic $\glm$ case \cite{Drinfeld_1986, Arakawa_1999, Davydov_2011}: $(S_{\pa})_{1\leq\pa\leq L-1}$ together with $(\theta_{\pa}\times\Id)_{1\leq\pa\leq L}$ form a representation of $\mathcal{H}_L$, the {\it degenerate affine Hecke algebra} (dAHA) on $L$ sites. Moreover, the dAHA and the Yangian  form a dual pair -- they are maximal mutual centralisers of one another when acting on $\mathcal{V}$. More formally, one can view $\mathcal{V}$ as the tensor product of $\SG_{L}$-modules $\mathcal{H}_L\otimes_{\SG_{L}}(\CC^\gm)^{\otimes L}$ where $\SG_{L}$ acts on $\mathcal{H}_L$ as a subalgebra and on $(\CC^\gm)^{\otimes L}$ by permutation of tensor factors. This point of view is conceptually interesting because to generalise Schur-Weyl duality to supersymmetric Yangians, one does not need to change the defining relations of the dAHA $\mathcal{H}_L$ but simply to replace the usual action of $\SG_{L}$ on $(\CC^\gn)^{\otimes L}$ by the graded action on $(\CC^{\gm|\gn})^{\otimes L}$ as in \eqref{eq:Spa}. The full mathematical treatment (for the affine Hecke algebra~\footnote{Recall that the degenerate affine Hecke algebra and the Yangian can be obtained as $q\simeq 1+\hbar$ expansions of respectively the affine Hecke algebra $\mathcal{H}_L(q)$ and $U_q(\hat{\gl}_{\gm|\gn})$.}) can be found in \cite{Flicker_2018}.

We are only going to use the slightly weaker statement of Lemma~\ref{thm:2} that $S_\pa$ commute with the Yangian action.

\paragraph{$\CC[\sse]$-Yangian module} Define $\mathcal{V}^\SG\subset \mathcal{V}$ as the subspace of $S_{\pa}$-invariant vectors. As $[S_{\pa},T_{AB}]=0$, the Yangian action is well-defined on $\mathcal{V}^\SG$. Multiplication by symmetric polynomials is also well-defined on $\mathcal{V}^\SG$ and, moreover, $\CC[\sse]\times\Id$ belongs to $ev_{\inhom[]}(\Yglmn)$ due to \eqref{mastereq2}. Hence we shall call this action of the Yangian on $\mathcal{V}^\SG$ the {\it symmetrised Yangian representation}. Note that $\mathcal{V}^\SG$ is also naturally a $\CC[\sse]$-module.

\subsection{Symmetrised Bethe modules and their characters}
\label{sec:Becha}
Since $\SG_L$ commutes with the global $\glmn$ action it is consistent to define $\mathcal{V}^\SG_{\Lambda}=\mathcal{V}^\SG\cap(\VTw\otimes\CC[\inhom[1],\ldots,\inhom[L]])$ -- the weight $\Lambda$ subspace of $\mathcal{V}^\SG$ -- and $\mathcal{V}^{\SG+}_{\Lambda}=\mathcal{V}^\SG\cap(\VTl\otimes\CC[\inhom[1],\ldots,\inhom[L]])\subset\mathcal{V}^\SG_{\Lambda}$ -- the subspace of $\glmn$ highest-weight vectors corresponding to the Young diagram $\LD$. These spaces are also naturally $\CC[\chi]$-modules.

\paragraph{Characters}
For an element $v$ of $\mathcal{V}^\SG$, define its degree as the maximal degree of the monomials in $\inhom$'s which occur in $v$. Define $\mathcal{F}_k \mathcal{V}^\SG$ as the space of all vectors of degree less or equal to $k$. Finally, define the character
\be
\label{chardnus}
{\rm ch}(\mathcal{V}^\SG)=\sum_{k=0}^\infty \left(\dim \mathcal{F}_k/\mathcal{F}_{k-1}\right)t^k\,.
\ee
Since the $\glmn$ action does not change the degree, we can also define in an analogous way ${\rm ch}(\mathcal{V}^\SG_{\Lambda})$ and ${\rm ch}(\mathcal{V}^{\SG+}_{\Lambda})$.

\begin{proposition}\label{thm:ch1}
$\lVTwS$ is a free $\CC[\chi]$-module of rank $\binom{L}{\lambda_1\ldots\nu_{\gn}}$ and its character is given by
\be
\label{ch1}
{\rm ch}(\mathcal{V}^\SG_{\Lambda})=t^{\Upsilon_\Lambda}\prod_{a=1}^{\gm}\prod_{k=1}^{\lambda_a}\frac 1{1-t^k}\prod_{i=1}^{\gn}\prod_{k=1}^{\nu_i}\frac {1}{1-t^k}\,,
\ee
where $\Upsilon_{\Lambda}:=\sum\limits_{i=1}^{\gn}\frac{\nu_i(\nu_i-1)}{2}$.
\end{proposition}
\begin{proof}
$\mathcal{V}^\SG_{\Lambda}$ is the image of $\VTw\otimes\CC[\theta]$ by the projector $p:=\frac{1}{L!}\sum\limits_{\sigma\in \SG_L}S_\sigma$, where the symmetry group acts with $S_\sigma$ on $\VTw\otimes\CC[\theta]$ as generated from $S_\pa$ \eqref{eq:Spa}. Since the construction is polynomial in $\hbar$ and the $\hbar$-term of $S_\pa$ lowers degrees of polynomials the proposition is true iff it is true for $\hbar=0$. Hence we will consider only the $\hbar=0$ case. Every element $\sigma\in\SG_L$ is then represented as $S_\sigma=\lP_\sigma\Pi_\sigma$, where $\lP_\sigma$ is the graded permutation acting on $\VTw$ and $\Pi_\sigma$ is the ordinary permutation acting on $\CC[\theta]$.

Denote by $\vve_\sA$ the standard basis vectors of $\CC^{\gm|\gn}$ defined by $\EE_{\sB\sB}\vve_{\sA}=\delta_{\sA \sB}\vve_{\sA}$. The standard spin basis of $\VTw$ is indexed by tuples $I=(i_1,\ldots,i_L)$ such that $|\{k : i_k=a\}|=\lambda_a$ and $|\{k : i_k=i\}|=\nu_i$ corresponding to vectors $\vve_I:=\otimes_{\pa=1}^L\vve_{i_\pa}$.

To get $\mathcal{V}^\SG_{\Lambda}$, it is enough to consider the image of $\VTo\otimes\CC[\theta]$ by $p$ with
\begin{align}
\label{eq:VTo}
\VTo=\underbrace{\vve_1\otimes\ldots\otimes \vve_1}_{\lambda_1}\otimes\ldots\otimes\underbrace{\vve_\gm\otimes\ldots\otimes \vve_\gm}_{\lambda_\gm}\otimes\underbrace{\vve_{\gm+1}\otimes\ldots\otimes \vve_{\gm+1}}_{\nu_1}\otimes\ldots\otimes\underbrace{\vve_{\gm+\gn}\otimes\ldots\otimes \vve_{\gm+\gn}}_{\nu_\gn}\,.
\end{align}
Indeed, for all $I$ we can always find $\sigma\in \SG_L$ such that $\lP_\sigma\cdot \vve_I=\vve_{\sigma(I)}=\pm \VTo$.

Denote by $\SG_\Lambda:=\prod_{a=1}^\gm\SG_{\lambda_a}\prod_{i=1}^\gn\SG_{\nu_i}$ the stabiliser of $\VTo$ and  by $H:=\SG_L/\SG_\Lambda$ the space of orbits with respect to the right group multiplication. Then the projection by $p$ can be represented as
\be
\label{eq:A6}
\lVTwS\simeq p(\VTo\otimes\CC[\inhom[]])=\frac 1{L!}\sum_{[\sigma]\in H} S_{[\sigma]}\cdot\left( \VTo\otimes R_B\cdot C_F\cdot \CC[\theta]\right)\,,
\ee
where $R_B:=\sum\limits_{\sigma\in \prod_{a=1}^\gm\SG_{\lambda_a}}\Pi_\sigma$ and $C_F:=\sum\limits_{\sigma\in \prod_{i=1}^\gn\SG_{\nu_i}}(-1)^{|\sigma|}\Pi_\sigma$.

To decide about linear independence in $\lVTwS$, it is enough to consider one term in the sum $\sum_{[\sigma]\in H} S_{[\sigma]}$, \eg $[\sigma]=[\Id]$ since different terms would be proportional to $\lP_{\sigma}\VTo$ which are linearly independent in $\VTw$.  Also, $R_BC_F$ commutes with symmetric polynomials and hence we conclude that $\lVTwS$ and $R_BC_F\cdot\CC[\inhom[]]$ are isomorphic as $\CC[\sse]$-modules. 

It is easy to describe $R_BC_F\cdot\CC[\inhom[]]$: it is spanned by $Q_{\Lambda}\times{\rm Sym}_{\lambda_1}\times\ldots\times {\rm Sym}_{\nu_{\gm}}$, where $Q_\Lambda:=\prod\limits_{i=1}^\gn\prod\limits_{1\leq k<l\leq\nu_i}(\theta_{j_i+k}-\theta_{j_i+l})$ with $j_i:=\sum_{a=1}^\gm \lambda_a+\sum_{s=1}^{i-1}\nu_s$, and ${\rm Sym}_{k}$ is the space of symmetric polynomials in $k$ variables. $Q_{\Lambda}\times{\rm Sym}_{\lambda_1}\times\ldots\times {\rm Sym}_{\nu_{\gm}}$ is free under the action of $\CC[\sse]$, see Appendix~\ref{sec:cth}, and its character is \eqref{ch1}. Rank is computed from $\frac{\ch(\lVTwS)}{\ch(\CC[\sse])}|_{t=1}$.
\end{proof}

\begin{corollary}
\label{thm:a6}
$\mathcal{V}^\SG=\bigoplus_\Lambda\mathcal{V}^{\SG}_{\Lambda}$ is a free $\CC[\chi]$-module of rank $(\gm+\gn)^L$.
\qed
\end{corollary}

\begin{proposition} \label{thm:ch2}
$\mathcal{V}^{\SG+}_\Lambda$ is a free $\CC[\chi]$-module of rank $\frac{L!}{\prod\limits_{(a,s)\in \LD} h_{a,s}}=\dim \VTl$, where $h_{a,s}$ is the hook length at box $(a,s)$. Its character is given by
\be
\label{ch2}
{\rm ch}(\mathcal{V}^{\SG+}_{\Lambda})=t^{\Upsilon_{\Lambda}^+}\prod_{(a,s)\in\LD}\frac 1{1-t^{h_{a,s}}}\,,
\ee
where $\Upsilon_{\Lambda}^+=\sum_{s=1}^{\lambda_1}\frac{h_s(h_s-1)}{2}$ with $h_s$ being the height of the $s$-th column of $\LD$.
\end{proposition}

Note that $\Upsilon_{\Lambda}^+\geq\Upsilon_{\Lambda}$ which is consistent with $\lVTlS\subset\lVTwS$.
\begin{proof} As in the previous proof, it is enough to consider $\hbar=0$.

  The space $\lVTlS$ can be constructed from $V\otimes\CC[\inhom[]]$ as follows. Take the standard Young tableau $\SYT$ which is obtained by filling the shape $\LD$ first by filling the boxes defining the weight $\lambda_1$, then $\lambda_2$, $\ldots$, then $\nu_{\gn}$. For instance, for $[\lambda_1|\nu_1,\nu_2]=[4|2,2]$,  $\SYT=\raisebox{-1em}{{\scriptsize\young(1234,57,68)}}$. Take the normalised~\footnote{We normalise all symmetrisations/antisymmetrisations such that they are projectors.} Young symmetriser $S_\SYT\propto R_\SYT C_\SYT$, where $R_\SYT$ is the symmetrisation over rows and $C_\SYT$ is the antisymmetrisation over columns. Then $\lVTlS$ is the image by the projector $p$ of $S_{R_\SYT C_\SYT} \VTo\otimes\CC[\inhom[]]$. Singling out a special vector $\VTo$ and (correlated to it) tableau $\SYT$ is enough to construct $\lVTlS$ because $p$ sums over all permutations. Using this feature of $p$ again, and by repeating the same construction as \eqref{eq:A6}, one gets
\be
\lVTlS&\simeq& p(S_{R_\SYT C_\SYT} \VTo\otimes\CC[\inhom[]])=p( \VTo\otimes\Pi_{C_\SYT R_\SYT}\CC[\inhom[]])
\nonumber\\
&=&\frac 1{L!}\sum_{[\sigma]\in H} S_{[\sigma]}\cdot\left( \VTo\otimes R_B C_F\Pi_{C_\SYT R_\SYT}\cdot \CC[\theta]\right)\,,
\ee
and so $\lVTlS$, as a $\CC[\sse]$-module, is isomorphic to $R_B C_F\Pi_{C_\SYT R_\SYT}\cdot \CC[\theta]$. We can omit $R_BC_F$ as it has zero kernel when acting on $\Pi_{C_\SYT R_\SYT}\cdot \CC[\theta]$. Indeed, $\Pi_{C_\SYT R_\SYT}R_B C_F\Pi_{C_\SYT R_\SYT}=\Pi_{C_\SYT R_\SYT}$. But the module $\Pi_{C_\SYT R_\SYT}\cdot \CC[\theta]$ is the standard application of (reversed) Young symmetriser to a polynomial ring which is well understood, see \eg \cite{macdonald1998symmetric, kostkafoulkes, CHARI2006928, Mukhin_2009.2}. It is a free $\CC[\sse]$-module with character given by \eqref{ch2}~\footnote{This character is an important combinatorial object: $\frac{{\rm ch}(\mathcal{V}^{\SG+}_{\Lambda})}{{\rm ch}(\CC[\sse])}$ is the Kostka-Foulkes polynomial $K_{\mu\nu}(t)$ with $\mu=\LD'$ and $\nu=(1^L)$, see \eg \cite{macdonald1998symmetric, kostkafoulkes}.}. Again, the rank is computed from $\frac{\ch(\lVTlS)}{\ch(\CC[\sse])}|_{t=1}$.
\end{proof}

\paragraph{Young diagram dependence}\label{p65} One may wonder how comes that the character \eqref{ch2} and, in fact, the $\CC[\sse]$-isomorphism class of $\lVTlS$ do not depend on $\glmn$ but only on the Young diagram $\LD$. To understand this property, let us extend the underlying symmetry algebra from $\glmn$ to $\gl_{\gm'|\gn'}$, where $\gm'={\rm max}(\gm,{\rm h}_{\LD})$ and $\gn'={\rm max}(\gn,\IPart_1)$.  In addition to $a=1,\ldots,\gm$, $i=\hat 1,\ldots,\hat \gn$, introduce also $\mathfrak{z}$ to label the new indices that appeared due to the extension. In the order $a<i<\mathfrak{z}$, the highest-weight vectors do not involve $v_\mathfrak{z}$ and hence $\lVTlS$ for the extended system is isomorphic to the one of the $\glmn$ system. Now we perform a chain of fermionic duality transformations (odd Weyl reflections) to get any other order in the set $a,i,\mathfrak{z}$ (the order between separately bosonic and fermionic indices will be preserved though). The procedure is described for instance in \cite{Gunaydin:2017lhg}. It changes the highest-weight vectors, \ie the actual embedding of $V_{\Lambda}^+$ inside $V$ is modified, but this change is performed by acting with elements of the global $\gl_{\gm'|\gn'}\subset Y(\gl_{\gm'|\gn'})$ that commutes with the dAHA and in particular with the action of $\CC[\sse]$. Since the procedure is invertible it establishes a $\CC[\sse]$-isomorphism between two spaces $\lVTlS$ that differ by the choice of the order defining the highest-weight vector. The order is bijected to a Manhattan-type path (\eg the one in the Example on page~\pageref{exmp:p44}), and only those indices that belong to the Young diagram part of the path participate in the highest-weight vectors. This last observation allows us to choose $\gm',\gn'$ to be any pair such that $(\gm',\gn')$ lies on the boundary of $\LD$ (black/red dots of Figure~\ref{fig:ShortLong}) or outside of $\LD$.

The dependence of $\lVTlS$, as a $\CC[\sse]$-module, on the Young diagram alone parallels results of Section~\ref{sec:variouspar} that show that the isomorphism class of the twist-less $\BAL$, as a $\CC[\sse]$-algebra, only depends on the Young diagram. This is of course not a coincidence because the twist-less $\BAL$ and $\lVTlS$ are isomorphic as $\CC[\sse]$-modules which follows from the results of the next subsection.

\subsection{Cyclicity of symmetrised Bethe modules}
We shall use the notation $\lVVS$ to cover both $\lVTwS$ and $\lVTlS$ in the discussion below. 

We know that $\WAL$ and $\BAL$ algebras are isomorphic as $\CC[\sse]$-algebras. The goal of this subsection is to show that the regular representation $\WAL$ of the Wronskian algebra and the symmetrised Bethe module $\lVVS$ are isomorphic as $\CC[\sse]$-modules. To do so we need a map from $\WAL$ to $\lVVS$ commuting with the action of $\WAL\equiv\BAL$. A standard approach is to take a vector $\omega\in U_\Lambda\otimes\CC[\inhom[]]$ and to consider the morphism of representations

\be
\begin{array}{cccl}
\psi_\omega : &\WAL & \longrightarrow & U_\Lambda\otimes\CC[\inhom[]]\,, \\
~& c_\pa & \longmapsto & \hat{c}_\pa\,\omega\,
\end{array}\,.
\ee
Of course this map has no reason to be an isomorphism. Nevertheless we can prove the following.
\begin{lemma}
\label{thm:nonzero}
For any non-zero vector $\omega\in U_\Lambda\otimes\CC[\inhom[]]$, $\psi_\omega$ is injective.
\end{lemma}
\begin{proof} Take $P\in\WAL$ such that $\hat{P}\,\omega=0$. As before, use WBE to write $P$ as a polynomial in $c_\pa$ only. Around a regular point $\theta$, we know from our previous results that the Bethe algebra can be fully diagonalised and that the spectrum of the operators $\hat{c}_\pa$ are exactly the solutions of the Bethe equations. Denote by $W(\theta)$ the corresponding $\theta$-dependent change of basis that diagonalises $\hat c_\pa$. At least one of the components of $W(\theta)\,\omega$ has to be non-zero at $\theta$ and hence in some $L$-dimensional ball $\mathcal{O}$ around $\theta$. Then $\hat{P}\omega=0$ implies that for one of the solutions $c_\pa(\theta)$, $P(c_\pa(\theta))=0$ for $\theta\in\mathcal{O}$. Since $c_\pa$ is a local diffeomorphism, $P$ vanishes on a $L$-dimensional ball and thus $P=0$.
\end{proof}

We hence see that for all $\omega$, $\psi_\omega$ is an isomorphism on its image. Let us take a very precise $\omega$ -- the vector of the smallest degree~\footnote{It is unique up to a normalisation as follows from \eqref{ch1} and \eqref{ch2}.} (as a polynomial in $\inhom$) that belongs to $\lVVS$.

\begin{lemma}
\label{thm:chiiso}
$\psi_{\omega}(\WAL)=\lVVS$. Therefore $\psi_{\omega}$ is an isomorphism of $\CC[\sse]$-modules between $\WAL$ and $\lVVS$.
\end{lemma}
\begin{proof}
Injectivity is proven by Lemma~\ref{thm:nonzero}. The fact that $\psi_{\omega}$ preserves the action of $\CC[\sse]$ is obvious from the definition. Note also that $\psi_{\omega}$ just increases the degree by $\Upsilon_{\Lambda}$ (by $\Upsilon_{\Lambda}^+$ in the non-twisted case) and otherwise preserves the natural filtrations on $\WAL$ and $\lVVS$. Therefore surjectivity of $\psi_{\omega}$ follows from the comparison of the corresponding characters computed in Sections~\ref{sec:Hilbert}~and~\ref{sec:Becha}.
\end{proof}

\subsection{Specialisations}
\label{sec:isosp}

Let us summarise what has been done so far. On one side, we have the Wronskian algebra $\WAL$ acting on itself \via the regular representation. On the other side, we have the Bethe algebra $\BAL$ acting on the space $\lVVS$. Moreover the two couples $(\WAL, \WAL)$ and $(\BAL, \lVVS)$ are isomorphic \via $(\varphi,\psi_\omega)$.

Now consider the ideal $\CI :=\langle\SW_1-\bse1,\ldots\SW_L-\bse{L}\rangle$ of $\WAL$ and its image $\hat{\CI}$ in $\BAL$ under $\varphi$. Automatically, the isomorphisms $(\varphi,\psi_\omega)$ will induce isomorphisms between $(\WAL/\CI, \WAL/\CI)$ and $(\BAL/\hat{\CI}, \lVVS/\psi_\omega(\CI))$. Moreover, as shown in Lemma~\ref{thm:trivialsp} $\psi_\omega(\CI)=\CJ\cdot\lVVS$, where $\CJ:=\langle \sse-\bsse\rangle\subset\CC[\sse]$. Denote $\BAL(\bsse):=\BAL/\hat{\CI}$ and recall that $\WAL(\bsse):=\WAL/\CI$.

We also have a third pair in this correspondence, namely $\BAL(\binhom[])$, the Bethe algebra evaluated at $\binhom[]$, acting on $U_\Lambda$. Our final goal is to show that $\varphi_{\binhom[]} : \WAL(\bsse)\simeq\BAL(\binhom[])$. This is equivalent to showing that $\BAL(\binhom[])\simeq\BAL(\bsse)$. Let us emphasise that {\it a priory} $\BAL(\bar\sse)$ and $\BAL(\binhom[])$ are two different objects.

\begin{example}
Consider $\WA$ and $\CB^{\rm bad}$ from the example on page~\pageref{ex:gb}. $\CB^{\rm bad}$ acts on the space $\mathcal{V}=\CC^2\otimes \CC[\inhom[1],\theta_2]$. $\lV$ is a $\CC[\se{1},\se{2}]$-module of rank four, we can take $\vtwo{1}{0}$, $\inhom[1]\vtwo 10$, $\vtwo 01$, $\inhom[2]\vtwo 01$ as its basis elements. In this basis
\be
\label{fullc1}
\hat c_1^{\rm bad}=
\begin{pmatrix}
0 & -\se 2 & 0 & 0  \\
1 & \se 1 & 0 & 0\\
0 & 0 & 0 & -\se 2 \\
0 & 0 & 1 & \se 1
\end{pmatrix}\,.
\ee
Take a vector $\omega=A\vtwo 10+ B\vtwo 01$, where $A,B$ are some polynomials in $\se{1},\se{2}$. Then $\lU^\SG:=\psi_{\omega}(\WA)$ is a $\CC[\se{1},\se{2}]$-module of rank two spanned by $\xi_1:=\omega$, and $\xi_2:=A\,\inhom[1]\vtwo 10+B\,\inhom[2]\vtwo 01$. $\hat c_1^{\rm bad}\in\CB^{\rm bad}$ acting on $\lU^\SG$ is $\mtwo 0{-\se{2}}{1}{\se{1}}$  in the basis $\xi_1,\xi_2$.

Specialisation $\CB^{\rm bad}(\bsse)$ is two-dimensional for any $\bar\sse$ and is clearly $\CC$-isomorphic to $\WA(\bsse)$, in particular $\hat c_1^{\rm bad}(\bsse)=\mtwo 0{-\bse{2}}{1}{\bse{1}}$. Note that the statement is completely independent of the choice of $\omega$. It holds even if $A(\bse{1},\bse{2})=A(\bse{1},\bse{2})=0$ because $\xi_i\notin \psi_\omega(\CI):=\langle \sse-\bar\sse\rangle\lU^\SG$, $i=1,2$.
\end{example}

We chose $\CB^{\rm bad}$ in the example above to explicitly demonstrate that $\BAL(\binhom[])$ and $\BAL(\bsse)$ can be in principle non-isomorphic.

Instead of showing isomorphism between $\BAL(\bar\sse)$ and $\BAL(\binhom[])$ directly, let us show isomorphism between $\lVVS/\CJ\cdot\lVVS$ and $U_\Lambda$. Since these spaces both carry representations of the Bethe algebra, if one can find an isomorphism commuting with these actions, it would automatically imply $\BAL(\bar\sse)\equiv\BAL(\binhom[])$ as is argued in Section~\ref{sec:specialisationofiso}. The advantage of this strategy is that we can leverage Yangian representation theory to prove such an isomorphism.

To relate the symmetrised Yangian representation at point $\bsse$ and the spin chain Yangian representation at point $\binhom[]$, recall that $\mathcal{V}^{\SG}$ is a subspace of $\mathcal{V}:=(\CC^{\gm|\gn})^{\otimes L}\otimes\CC[\inhom[1],\ldots,\inhom[L]]$ and so we can define a map $\mathrm{Ev}_{\binhom[]} :\mathcal{V}^{\SG} \rightarrow (\CC^{\gm|\gn})^{\otimes L} $ simply by evaluating all vectors at $\binhom[]$. Since $\CJ\cdot\mathcal{V}^{\SG}\subset\mathrm{Ker}~\mathrm{Ev}_{\binhom[]}$, this induces a well-defined map 
\be
\mathrm{ev}_{\binhom[]}:\mathcal{V}^{\SG}(\bsse)\to (\CC^{\gm|\gn})^{\otimes L}\,,
\ee
where $\mathcal{V}^{\SG}(\bsse):=\mathcal{V}^{\SG}/\CJ\cdot\mathcal{V}^{\SG}$. Concretely this just means the following: take a class $[v]\in\mathcal{V}^{\SG}(\bsse)$, represent it by some $v\in\mathcal{V}$ and evaluate it at $\binhom[]$.

Note that $\mathcal{V}^{\SG}(\bsse)$ is the space where the symmetrised Yangian representation at point $\bsse$ is realised and $(\CC^{\gm|\gn})^{\otimes L}$ is the Hilbert space of the spin chain. We can realise on it the spin chain Yangian representation at point $\binhom[]$.

We are ready to formulate the main conceptual result of this appendix which is  Proposition~3.5 of \cite{MTV} generalised to the supersymmetric case. 
\begin{theorem}
\label{thm:iso}
Let $\binhom[]=(\binhom[1],\ldots,\binhom[L])$ be a solution of equations $\se{\pa}(\theta)=\bar{\se{\pa}}$ such that $\binhom+\hbar\neq\binhom[\pa']$ for $\pa<\pa'$. Then $\mathrm{ev}_{\binhom[]}$ is an isomorphism of $\Yglmn$ representations.
\end{theorem}

\begin{proof}
Since $\mathrm{ev}_{\binhom[]}$ commutes with the Yangian action, it defines a homomorphism from the symmetrised representation to the spin chain representation. Assume $\gm\geq 1$ and consider the vector $\vve^+:=\vve_1^{\otimes L}\in\mathcal{V}^{\SG}(\bsse)$. Since $\mathrm{ev}_{\binhom[]}:\vve^+\mapsto \vve^+$, and $\vve^+$ is a cyclic vector of the spin chain module by Theorem~\ref{thm:cyclicity}, $\mathrm{ev}_{\binhom[]}$ is surjective. As $\mathcal{V}^{\SG}({\bsse})$ and $(\CC^{\gm|\gn})^{\otimes L}$ are of the same dimension by Corollary~\ref{thm:a6}, $\mathrm{ev}_{\binhom[]}$ is an isomorphism.

If $\gm=0$, one can check that $\mathcal{V}^{\SG}$ contains the vector $\prod_{1\leq \pa<\pa'\leq L}(\inhom - \inhom[\pa']+\hbar)\vve^+$, whose image under $\mathrm{ev}_{\binhom[]}$ is nonzero as long as $\binhom+\hbar\neq\binhom[\pa']$ for $\pa<\pa'$. The rest of the proof is the same.
\end{proof}

\subsection{An explicit case study}
\label{sec:OE}
Finally, we provide a concrete comprehensive example to illustrate the above-discussed ideas. 

\paragraph{Explicit Q-operators} Consider the $L=3$ $\gls{2}$ spin chain in the absence of twist. By convention, we use the following basis of $\left(\mathbb{C}^2\right)^{\otimes 3}$
\begin{align}
 \left\{
\left|\downarrow\downarrow\downarrow\right>, \left|\uparrow\downarrow\downarrow\right>, \left|\downarrow\uparrow\downarrow\right>, \left|\downarrow\downarrow\uparrow\right>, \left|\downarrow\uparrow\uparrow\right>, \left|\uparrow\downarrow\uparrow\right>, \left|\uparrow\uparrow\downarrow\right>, \left|\uparrow\uparrow\uparrow\right>\right\}\,.
\end{align}
In this basis, the periodic (``twist-less'') limit of the $Q$-operators
is:
\begin{align}
  Q_{\emptyset}&=1,&
  Q_1&=\begin{pmatrix}
1&\\
&M_{1}&\\
&&M_{1}&\\
&&&1
\end{pmatrix},&Q_{12}&=\prod_{i=1}^3 (u-\theta_i)
\,,
\end{align}
where $M_{1}$ is the following $3\times 3$ block matrix
\begin{multline}
\label{eq:M1}  M_{1}=\underbrace{\frac 1 3 \begin{pmatrix} 1&1&1\\1&1&1\\1&1&1
  \end{pmatrix}}_{P_s}+\left(u-2\frac {\inhom[1]+\theta_2+\theta_3} 3\right) \underbrace{\frac 1 3 \begin{pmatrix} 2&-1&-1\\-1&2&-1\\-1&-1&2
  \end{pmatrix}}_{P_h}\\+\frac 1 6 \,\times\,\underbrace{
\begin{pmatrix}2\theta_2+2\theta_3&-\hbar-2\theta_3&\hbar-2\theta_2\\
\hbar-2\theta_3&2\inhom[1]+2\theta_3&-\hbar-2\inhom[1]\\
-\hbar-2\theta_2&\hbar-2\inhom[1]&2\inhom[1]+2\theta_2
\end{pmatrix}
}_{c_0}\,.
\end{multline}

Notice that the projector to the symmetric irrep
\begin{tikzpicture}[scale=.15]
\draw (0,0) |-
(3,1) |- cycle (1,0) -- (1,1) (2,0) -- (2,1);  
\end{tikzpicture} is $\left(
  \begin{smallmatrix}
    1\\&P_s\\&&P_s\\&&&1
  \end{smallmatrix}
\right)$, whereas the projector to hook irreps $2\times \begin{tikzpicture}[scale=.15,baseline=-.1cm]
\draw (1,1) |- (0,-1) |- (2,1) |-(0,0);
\end{tikzpicture}$
is $\left(
  \begin{smallmatrix}
    0\\&P_h\\&&P_h\\&&&0
  \end{smallmatrix}
\right)$. Also notice that $c_0 P_s=0$, \ie $c_0$ only affects the hook irreps \begin{tikzpicture}[scale=.15,baseline=-.1cm]
\draw (1,1) |- (0,-1) |- (2,1) |-(0,0);
\end{tikzpicture}.

In addition to these notations, use
$\sse_1=\inhom[1]+\inhom[2]+\inhom[3]$,
$\sse_2=\inhom[1]\inhom[2]+\inhom[1]\inhom[3]+\inhom[2]\inhom[3]$, 
$\sse_3=\inhom[1]\inhom[2]\inhom[3]$,
and get 
\begin{multline}\label{eq:Q2}
  Q_2=\frac{-6 u^4+8 u^3 \sse_1+u^2\left(3\hbar^2-12 \sse_2\right)+u(-2\hbar^2 \sse_1+24 \sse_3)}{24
         \hbar}
\left(
  \begin{smallmatrix}
    1\\&P_s\\&&P_s\\&&&1
  \end{smallmatrix}
\right)\\
+\left(-\frac {u^3}{2\hbar}+\frac{\hbar \sse_1}{12}-\frac{\sse_3}{\hbar}\right)
 \left(
  \begin{smallmatrix}
    0\\&P_h\\&&P_h\\&&&0
  \end{smallmatrix}
\right)+\left(\frac{u^2}{4\hbar}-\frac{\hbar}{12}\right) \left(
  \begin{smallmatrix}
    0\\&c_0\\&&c_0\\&&&0
  \end{smallmatrix}
\right)\,,
\end{multline}
\begin{multline}\label{eq:Q01}
  \mathbb Q_{0,1}=\left(3\hbar u^2-2\hbar \sse_1\,u+\hbar \sse_2+\frac{\hbar^3}4\right)\left(
  \begin{smallmatrix}
    1\\&P_s\\&&P_s\\&&&1
  \end{smallmatrix}
\right)\\
+3\hbar\, u \left(
  \begin{smallmatrix}
    0\\&P_h\\&&P_h\\&&&0
  \end{smallmatrix}
\right)-\frac \hbar 2 \left(
  \begin{smallmatrix}
    0\\&c_0\\&&c_0\\&&&0
  \end{smallmatrix}
\right)\,.
\end{multline}

\paragraph{Restriction to the hook irreps}
If we restrict to the subspace $2\times \begin{tikzpicture}[scale=.15,baseline=-.1cm]
\draw (1,1) |- (0,-1) |- (2,1) |-(0,0);
\end{tikzpicture}$, we obtain $2\times 2$ matrices written for
instance in the basis~\footnote{This basis is an orthogonal basis
  of the 2D subspace, the expression of which is ``quite symmetric''.}
\begin{multline}\label{hookirrepbasis}
                 \left\{
                 \frac {\sqrt 3+3}6
                 \left|\uparrow\downarrow\downarrow\right>-\frac{\sqrt
                 3} 3
                 \left|\downarrow\uparrow\downarrow\right>+\frac{\sqrt
                 3 -3}6 \left|\downarrow\downarrow\uparrow\right>
               ,\frac {\sqrt 3-3}6
                 \left|\uparrow\downarrow\downarrow\right>-\frac{\sqrt
                 3} 3
                 \left|\downarrow\uparrow\downarrow\right>+\frac{\sqrt
                 3 +3}6 \left|\downarrow\downarrow\uparrow\right>
                 \right\}\,.
               \end{multline}
In this basis, $c_0$ becomes the $2\times 2$ matrix
\begin{equation}
\label{eq:cmatrix}
  \CCC:=c_0=
\begin{pmatrix}
    2\sse_1-\sqrt 3 (\inhom[1]-\theta_3)&\sse_1-3\theta_2+\sqrt 3
    \hbar\\ \sse_1-3\theta_2-\sqrt 3 \hbar&2\sse_1+\sqrt 3 (\inhom[1]-\theta_3)
  \end{pmatrix}
\end{equation}
and equations \eqref{eq:M1}, \eqref{eq:Q2} and \eqref{eq:Q01} become respectively
\begin{align}
  Q_1\rule[-.3cm]{.02cm}{.7cm}_{2\times \begin{tikzpicture}[scale=.15,baseline=-.1cm]
\draw (1,1) |- (0,-1) |- (2,1) |-(0,0);
\end{tikzpicture}}&=(u-2\frac {\sse_1}3)\mathbb I+\frac {\CCC}6
= \begin{pmatrix} u-\frac {\sse_1}3-\frac{\sqrt 3}6 (\inhom[1]-\theta_3) &\frac
  {\sse_1}6 -\frac {\theta_2}{2}+\frac {\sqrt 3}6 \hbar\\
\frac
  {\sse_1}6 -\frac {\theta_2}{2}-\frac {\sqrt 3}6 \hbar&
u-\frac {\sse_1}3+\frac{\sqrt 3}6 (\inhom[1]-\theta_3)
\end{pmatrix}\\
  Q_2\rule[-.3cm]{.02cm}{.7cm}_{2\times \begin{tikzpicture}[scale=.15,baseline=-.1cm]
\draw (1,1) |- (0,-1) |- (2,1) |-(0,0);
\end{tikzpicture}}&=\left(-\frac {u^3}{2\hbar}+\frac{\hbar \sse_1}{12}-\frac{\sse_3}{\hbar}\right)\mathbb I+\left(\frac{u^2}{4\hbar}-\frac{\hbar}{12}\right)\CCC
\\
\mathbb Q_{0,1}\rule[-.3cm]{.02cm}{.7cm}_{2\times \begin{tikzpicture}[scale=.15,baseline=-.1cm]
\draw (1,1) |- (0,-1) |- (2,1) |-(0,0);
\end{tikzpicture}}&=3\hbar u-\frac{\hbar}{ 2}{\CCC}=
3\hbar\begin{pmatrix}
     u-\frac{\sse_1}3+\frac {\sqrt 3}6(\inhom[1]-\theta_3)&-\frac{\sse_1}6+\frac {\theta_2} 2
     -\frac{\sqrt 3}6\hbar\\-\frac{\sse_1}6+\frac {\theta_2} 2
     +\frac{\sqrt 3}6\hbar&u-\frac{\sse_1}3-\frac {\sqrt 3}6(\inhom[1]-\theta_3)
\end{pmatrix}
\end{align}

\paragraph{Symmetrised modules}
Let us now compute these operators in the symmetrised Yangian representation. This results in presenting a $\inhom[]$-dependent change of basis such that all the matrix coefficients of $\CCC$, the only non-trivial operator of the Bethe algebra, are symmetric polynomials in $\inhom[\pa]$. We will explicitly compute this basis by using the proof of \ref{ch2}.

Let us first consider the case $\hbar=0$. The normalised Young symmetriser for $\SYT=\raisebox{-0.5em}{{\scriptsize\young(12,3)}}$ is given by
\be
S_\SYT=\frac 1 3 R_\SYT C_\SYT=\frac{1}{3}(1+(2~1~3)-(3~2~1)-(3~1~2))\,.
\ee
We now have to compute $S_\SYT\cdot\CC[\inhom[]]$ and moreover to find a $\CC[\chi]$-basis for it. Start by picking a $\CC[\chi]$-basis of $\CC[\inhom[]]${},  for example the Schubert polynomials (in the case of $\SG_3$ they are given by $\{1, \inhom[1],\inhom[1]+\inhom[2],\inhom[1]^2,\inhom[1]\inhom[2],\inhom[1]^2\inhom[2]\}$)~\footnote{It is also a good example to check equation \eqref{ch1}. Indeed $\frac{\ch(\CC[\inhom[]])}{\ch(\CC[\sse])}=\frac{(1-t)(1-t^2)(1-t^3)}{(1-t)^3}=1+2t+2t^2+t^3$ which is exactly the character of the Schubert basis.}. Since $S_\SYT$ commutes with multiplication by symmetric polynomials we just have to compute its action on Schubert polynomials. Eliminating obvious redundancies we obtain only two $\CC[\sse]$-independent basis elements
\begin{equation}
\eta_1:=\inhom[1]+\inhom[2]-2\inhom[3]\qquad\eta_2:=2\inhom[1]\inhom[2]-\inhom[1]\inhom[3]-\inhom[2]\inhom[3]\,.
\end{equation}
At this stage we can already check the character formula \eqref{ch2}. Indeed
\be
\frac{\ch(\mathcal{V}^{\SG+}_{(2,1)})}{\ch(\CC[\sse])}=\frac{t(1-t)(1-t^2)(1-t^3)}{(1-t)^2(1-t^3)}=t(1+t)=t^{\deg\eta_1}+t^{\deg\eta_2}\,.
\ee
To obtain a $\CC[\sse]$-basis of $\mathcal{V}^{\SG+}_{(2,1)}$ it remains to compute $p(\VTo\otimes\eta_1)$ and $p(\VTo\otimes\eta_2)$ which can be done straightforwardly.

Now assume $\hbar\neq0$. This case is more complicated because now we have to take $\hbar$ corrections into account. In particular now $p(\VTo\otimes\eta_1), p(\VTo\otimes\eta_2)\notin\mathcal{V}^{\SG+}_{(2,1)}$. 
Nevertheless $p(\VTo\otimes\eta_1), p(\VTo\otimes\eta_2)\in\mathcal{V}^{\SG}_{(2,1)}$ and we have to correct them by some vectors of lower degree such that they belong to $\mathcal{V}^{\SG+}_{(2,1)}$. Since $\deg p(\VTo\otimes\eta_1)=\deg\eta_1=1$ it can only be corrected by a vector of degree zero. There is only one such vector in $\mathcal{V}^{\SG}_{(2,1)}$: the totally symmetric combination $\ket{\uparrow\downarrow\downarrow}+\ket{\downarrow\uparrow\downarrow}+\ket{\downarrow\downarrow\uparrow}$. Its coefficient can be uniquely fixed by requiring that the corrected vector is highest-weight. By a similar argument we can compute the $\hbar$ corrections to $p(\VTo\otimes\eta_2)$. In the end we obtain the following $\CC[\sse]$-basis of $\mathcal{V}^{\SG+}_{(2,1)}$
\be
\begin{aligned}
\xi_1:= & \frac{1}{6}(-2\inhom[1]+\inhom[2]+\inhom[3]-3\hbar)\ket{\uparrow\downarrow\downarrow}\\
&~~+\frac{1}{6}(\inhom[1]-2\inhom[2]+\inhom[3])\ket{\downarrow\uparrow\downarrow}\\
&~~~~+\frac{1}{6}(\inhom[1]+\inhom[2]-2\inhom[3]+3\hbar)\ket{\downarrow\downarrow\uparrow}\\
\xi_2:= & \frac{1}{18}(-3\inhom[1]\inhom[2]-3\inhom[1]\inhom[3]+6\inhom[2]\inhom[3]-\hbar(\inhom[1]+4\inhom[2]+4\inhom[3]))\ket{\uparrow\downarrow\downarrow}\\
&~~+\frac{1}{18}(-3\inhom[1]\inhom[2]+6\inhom[1]\inhom[3]-3\inhom[2]\inhom[3]-\hbar(4\inhom[1]+\inhom[2]-5\inhom[3])-3\hbar^2)\ket{\downarrow\uparrow\downarrow}\\
&~~~~+\frac{1}{18}(6\inhom[1]\inhom[2]-3\inhom[1]\inhom[3]-3\inhom[2]\inhom[3]+\hbar(5\inhom[1]+5\inhom[2]-\inhom[3])+3\hbar^2)\ket{\downarrow\downarrow\uparrow}
\end{aligned}
\ee
that is, $\mathcal{V}^{\SG+}_{(2,1)}=\CC[\sse]\xi_1\oplus\CC[\sse]\xi_2$. Changing $\CCC$ to this basis we finally obtain its symmetrised representative
\be
\CCC^\SG:=
\begin{pmatrix}
-\hbar & -2\sse_2-\frac{2}{3}\hbar(\sse_1+\hbar)\\
6 & 4\sse_1+\hbar
\end{pmatrix}\,.
\ee

A few remarks are in order. First, note that $\CCC$ is homogeneous of degree $1$ (if we consider $\deg\hbar=1$) whereas $\CCC^\SG$ is non homogeneous of degree $2$. This has to do with the fact that $\xi_1$ and $\xi_2$ are $\inhom[]$-dependent and of different degrees. Second, one can check that $\CCC$ and $\CCC^\SG$ have the same characteristic polynomial and therefore the same spectrum. However there is a crucial difference: if we set $\hbar=0$ and take $\binhom[1]=\binhom[2]=\binhom[3]$ (a potentially ``bad'' point by Theorem~\ref{thm:iso}), $\CCC$ becomes proportional to the identity whereas $\CCC^\SG$ does not, so they cannot be related by a change of basis. In particular the symmetrised Bethe algebra $\BAL(\bsse)$ is still maximal, whereas the evaluated Bethe algebra $\BAL(\binhom[])$ is not. Again this has to do with the fact that the basis $(\xi_1,\xi_2)$ is $\inhom[]$-dependent: at this particular point it becomes degenerate \emph{as a basis of the physical vector space}, but not \emph{as a basis of the quotient} $\mathcal{V}^{\SG+}_{(2,1)}(\bsse)$. Actually, the determinant of the matrix of change of basis between \eqref{hookirrepbasis} and $(\xi_1,\xi_2)$ is given by $\frac{-1}{4\sqrt{3}}(\inhom[1]-\inhom[2]+\hbar)(\inhom[1]-\inhom[3]+\hbar)(\inhom[2]-\inhom[3]+\hbar)$ and  it equates to zero precisely at the ``bad'' points of Theorem~\ref{thm:iso}. Note however that even at most ``bad'' points $\BAL(\binhom[])$ and $\BAL(\bsse)$ are still isomorphic and maximal.

This example shows that a $\CC[\sse]$-basis of the symmetrised Yangian representation is quite difficult to compute. As long as the hypothesis of Theorem~\ref{thm:iso} is satisfied, the traditional physical frame is perfectly equivalent to the symmetrised one and can be used without trouble for all practical applications.

\section{Q-operators belong to the Bethe algebra}
\label{sec:q-belongs-bethe}

In order to show that the Q functions/operators belong to the Bethe algebra we will show how to find them from $\wT$ functions/operators. When $a=1$, the contraction with the Levi-Civita tensor in the Wronskian expression (\ref{eq:WronskianwT}) reduces to $(\gm-1)!\sum\limits_b (-1)^bQ_b^{[\gm-\gn+s]}Q_{\bar b}^{[-s]}$, hence we
can express the $t$-dependent sum (where $t$ is a free parameter)
\begin{equation}
  \label{eq:DefS}
  S(t):= \sum_b \left(\sum_{s\ge 1} Q_{\bar b}^{[-\gm+\gn-2s]} t^s\right)  Q_{b\vphantom{\bar b}}^{\vphantom{[]}}
\end{equation}
as an infinite linear combination of the $\wT_{(s^1)}$: this linear combination looks like $\sum\limits_{s\ge 1} \wT_{(s^1)}^{[-\gm+\gn-s]} t^s$, up to the first terms (when $s<a-\gm+\gn$) and up to factors that have no impact on the present argument and would make expressions extremely bulky. These are the factors $Q_{\bar\emptyset|\bar\emptyset}$, $\Ber G$ (which we set to $1$), $(\gm-1)!$  and the proportionality~\footnote{This proportionality factor is a supersymmetric version of a Vandermonde determinant of the eigenvalues $z_\alpha$, as can be found from requiring that $q_\emptyset$, $q_a$ and $q_i$ are monic.} factor denoted by the symbol $\propto$ in (\ref{eq:WronskianwT}).

This infinite sum (where $s$ runs from $0$ to $+\infty$) converges
in the disk $|t|<\min \left|\frac {1} {x_b}\right|$ and is then analytically continued to $t\in\mathbb C\setminus\left\{\frac {1} {x_b}\middle|1\le b\le \gm\right\}$. Indeed, if we denote $Q_{\bar b}=
\left(\frac {1} {x_b}\right)^{u/\hbar} \sum_{k=0}^{M_{\bar b}}c_{{\bar b}}^{(k)}u^{k}
$ then for $|t|<\min \left|\frac {1} {x_b}\right|$ we have:
\begin{multline}
  \label{eq:Sumq}
  \sum_{s\ge 1} Q_{\bar b}^{[-\gm+\gn-2s]} t^s=\\\left(\frac {1} {x_b}\right)^{\frac{-\gm+\gn}{2}}\sum_{j=0}^{M_{\bar b}}
                                                (-1)^j  \left(\sum_{k=j}^{M_{\bar b}} {\binom{k}{j}} (u^{[-\gm+\gn]})^{k-j} c^{(k)}_{\bar b} \right)\underbrace{\sum_{s\ge 1} s^j \left(t \,x_b\right)^s}_{\frac{
\sum_{k=0}^{j-1} A(j,k)\left(t\, x_b\right)^{k+1}
}{\left(1-t\, x_b\right)^{j+1}}\,,}
\end{multline}
where the combinatorial factors $A(j,k)$ are the positive integers known as the Eulerian numbers, and where the sum $\sum\limits_{k=0}^{j-1} A(j,k)\left(t\, x_b\right)^{k+1}$ should be replaced by $1$ in the ill-defined case $j=0$.

If the eigenvalues are $x_b$ are pairwise distinct we deduce that
\begin{equation}
  \label{eq:defQfromGenSeries}
  Q_b\propto \lim_{t\to \frac {1}{x_b}}\left(1-t\, x_b\right)^{M_{\bar b}+1} S(t)\,.
\end{equation}

This expression allows one to conclude that, at the level of representations, it belongs to the Bethe algebra. Indeed, the Bethe Algebra forms a linear subspace of the space of operators on the Hilbert space, which is finite-dimensional. It is hence topologically closed, so that the sum $S(t)$ belongs to the Bethe algebra, not only when $|t|<\min \left|\frac {1} {x_b}\right|$ but even for arbitrary $t$ by analytic continuation. Then, by taking the limit \eqref{eq:defQfromGenSeries} $Q_b$ belongs to the Bethe Algebra.

In addition to the Q-operators $Q_b$ ($1\le b\le \gm$), the same approach also allows producing the operators $Q_i$ ($\hat 1\le i\le \hat \gn$) by focusing on a sum of the form $\sum_{a\ge 0} \wT_{(1^a)}^{[-\gm+\gn+a]}t^a$. Hence it allows expressing all Q-operators using (\ref{sport}) and (\ref{eq:qqforms}).

In \cite{Kazakov:2010iu}, explicit computations of such infinite sums and of their limit were performed combinatorically at the level of the representation $ev_{{\theta}}$ (for twist with pairwise-distinct eigenvues, as in the above discussion). This explicit construction of the Q-operators shows that their matrix coefficients are indeed polynomial functions of $u$ and of the inhomogeneities $\inhom$, and that they are rational functions of the twist eigenvalues $z_\sA$. It also shows that the degree $M_{A|I}$ of $q_{A|I}$ is indeed given by \eqref{degree}.

For comparison purposes, we note that the infinite sum which was computed in \cite{Kazakov:2010iu} is actually of the form $\sum_{s\ge 1} \wT_{(s^1)}^{[+s]} t^s$, by contrast with the opposite shifts in the above discussion. Consequently the explicit combinatorial description gives an expression of $Q_{\bar b}$ instead of $Q_b$, and $Q_b$ was extracted after a few more steps -- after $\gm-1$ successive limits-- and the whole Q-system is expressed explicitly.

\section{Details about structural study of Bethe equations}
\label{app:BetheDetails}
In this section we will adopt the analytic point of view of Section~\ref{sec:basicproperties} on inhomogeneities. Namely we will think them as complex numbers that we are going to vary. Then $c_{\pa}$ -- the coefficients of (twisted) Baxter polynomials -- turn out to be algebraic functions of inhomogeneities. The main purpose of this section is to prove properness of WBE -- that all solutions $c_{\pa}$ are bounded if $\inhom$ (and hence $\sepa$) are bounded. A natural consequence of this analysis will be the behaviour of $c_{\pa}$ when $\inhom$ tend to infinity (in the twist-less case) which we analyse in \ref{sec:enum-solut-with} to provide the necessary technical results for Section~\ref{sec:twist-less-case}.

\subsection{Properness  - twisted case}\label{sec:Finiteness}

The feature that ensures that $c_{\alpha}$ are bounded at finite $\theta_{\alpha}$ is the fact that $c_\alpha$ are coefficients of polynomials in $u$ and the quantization condition \eqref{mastereq} is an equation on these polynomials.

Assume that there is a point $\bsse \in \Sedom$ by approaching which some of $c_{\alpha}$ diverge (become unbounded). If $c_{\alpha}$ is a coefficient of a twisted polynomial $Q(u)$ then divergence of $c_{\alpha}$ implies divergence of some of the roots of $Q(u)$. Factorise $Q(u)$ in the form $Q=Q^{\gg}Q^{\lesssim}$, where $Q^{\gg}$ is a polynomial containing all diverging roots, and $Q^{\lesssim}$ is the function containing the twist prefactor and all finite roots.

Consider first the equation $Q_{a|i}^+-Q_{a|i}^-=Q_{a|\es}Q_{\es|i}$. One can rewrite $Q_{a|i}^+-Q_{a|i}^-=Q^{\gg}_{a|i}((Q^{\lesssim}_{a|i})^+-(Q^{\lesssim}_{a|i})^-)+H$, where $H$ is the function that ensures equality. By taking a point $\sse$ sufficiently close to the point $\bsse$, one can ensure that $H$ is small in the following sense: there exists $R$ such that all roots of $(Q^{\lesssim}_{a|i})^+-(Q^{\lesssim}_{a|i})^-$ lie inside the circle $|u|=R$, all roots of $Q^{\gg}$ lie outside it, and that the absolute value of $H$ is smaller than the absolute value of $Q^{\gg}_{a|i}((Q^{\lesssim}_{a|i})^+-(Q^{\lesssim}_{a|i})^-)$ when $|u|=R$. Then by Rouch\'e's theorem, the number of zeros of $Q_{a|i}^+-Q_{a|i}^-$ inside and outside of the circle is the same as that of $Q^{\gg}_{a|i}((Q^{\lesssim}_{a|i})^+-(Q^{\lesssim}_{a|i})^-)$. Hence existence of large zeros of $Q_{a|i}$ imply existence of large zeros, in the same amount, in the product $Q_{a|\es} Q_{\es|i}$. Their distribution between $Q_{a|\es}$ and $Q_{\es|i}$ depends on the solution we consider.

Consider now the WBE \eqref{mastereq} and recall that SW is explicitly the determinant \eqref{susyw}. Applying the same logic, we write  
\be\label{eq:D1}
\SW(C_{\Lambda})=\prod_{a=1}^{\gm} Q_{a|\es}^{\gg}\prod_{i=1}^{\gn}Q_{\es|i}^{\gg}\,\SW(Q^{\lesssim})+H\,,
\ee 
and then conclude using Rouché's theorem that $Q_{\theta}=\SW(Q)$ has large zeros, \ie  the point $\bsse$ cannot have $\bsepa$ all finite. For the argument to work, one needs to ensure that $\SW(Q^{\lesssim})$ is not vanishing but it is straightforward as the presence of twist prefactors $z_{\alpha}^{u/\hbar}$ ensures that already the leading-$u$ term in $\SW(Q^{\lesssim})$ is non-vanishing.

\subsection{Properness - twist-less case}
\label{sec:deta-constr-mapp}
To study the twist-less case, we will focus on the bosonised parameterisation of the Q-system on a Young diagram \eqref{eq:reconsQ}. In particular, one has
\be
\label{eq:QB00}
Q_{\theta}=\wQ_{0,0}\propto W(B_1,\ldots,B_{\gm})\,,
\ee
where $\gm=h_\LD$. Equation \eqref{eq:QB00} contains in principle the full information since the $\gls{\gm|0}$ Q-system is a possible way to parameterise the Bethe algebra $\BAL$.

To prove properness, we would like to use an argument similar to that of \eqref{eq:D1}, however cancellations in the Wronskian determinant make things more subtle.
\begin{example}\label{exmp:72}
Take $B_1=u$, $B_2=(u-\Lambda)(u+1)$, $B_3=(u-\Lambda)^3$. Let $\Lambda\to\infty$ if $\sse\to\bsse$. Formally there are four divergent roots when $\Lambda\to\infty$. However $W(B_1,B_2,B_3)=u^3+u(3\Lambda-\hbar^2)-\Lambda^2(\Lambda+3)\,$ which has three divergent roots.
\end{example}
The issue in the example comes (at least) from the fact that in the decomposition $Q=Q^{\gg}Q^{\lesssim}$, $B_1^{\lesssim}$ and $B_2^{\lesssim}$ are polynomials of the same degree (equal to one). 

To continue, we do a couple of formalisations.

\paragraph{Parametric factorisation} Let $\Lambda$ be a parameter, and we intend to consider the $\Lambda\to\infty$ behaviour of Q-functions that are algebraic functions of $\Lambda$. Define a scale function $S=\Lambda^{\beta}$ for some real $\beta$. We say that  $Q=Q^{\gg S}Q^{\lesssim S}$  is the parametric factorisation of the polynomial $Q$ at scale $S$ if all roots of the monic polynomial $Q^{\gg S}$ are much larger than $S$, and all roots of the monic polynomial $Q^{\lesssim S}$ are comparable to or smaller than $S$. More precisely, for each $\bu$ that satisfies $Q^{\gg S}(\bu)=0$ one has $\lim\limits_{\Lambda\to\infty} S/\bu=0$, and for each $\bu$ that satisfies $Q^{\lesssim S}(\bu)=0$ the $\Lambda\to\infty$ limit of $\bu/S$ is finite.

Then the argument around \eqref{eq:D1} can be formalised by the following lemma.
\begin{lemma}
\label{thm:D1}
For some polynomials $Q_1,\ldots,Q_a$, let $Q_{12\ldots a}=W(Q_1,\ldots,Q_a)$ and let $Q=Q^{\gg S}Q^{\lesssim S}$ be the parametric factorisation at scale $S$.

If degrees $\deg Q_1^{{\lesssim} S}$, $\deg Q_2^{{\lesssim} S}$, $\deg Q_3^{{\lesssim} S}$, $\dots$ are  pairwise distinct, and degrees $\deg Q_1$, $\deg Q_2$, $\deg Q_3$, $\dots$ are also pairwise distinct then 
\be
\deg Q^{\gg S}_{12\ldots a}=\sum_{a'=1}^{a} \deg Q^{\gg S}_{a'}\,.
\ee
\end{lemma}
\begin{proof}
  Perform an equivalent of decomposition \eqref{eq:D1} and apply Rouché's theorem. An important ingredient is that $W(u^{\deg Q_1},\dots, u^{\deg Q_a})$ (resp. $W(u^{\deg Q_1^{{\lesssim} S}},\dots, u^{\deg Q_a^{{\lesssim} S}})$) is not zero and hence provides the term of the highest degree of $W(Q_1,\ldots, Q_a)$ (resp. $W(Q_1^{{\lesssim} S},\ldots, Q_a^{{\lesssim} S})$) which is why the restriction on the degrees is imposed.
\end{proof}
Now we recall that not all coefficients of $B_a$ bear physical information as they are subject to the symmetry transformation  \eqref{eq:225}. We shall benefit from \eqref{eq:225} to ensure that a parametric factorisation of $B_a$ satisfies the conditions of the above lemma.
\begin{lemma}
  \label{lemma:rotation}
Let $B_1,\ldots,B_\gm$ be monic polynomials with $\deg B_1<\ldots<\deg B_\gm$. For any scale $S$, one can find a ``rotation''
  \begin{align}
    \label{eq:rotation}
    B_a&\to B_a+\sum_{b<a}h_{ab} B_b\,,
  \end{align}
  where the $h_{ab}$'s are complex-valued functions~\footnote{\label{fn:lemmacoefs}More
    accurately, they are algebraic functions of the inhomogeneities
    $\inhom[\pa]$ whose values depend on $\Lambda$.} of $\Lambda$,  such that, after the rotation, the degrees of ${B}_{1}^{{\lesssim} S}$,
$B_2^{{\lesssim} S}$, $\ldots$, ${B}_{\gm}^{{\lesssim} S}$ are  pairwise distinct.
\end{lemma}
We note that the proof below is constructive and it provides an algorithm to find $h_{ab}$ explicitly.
\begin{proof}
Without loss of generality one can set $S=1$ in which case we denote the parametric factorisation as $Q=Q^{\gg}Q^{\lesssim}$. Indeed, we can always perform the rescaling $u\to u\,S$.

In the labelling of polynomials $B_a=u^{\lambda_a+\gm-a}+\ldots+b_k^{(a)} u^k+\ldots + b_1^{(a)} u+b_0^{(a)}$, consider all $b_{k'}^{(a)}$ that have the largest exponent when $\Lambda\to\infty$ and choose $b_{(k)}^{(a)}$ with the largest $k$ among them. For instance, in $u^3+\Lambda u^2+\Lambda^3 u+2\Lambda^3$, it is $b_{(1)}=\Lambda^3$. Then $\deg B_a^{\lesssim}=k$.

If there exist such $a,b$, $b<a$ that $\deg B_a^{\lesssim}=\deg B_b^{\lesssim}=k$ then perform the transformation $B_a\to B_a-\frac{b_{(k)}^{(a)}}{b_{(k)}^{(b)}}\,B_b$. This transformation will affect the parametric factorisation of $B_a$. Two things can happen. First, $\deg B_a^{\lesssim}$ becomes smaller. Second, all terms with the largest exponent are cancelled out from $B_a$ in which case one gets a new (smaller) largest exponent and a new value for $\deg B_a^{\lesssim}$ (in principle arbitrarily large, only bounded by the degree of $B_a$).

We repeat recursively the procedure of comparison between all available pairs of $a,b$ and terminate when $\deg B_a^{\lesssim}$ become pairwise distinct. The recursion will terminate in a finite number of steps and produce a meaningful result for the following reasons: there are finitely many polynomials of finite degree to operate with, the maximal exponents can only  decrease in the procedure and they are bounded by zero from below, and $B_a$ cannot vanish entirely as $\deg B_a$ are pairwise distinct and so the leading monomial is never affected by the performed transformations.
\end{proof}
\begin{example}
For $S=1$ and the system in the Example on page~\pageref{exmp:72}, the rotation is done as follows. First, transformation $B_2\to B_2+\Lambda B_1=u^2+u-\Lambda$ drops the degree of $B_2^{\lesssim}$ to zero. Now both degrees of $B_2^{\lesssim}$ and $B_3^{\lesssim}$ are zero. We perform transformation $B_3\to B_3-\Lambda^2 B_2$. This drops the maximal exponent in $B_3$ from three to two, and $\deg B_3^{\lesssim}$ computed with respect to the new maximal exponent is two. Now all $\deg B_a^{\lesssim}$ are pairwise distinct. In summary, we get the rotated values $B_1=u,B_2=u^2+u-\Lambda,B_3=u^3-\Lambda(\Lambda+3)\,u^2+2\Lambda^2\,u$. $B_2$ has two divergent roots, $B_3$ has one divergent root, $W(B_1,B_2,B_3)$ remains unchanged by the performed rotation and it has three divergent roots.
\end{example}
Now we are ready to prove properness as declared on page~\pageref{par:prop}. Assume that there is a finite point $\bsse \in \Sedom$ by approaching which some coefficients of $B_a$ diverge. We follow some path parameterised by $\Lambda$, and $\Lambda\to\infty$ corresponds to the approach of the point. Choose $S=1$ and perform the transformation \eqref{eq:rotation} to get $\deg B_a^{\lesssim}$ pairwise distinct. If there are still divergent coefficients after this transformation, we get $\deg Q_{\theta}^{\gg}>0$ by Lemma~\ref{thm:D1} and \eqref{eq:QB00} and hence reach a contradiction. Thus all $B_a$ have finite coefficients. Compute $\wQ_{a,s}$ following \eqref{eq:reconsQ} and use the procedure in the proof of Lemma~\ref{thm:pollemma} to compute the set $C_{\Lambda}$ introduced after \eqref{eq:225}. $c_{\pa}$ appearing in this set are hence non-divergent when we approach $\bsse$ which is the properness in the sense of Section~\ref{sec:C}.

\subsection{Labelling solutions with standard Young tableaux --  technical details}
\label{sec:enum-solut-with}
\subsubsection*{All solutions approach \eqref{eq:wQlead0}}
The key step is to justify the formula  (\ref{eq:QSplit}). Consider the situation when $\theta_L=\Lambda$, $\Lambda\to\infty$, and all other $\inhom$ are finite.
In any scaling $S=\Lambda^{\beta}$ with $0<\beta<1$, we rotate $B_a$ to a frame where $\deg B_a^{\lesssim}$ are pair-wise distinct. By Lemma~\ref{thm:D1}, there is precisely one $a=a_0$ for which precisely one root of $B_{a_0}$ diverges, and all roots of $B_{a\neq a_0}$ stay finite.

Recall that $\deg B_a=\lambda_a+\gm-a$ and $\deg B_i=d_i$. The assignment rules are explained in Figure~\ref{fig:boso}. Since $B_a$ are in the frame with pair-wise distinct $\deg B_a^{\lesssim}$, if $a_0\neq\gm$, it must be that $\lambda_{a_0}>\lambda_{a_0+1}$ (equality is impossible). Then, there exists $s_0$ such that $d_{s_0}=\deg B_{a_0}-1$ and so the box $(a_0,s_0)$ is a corner box of the Young diagram.

Finally, we consider \eqref{eq:reconsQ} to decide which $\wQ_{a,s}$ have a divergent root. If $a>a_0$ then $\wQ_{a,s}$ does not depend on $B_{a_0}$ and hence has no  divergent roots. If $s>s_0$ then all polynomials of degree from $0$ to $\deg B_{a_0}-1$ appear in the Wronskian determinant. Then the polynomial structure of $B_{a_0}$ is irrelevant, we can replace it with the leading monomial and so $\wQ_{a,s}$ cannot have divergent roots. Finally if $a\leq a_0$ and $s\leq s_0$ then $B_{a_0}$ is present in the Wronskian and there is no polynomial of degree $\deg B_{a_0}-1$ in the Wronskian, so the conditions of the Lemma~\ref{thm:D1} are satisfied, hence $\wQ_{a,s}$ has precisely one divergent root.

Now we can deduce that roots scale exactly as $\Lambda$. Indeed, for $\beta'\in]\beta,1[$, the rotation of Lemma~\ref{lemma:rotation} may change but the $\wQ$ functions are invariant under this triangular rotation. If the value of $a_0$ changes at scale $\Lambda^{\beta'}$ compared to scale $\Lambda^{\beta}$, then there would be another corner-box $(a_0',s_0')\neq(a_0,s_0)$ such that at scale $\Lambda^{\beta'}$ the $\wQ$-functions with diverging roots are the nodes with $a\leq a_0'$ and $s\leq s_0'$. This is impossible because all $\wQ$ functions that diverge at scale $\Lambda^{\beta'}$ also diverge at scale $\Lambda^{\beta}$. Therefore $a_0$ is independent of $\beta\in]0,1[$, the number of diverging roots is thus also independent of $\beta$ and the diverging roots scale exactly as $\Lambda$. Finally, we get \eqref{eq:QSplit}, and also that $\tilde\wQ_{a,s}$ introduced alongside \eqref{eq:QSplit} is a Q-system on the Young diagram with the box $(a_0,s_0)$ removed.

\subsubsection*{Unambiguous continuation of the limiting solution \eqref{eq:wQlead0} for each SYT to finite inhomogeneities}
To discuss this question, we should not send inhomogeneities one after another to infinity, but do a more smooth realisation of the limit \eqref{scaling}. Namely, for $\alpha_1$, $\alpha_2$, $\dots$, $\alpha_L\in\mathbb C^*$ and  $\beta_L>\beta_{L-1}>\dots>\beta_1>0$, we parameterise $\inhom[1]=\alpha_1\,\Lambda^{\beta_1}$, $\inhom[2]=\alpha_2\,\Lambda^{\beta_2}$, $\dots$, $\inhom[L]=\alpha_L\,\Lambda^{\beta_L}$. All roots and all coefficients of all Q-polynomials are algebraic functions~\footnote{For a Q-system on a Young diagram, this is an immediate consequence of the QQ-relations. For a Q-system on a Hasse diagram, we should restrict symmetry transformations \eqref{eq:225} to algebraically depend on $\inhom$.} of the parameter $\Lambda$ and hence have a large-$\Lambda$ behavior of the form $\alpha\,\Lambda^\beta$ which allows applying scaling argumentation from previous sections.
  
If $\beta_\pa$ are spaced apart well, we can recover the same results as if inhomogeneities are sent to infinity one by one. But now, after we know that all solutions of the Q-system approach \eqref{eq:wQlead0}, a sharper judgement about possible $\beta_\pa$ can be made by observing that the leading order of the large $\Lambda$ expansion \eqref{eq:wQlead0} is solved by~\footnote{The computation reduces to the fact that
  $W\left(1,u,\dots,u^{k-1}, u^{k}(u-(k+1)\alpha\,\inhom[])\right)\propto (u-\alpha\, \inhom[])
  $.}
\begin{align}
  \label{eq:QparamL}
  \forall a&\le \gm,&
B_a\sim u^{\gm-a}\prod_{s=1}^{\IPart_{a}}\left(u-(\gm-a+s)\N{\SYT_{a,s}}{a-1}{s-1}\inhom[\SYT_{a,s}]\right)\,.
\end{align}
Let us now {\it parameterise} the Q-system by $\kappa_1,\ldots,\kappa_L$ as follows
\begin{align}
  \forall a&\le \gm,&
B_a=u^{\gm-a}\prod_{s=1}^{\IPart_{a}}\left(u-(\gm-a+s)\N{\SYT_{a,s}}{a-1}{s-1}\kappa_{\SYT_{a,s}}\Lambda^{\beta_{\SYT_{a,s}}}\right)\,.
\end{align}
Recall that $\inhom=\alpha_\pa\Lambda^{\beta_\pa}$. Then $Q_\theta=W(B_1,\ldots,B_\gm)$ realises a map from $\kappa_\pa$ to $\alpha_\pa$ which analytically depends on $1/\Lambda$ for $\beta_{\pa}$ being integers. When $1/\Lambda=0$ this map is simply an identity map with obviously non-zero Jacobian. Hence we can apply the analytic implicit function theorem to invert the map. By the theorem, for some neighbourhood of the point $1/\Lambda=0$, $\kappa_\pa$ are analytic functions of $\alpha_1,\ldots,\alpha_L$ and $1/\Lambda$ and hence each limiting solution \eqref{eq:wQlead0} can be continued to finite values of inhomogeneities.

\bibliographystyle{utphys}
\bibliography{References}
\end{document}